\pdfoutput=1
\documentclass[acmsmall, screen]{acmart}

\usepackage{booktabs}   %% For formal tables:
\usepackage[labelformat=simple]{subcaption} %% For complex figures with subfigures/subcaptions

\usepackage{peanutgallery}
\usepackage{bold-extra}

\usepackage{mathtools}
\usepackage[final]{listings}
\usepackage{proof}
\usepackage{etoolbox}
\usepackage{array}
\newcolumntype{L}{>$l<$}
\usepackage{stmaryrd}
\usepackage{tikz}
\usetikzlibrary{matrix,arrows,positioning,calc,fit,backgrounds}
\usepackage[shortlabels, inline]{enumitem}
\makeatletter
\newcommand{\myitem}[1][]{
  \protected@edef\@currentlabel{#1}%
\item[#1]
}
\makeatother

\newcommand{\moreless}[2]{#1}
\newcommand{\techreport}[1]{\moreless{#1}{}}

\newcommand{\ifspace}[1]{}

\newcommand{\m}[1]{\mathsf{#1}}

\newcommand{\defFunc}[2]{\newcommand{#1}{\mathsf{#2}}}

\newcommand{\Nat}{\mathbb{N}}
\newcommand{\Int}{\mathbb{Z}}

\newcommand{\impl}{\Rightarrow}
\renewcommand{\iff}{\Leftrightarrow}

\newcommand{\dom}{\operatorname{\mathsf{dom}}}

\newcommand{\pto}{\rightharpoonup}

\newcommand{\defeq}{\coloneqq}
\newcommand{\defiff}{\vcentcolon\iff}

\newcommand{\paren} [1] {\ensuremath{ \left( {#1} \right) }}

\newcommand{\bracket}[1]{\left[#1\right]}
\newcommand{\tuple}[1]{\ensuremath{\langle #1 \rangle}}

\newcommand{\set}[1]{\ensuremath{\left\{#1\right\}}}
\newcommand{\setcomp}[2]{\ensuremath{\left\{#1\;\middle|\;#2\right\}}}

\newcommand{\todo}[1]{}%{\textcolor{red}{TODO: #1}}
\newcommand{\ftodo}[1]{}%{\fcomment{TODO}{#1}}

\newcommand{\rSc}[1]{\S\ref{#1}}
\newcommand{\rF}[1]{Fig.~\ref{#1}}

\newcommand{\rL}[1]{Lemma~\ref{#1}}
\newcommand{\rT}[1]{Theorem~\ref{#1}}

\newcommand{\rE}[1]{Example~\ref{#1}}

\newcommand{\code}[1]{\textnormal{\texttt{#1}}}

\newcommand{\graph}{G}
\newcommand{\graphPlus}{\uplus}
\newcommand{\graphEmpty}{G_{\mathsf{e}}}

\newcommand{\nodeDom}{\mathcal{N}}

\newcommand{\inflowEmpty}{\inflow_{\mathsf{e}}}
\newcommand{\flowEmpty}{\flow_{\mathsf{e}}}

\newcommand{\intDom}{D}
\newcommand{\intLeq}{\sqsubseteq}%{\preccurlyeq}
\newcommand{\intJoin}{\sqcup}
\newcommand{\intMeet}{\sqcap}

\newcommand{\intPlus}{+}
\newcommand{\intBigPlus}{\sum}
\newcommand{\intMult}{\cdot}

\newcommand{\intEquiv}{\approx}
\newcommand{\intLessEquiv}{\precsim}
\newcommand{\intComp}{\oplus}

\newcommand{\intVar}{d}
\newcommand{\zerofun}{\mathbf{0}}

\newcommand{\inflowOfInt}[1]{{#1}^\inflow}
\newcommand{\inflowClassOfInt}[1]{{#1}^\inflowClass}
\newcommand{\nodelabelOfInt}[1]{{#1}^\nodelabel}
\newcommand{\flowmapOfInt}[1]{{#1}^\flow}

\newcommand{\mapEmpty}{\epsilon}

\newcommand{\interfaces}{\mathsf{FI}}
\newcommand{\interface}{I}
\defFunc{\interfaceFn}{int}
\newcommand{\interfaceEmpty}{\interface_e}

\newcommand{\logic}{\rgsep[\mathsf{\interfaces}]}

\newcommand{\inflowEquiv}{\sim}
\newcommand{\inflowClass}{\mathit{In}}
\newcommand{\inflowDom}{\mathsf{IN}}

\newcommand{\graphIn}{H}
\newcommand{\graphInDom}{\mathsf{FG}}
\newcommand{\graphInComp}{\bullet}%{\oplus}
\newcommand{\graphInEmpty}{H_{\mathsf{e}}}

\newcommand{\nodelabel}{a}
\newcommand{\nodelabelBot}{\nodelabel_e}
\newcommand{\nodelabelling}{\lambda}
\newcommand{\nodelabelDom}{A}
\newcommand{\nodelabelJoin}{\sqcup}

\newcommand{\nodelabelLeq}{\sqsubseteq}
\newcommand{\edgelabel}{\varepsilon}

\defFunc{\init}{init}
\defFunc{\capacity}{cap}
\defFunc{\netflow}{flow}
\defFunc{\flowmap}{flm}

\defFunc{\goodCondition}{good}

\defFunc{\pathCount}{pc}

\newcommand{\KS}{\mathsf{KS}}
\newcommand{\inflow}{\mathit{in}}
\defFunc{\edgeset}{es}
\defFunc{\contents}{c}
\defFunc{\lock}{lock}
\defFunc{\inset}{ins}
\defFunc{\keyset}{ks}
\defFunc{\outflow}{outf}
\defFunc{\pathset}{path}
\defFunc{\flowFn}{flow}
\newcommand{\flow}{f}
\defFunc{\inflows}{In}

\defFunc{\composition}{comp}
\defFunc{\projection}{proj}

\newcommand{\rgsep}{\mathsf{RGSep}}

\newcommand{\vars}{\mathsf{Var}}
\newcommand{\values}{\mathsf{Val}}
\newcommand{\addrs}{\mathsf{Addr}}
\newcommand{\states}{\mathsf{State}}
\newcommand{\cmds}{\mathsf{Com}}
\newcommand{\stateEmpty}{\sigma_e}

\newcommand{\heap}{h}
\newcommand{\heapEmpty}{\heap_e}
\newcommand{\interp}{i}

\newcommand{\nodeMap}{r}
\newcommand{\nodeMapEmpty}{\nodeMap_e}

\newcommand{\annot}[1]{\textcolor{blue}{\left\{\begin{aligned}#1\end{aligned}\right\}}}
\newcommand{\hoareTriple}[3]{\annot{#1} \; #2 \; \annot{#3}}

\newcommand{\denotation}[1]{\llbracket #1 \rrbracket}

\newcommand{\sepincl}{\mathrel{-\mkern-9mu*\mkern-7mu*}}
\newcommand{\magicwand}{\mathrel{-\mkern-6mu*}}
\newcommand{\septract}{\mathrel{-\mkern-2.5mu\circledast}}
\newcommand{\ite}{\mathsf{ITE}}

\newcommand{\listPred}{\mathsf{ls}}
\newcommand{\lseg}{\mathsf{lseg}}

\newcommand{\emp}{\mathit{emp}}
\newcommand{\true}{\mathit{true}}

\newcommand{\nullVal}{\mathit{null}}

\newcommand{\graphPred}{\mathsf{Gr}}

\newcommand{\nodePred}{\mathsf{N}}
\newcommand{\emptyFn}{\epsilon}

\newcommand{\goesto}{\rightarrowtail}

\newcommand{\abort}{\mathbf{abort}}
\newcommand{\skipCommand}{\text{\texttt{\textbf{skip}}}}
\newcommand{\syncCommand}{\text{\texttt{\textbf{sync}}}}
\newcommand{\markCommand}{\text{\texttt{\textbf{mark}}}}
\newcommand{\unmarkCommand}{\text{\texttt{\textbf{unmkark}}}}
\newcommand{\reduces}[2]{\xrightarrow[#2]{#1}}

\newcommand{\freeListHead}{\mathit{fh}}
\newcommand{\freeListTail}{\mathit{ft}}
\newcommand{\mainListHead}{\mathit{mh}}

\newcommand{\unmarked}{\diamondsuit}

\lstdefinelanguage{SPL}{
  morekeywords={acc, struct,if,else,returns,procedure,requires,ensures,:=,var,
    new,old,free,implicit,modifies,call,locals,assume,assert,choose,havoc,ghost,
    predicate,function,invariant,while, return,atomic, sync, mark, unmark},
  deletekeywords={union,int},
  numbers=left,
  xleftmargin=2em,
  escapeinside={@}{@},
  numberstyle=\tiny,
  basicstyle=\footnotesize\ttfamily,
  columns=flexible,
  morecomment=*[s][\color{darkgray}]{/*}{*/},
  morecomment=*[l][\color{darkgray}]{//},
  mathescape=true,
}
\tikzset{%
  array/.style={matrix of nodes,nodes={draw, minimum size=5mm, anchor=center},column sep=-\pgflinewidth, row sep=-\pgflinewidth, nodes in empty cells,anchor=center},
  ptr/.style={*->, shorten <=-(1.8pt+1.4\pgflinewidth)},
  edge/.style={->},
  fedge/.style={->, dashed},
  unode/.style={circle, draw=black, thick, minimum size=8mm},
  mnode/.style={circle, draw=black, thick, fill=gray!20, minimum size=8mm},
  gnode/.style={circle, draw=black, thick, minimum size=1cm},
  pnode/.style={circle, draw=black, thick, dotted, minimum size=1cm},
  inflow/.style={circle, fill=none, inner sep=0pt, minimum size=5mm, font=\normalsize},
  stackVar/.style={circle, fill=none, inner sep=0pt, minimum size=8mm, font=\normalsize},
  phantomNode/.style={circle, fill=none, inner sep=0pt, minimum size=0pt}
}

\moreless{
\acmArticle{1}
\settopmatter{printacmref=false}
}{
\acmJournal{PACMPL}
\acmYear{2018} \acmVolume{2} \acmNumber{POPL} \acmArticle{37} \acmMonth{1} \acmPrice{}\acmDOI{10.1145/3158125}
}

\moreless{
\setcopyright{none}
}{
\setcopyright{acmcopyright}
}

\begin{document}

%% Title information
\title[Go with the Flow: Compositional Abstractions for Concurrent
Data Structures\techreport{ (Extended)}]
{Go with the Flow: Compositional Abstractions for Concurrent Data Structures\techreport{ (Extended Version)}}

\titlenote{This work is funded in parts by NYU WIRELESS and by the \grantsponsor{GS100000001}{National Science Foundation}{http://dx.doi.org/10.13039/100000001} under grants~\grantnum{GS100000001}{MCB-1158273},~\grantnum{GS100000001}{IOS-1339362},~\grantnum{GS100000001}{MCB-1412232}, and~\grantnum{GS100000001}{CCF-1618059}.}

% \subtitle{Subtitle}                     %% \subtitle is optional
%\subtitlenote{with subtitle note}       %% \subtitlenote is optional;
                                        %% can be repeated if necessary;
                                        %% contents suppressed with 'anonymous'

%% Author information
%% Contents and number of authors suppressed with 'anonymous'.
%% Each author should be introduced by \author, followed by
%% \authornote (optional), \orcid (optional), \affiliation, and
%% \email.
%% An author may have multiple affiliations and/or emails; repeat the
%% appropriate command.
%% Many elements are not rendered, but should be provided for metadata
%% extraction tools.

%% Author with single affiliation.
\author{Siddharth Krishna}
%\authornote{with author1 note}          %% \authornote is optional;
                                        %% can be repeated if necessary
%\orcid{nnnn-nnnn-nnnn-nnnn}             %% \orcid is optional
\affiliation{
  %\position{Position1}
  %\department{Department1}              %% \department is recommended
  \institution{New York University}            %% \institution is required
  % \streetaddress{60 Fifth Avenue}
  % \city{New York}
  % \state{NY}
  % \postcode{10011}
  \country{USA}
}
\email{siddharth@cs.nyu.edu}          %% \email is recommended

\author{Dennis Shasha}
%\authornote{with author2 note}          %% \authornote is optional;
                                        %% can be repeated if necessary
%\orcid{nnnn-nnnn-nnnn-nnnn}             %% \orcid is optional
\affiliation{
  %\position{Position1}
  %\department{Department1}              %% \department is recommended
  \institution{New York University}            %% \institution is required
  % \streetaddress{60 Fifth Avenue}
  % \city{New York}
  % \state{NY}
  % \postcode{10011}
  \country{USA}
}
\email{shasha@cims.nyu.edu}         %% \email is recommended

\author{Thomas Wies}
%\authornote{with author2 note}          %% \authornote is optional;
                                        %% can be repeated if necessary
%\orcid{nnnn-nnnn-nnnn-nnnn}             %% \orcid is optional
\affiliation{
  %\position{Position1}
  %\department{Department1}              %% \department is recommended
  \institution{New York University}            %% \institution is required
  % \streetaddress{60 Fifth Avenue}
  % \city{New York}
  % \state{NY}
  % \postcode{10011}
  \country{USA}
}
\email{wies@cs.nyu.edu}         %% \email is recommended

%% Paper note
%% The \thanks command may be used to create a "paper note" ---
%% similar to a title note or an author note, but not explicitly
%% associated with a particular element.  It will appear immediately
%% above the permission/copyright statement.
%% \thanks is optional
                                        %% can be repeated if necesary
                                        %% contents suppressed with 'anonymous'

%% Abstract
%% Note: \begin{abstract}...\end{abstract} environment must come
%% before \maketitle command
\begin{abstract}
  Concurrent separation logics have helped to significantly simplify
  correctness proofs for concurrent data structures.
  However, a recurring problem in such proofs is that data
  structure abstractions that work well in the sequential setting are much
  harder to reason about in a concurrent setting due to complex
  sharing and overlays.
  %Another common
  %problem is that proof rules that have been established for
  %abstractions of one type of data structure do not easily carry over
  %to another, in particular, if the abstractions maintain constraints
  %on data. This imposes limits on proof automation.
  To solve this problem, we propose a novel approach to abstracting
  regions in the heap by encoding the data structure invariant into a
  local condition on each individual node. This condition may depend
  on a quantity associated with the node that is computed as a fixpoint over the
  entire heap graph. We refer to this quantity as a \emph{flow}. Flows
  can encode both structural properties of the heap (e.g. the
  reachable nodes from the root form a tree) as well as data
  invariants (e.g. sortedness). We then introduce the notion of a \emph{flow
    interface}, which expresses the relies and guarantees that a heap
  region imposes on its context to maintain the local flow invariant
  with respect to the global heap.
  Our main technical result is that this notion leads to a new
  semantic model of separation logic. In this model, flow interfaces
  provide a general abstraction mechanism for describing complex data
  structures. This abstraction mechanism admits proof rules that generalize over a wide variety of data structures.
  %These include rules that allow a heap region to be
  %split into arbitrary chunks which can be modified and recomposed to
  %form a new region, while maintaining the global data structure
  %invariant.
  To demonstrate the versatility of our approach, we show
  how to extend the logic RGSep with flow interfaces.
  We have used this new logic to prove linearizability and memory
  safety of nontrivial concurrent data structures. In particular, we
  obtain parametric linearizability proofs for concurrent dictionary
  algorithms that abstract from the details of the underlying data
  structure representation. These proofs cannot be easily expressed
  using the abstraction mechanisms provided by existing separation logics.
\end{abstract}

%% 2012 ACM Computing Classification System (CSS) concepts
%% Generate at 'http://dl.acm.org/ccs/ccs.cfm'.

 \begin{CCSXML}
<ccs2012>
<concept>
<concept_id>10003752.10003790.10002990</concept_id>
<concept_desc>Theory of computation~Logic and verification</concept_desc>
<concept_significance>500</concept_significance>
</concept>
<concept>
<concept_id>10003752.10003790.10011742</concept_id>
<concept_desc>Theory of computation~Separation logic</concept_desc>
<concept_significance>500</concept_significance>
</concept>
<concept>
<concept_id>10003752.10003809.10011778</concept_id>
<concept_desc>Theory of computation~Concurrent algorithms</concept_desc>
<concept_significance>500</concept_significance>
</concept>
<concept>
<concept_id>10003752.10003790.10003806</concept_id>
<concept_desc>Theory of computation~Programming logic</concept_desc>
<concept_significance>300</concept_significance>
</concept>
</ccs2012>
\end{CCSXML}

\ccsdesc[500]{Theory of computation~Logic and verification}
\ccsdesc[500]{Theory of computation~Separation logic}
\ccsdesc[500]{Theory of computation~Concurrent algorithms}
\ccsdesc[300]{Theory of computation~Programming logic}

%% End of generated code

%% Keywords
%% comma separated list
\keywords{memory safety, linearizability, flow
  interfaces, separation algebra}

%% \maketitle
%% Note: \maketitle command must come after title commands, author
%% commands, abstract environment, Computing Classification System
%% environment and commands, and keywords command.
\maketitle

\section{Introduction}

With the advent of concurrent separation logics (CSLs), we have
witnessed substantial inroads into solving the difficult problem of
concurrent data structure
verification~\cite{DBLP:conf/concur/OHearn04,
  Bornat:2005:PAS:1040305.1040327, DBLP:conf/concur/VafeiadisP07,
  viktor-thesis, DBLP:conf/ecoop/Dinsdale-YoungDGPV10,
  DBLP:conf/popl/Dinsdale-YoungBGPY13, DBLP:conf/ecoop/PintoDG14,
  DBLP:conf/esop/NanevskiLSD14, Jung:2015:IMI:2676726.2676980,
  Dodds:2016:VCS:2866613.2818638, DBLP:conf/esop/PintoDGS16,
  DBLP:conf/popl/GuKRSWWZG15}. CSLs
provide compositional proof rules that untangle the complex
interference between concurrent threads operating on shared memory
resources.
At the heart of these logics lies separation logic (SL)~\cite{SL1, SL2}.
The key ingredient of SL is the separating
conjunction operator, which allows the global heap memory to be split
into disjoint subheaps. An operation working on one of the subheaps
can then be reasoned about compositionally in isolation using so-called frame rules
that preserve the invariant of the rest of the heap. To support
reasoning about complex data structures (lists, trees, etc.),
SL is typically extended with predicates that are
defined inductively in terms of separating conjunction to express
invariants of unbounded heap regions. However, inductive
predicates are not a panacea when reasoning about
concurrent data structures.

One problem with inductive predicates is that the recursion scheme
must follow a traversal of the data structure in the heap that visits
every node exactly once. Such definitions are not well-suited for
describing data structures with unbounded sharing and
overlays of multiple data structures with separate roots whose
invariants have mutual
dependencies~\cite{Hobor:2013:RSD:2429069.2429131}. Both of these
features are prevalent in concurrent algorithms (examples include
% Bw-trees~\cite{DBLP:journals/debu/LevandoskiS13}
B-link trees~\cite{Lehman:1981:ELC:319628.319663} and non-blocking
lists with explicit memory management~\cite{DBLP:conf/wdag/Harris01}).

Another challenge is that proofs involving inductive predicates rely
on lemmas that show how the predicates compose, decompose, and
interact. For example, SL proofs of algorithms that manipulate linked
lists often use the inductive predicate $\lseg(x,y)$, which
denotes subheaps containing a list segment from $x$ to $y$. A common
lemma about $\lseg$ used in such proofs is that two disjoint list
segments that share an end and a start point compose to a larger list
segment (under certain side conditions). Unfortunately, these lemmas
do not easily generalize from one data structure to another since the
predicate definitions may follow different traversal patterns and
generally depend on the data structure's implementation details.
Hence, there is a vast literature describing techniques to derive such
lemmas automatically, either by expressing the predicates in decidable
fragments of SL~\cite{BerdineETAL04DecidableFragmentSeparationLogic,
  CooketALFragmentSepLog, DBLP:conf/pldi/PerezR11,
  DBLP:conf/atva/BouajjaniDES12, DBLP:conf/atva/IosifRV14,
  DBLP:conf/cav/PiskacWZ13, DBLP:conf/aplas/TatsutaLC16,
  DBLP:conf/vmcai/ReynoldsIS17, DBLP:conf/nfm/EneaLSV17} or by using
heuristics~\cite{DBLP:conf/cav/NguyenC08,
  DBLP:conf/cade/BrotherstonDP11, DBLP:conf/pldi/Chlipala11,
  DBLP:conf/pldi/PekQM14, DBLP:conf/atva/EneaSW15}. However, these
techniques can be brittle, in particular, when the predicate
definitions involve constraints on data.

What appears to be missing in existing SL variants is an abstraction
mechanism for heap regions that is agnostic to specific traversal
patterns, yet can still express inductive properties of the
represented data structure.
To address this shortcoming, we here take a radically different
approach to abstraction in separation logic. Instead of relying on
data-structure-specific inductive predicates, we introduce a new
abstraction mechanism that specifies inductive structural and data
properties uniformly but independently of each other,
without fixing a specific traversal strategy.

\iffalse
This abstraction mechanism supports composition
and decomposition rules that are agnostic to specific traversal
patterns and appears to be particularly suitable for reasoning about concurrent data
structures.
\fi

As a first step, we describe data structure invariants in a uniform
way by expressing them in terms of a condition local to each node of
the heap graph. However, this condition is allowed to depend
on a quantity associated with the node that is computed inductively over the entire
graph. We refer to this inductively defined quantity as a \emph{flow}.
%and the set of values over which the flow ranges as the \emph{flow
%domain}%
An example of a flow is the
function $\pathCount$ that maps each node $n$ to the number of paths from
the root of the data structure to $n$.
%Here the flow domain is the set
%of natural numbers extended with infinity (to account for cycles in
%the graph).
The property that a given heap graph is a tree can then be
expressed by the condition that for all nodes $n$,
$\pathCount(n)=1$. Flows may also depend on the data stored in the heap
graph. For instance, in a search data structure that implements a
dictionary of key/value pairs, we can define the \emph{inset
  flow}. The inset of a node $n$ is a set of
keys. Intuitively, a key $k$ is in the inset of $n$ if and only if a search for $k$ that starts
from the root of the data structure will have to traverse the node
$n$. The inset flow can be used to express the invariants of search
data structure algorithms in a way that abstracts from the concrete
data structure implementation (i.e., whether the concrete data
structure is a list, a tree, or some more complicated structure with
unbounded sharing)~\cite{dennis-keyset}.
\iffalse
Yet another example of a flow
is a type abstraction where the flow of a node $n$ captures the
constraints imposed on $n$'s type by other nodes pointing to $n$. Such
an abstraction can be used to maintain type invariants of the nodes
of a data structure in type-unsafe languages. {\em Dennis still (Oct 6) thinks this last example is not so clear.}

The abstraction of concrete data structures in terms of flows is
facilitated by mapping each heap graph to a \emph{flow graph}. The
flow graph abstracts from the details of how the nodes and edges of
the data structure are represented in memory. E.g., the flow graph no
longer distinguishes whether an edge is implemented by a pointer
dereference or an array look-up. The flow is then defined over the flow
graph rather than directly over the concrete heap graph. This allows
one to build flow-based abstractions that are agnostic to
implementation details of the underlying data structure.
\fi

We would like to reason compositionally about flows in SL using
separating conjunctions. A separating conjunction
is defined in terms of a composition operator
on the semantic models of SL that yields a so-called \emph{separation
  algebra}~\cite{DBLP:conf/lics/CalcagnoOY07}. For the standard heap
graph model of SL, composition is disjoint graph union:
$G = G_1 \uplus G_2$. To enable local reasoning about flows, the
composition operator needs to ensure that the flow on the resulting
graph $G$ maintains the local conditions with respect to the flows on
its constituents $G_1$ and $G_2$. However, this is typically not the
case for disjoint union. For instance, consider again our
example of the path-counting flow where the flow condition bounds the
path count of each node. Even if $G_1$ and $G_2$ individually satisfy
the bound on the path count, $G_1 \uplus G_2$ may not since the union
of the graphs can create cycles.

We therefore define a disjoint union operator on graphs that
additionally asserts that the two subgraphs respect their mutual
constraints on the flow of the composite graph. Our key technical
contribution is to identify a class of flows
% in terms of their algebraic properties
that give rise to a separation algebra with this
new composition operator. This class is closed under product (which
enables reasoning about data structure overlays) and subsumes many
natural examples of flows, including the ones given above.

To express abstract SL predicates that describe unbounded graph
regions and their flows, we introduce the notion of a \emph{flow
  interface}. A flow interface of a graph $G$ expresses the
constraints on $G$'s contexts that $G$ relies upon in order to satisfy
its internal flow conditions, as well as the
guarantees that $G$ provides its contexts so that they can satisfy
their flow conditions. In the example of path-counting flows used to
express ``treeness'', the rely of the flow interface specifies how
many paths $G$ expects to exist from the global root in the composite
graph to each of the roots in $G$ (there can be many roots in $G$
since it may consist of disconnected subtrees), and the guarantee
specifies how many paths there are between each pair of root and sink
node in $G$.

The algebraic properties of flows that guarantee that flow composition
is well-defined also give rise to generic proof rules for reasoning
about flow interfaces.  These include rules that allow a flow
interface to be split into arbitrary chunks which can be modified and
recomposed, enabling reasoning about data structure algorithms
that do not follow a fixed traversal strategy.

To demonstrate the usefulness of our new abstraction mechanism, we
have instantiated rely/guarantee separation logic
($\rgsep$)~\cite{DBLP:conf/concur/VafeiadisP07, viktor-thesis} with
flow interfaces (other CSL flavors can be extended in a similar
fashion). We have used the new logic to obtain simple correctness
proofs of intricate concurrent data structure algorithms such as the
Harris list with explicit memory
management~\cite{DBLP:conf/wdag/Harris01}, which cannot be easily
verified with existing logics. Moreover, we show how the new logic can
be used to formalize the edgeset framework for verifying concurrent
dictionary data structures~\cite{dennis-keyset}. We apply this
formalization to proving
linearizability~\cite{DBLP:journals/toplas/HerlihyW90} of an algorithm
template for concurrent dictionary implementations. The template
abstracts from the specifics of the underlying data structure
representation allowing it to be refined to diverse concrete
implementations (e.g. using linked lists, B+ trees, etc.). By using
flow interfaces, the correctness proof of the template can also
abstract from the concrete implementation details, enabling proof
reuse across the different refinements. We are not aware of any other
logic that provides an abstraction mechanism to support this style of
proof modularization at the level of data structure algorithms.

\techreport{
This is an extended version of a conference paper~\cite{flows-popl-paper}, containing the following additional material: \rSc{sec-logic} contains formal definitions of the syntactic shorthands we use in our logic, the formal semantics of our new ghost commands, and some additional lemmas for proving entailments; \rSc{sec-harris-simple-appendix} contains proof sketches of the full insert procedure and the delete procedure on the Harris list; \rSc{sec-btree} contains an instantiation of our dictionary framework to the B+ tree data structure; and \rSc{sec-stability-proof} contains an example stability proof, needed to prove safety under interference of other threads, in our logic.
}

%%% Local Variables:
%%% mode: latex
%%% TeX-master: "main"
%%% End:

\section{Motivating Example}
\label{sec-overview}

Our aim is to prove memory safety and functional properties of concurrent data structures using local reasoning.
We begin with a quick reminder of the standard way of specifying data structures in separation logic using inductive predicates.
Using the example of the Harris list~\cite{DBLP:conf/wdag/Harris01}, we see the limitations of the standard approach when dealing with concurrent data structures.

%\subsection{Limitations of Inductive Predicates in Separation Logic}
\label{sec-inductive-preds}

Separation logic (SL) is an extension of Hoare logic built for reasoning about programs that access and mutate data stored on the heap.
Assertions in SL describe partial heaps (usually represented as a partial function from program locations to values), and can be interpreted as giving the program the permission to access the described heap region.
Assertions are composed using the separating conjunction $P * Q$ which states that $P$ and $Q$ hold for disjoint regions of the heap.
This enables local reasoning by means of a \emph{frame rule} that allows one to remove regions of the heap that are not modified by a program fragment while reasoning about it.

The standard way of describing dynamic data structures (i.e. structures covering an unbounded and statically-unknown set of locations) in SL is using inductive predicates.
For example, one can describe a singly-linked null-terminated acyclic list using the predicate
\[ \listPred(x) \defeq x\!=\!\nullVal \land \emp \,\lor\, \exists y.\; x \mapsto y * \listPred(y). \]
A heap $h$ consists of a list from $x$ if either $x = \nullVal$ and $h$ is empty ($\emp$) or $h$ contains a location $x$ whose value is $y$ ($x \mapsto y$) and a \emph{disjoint} null-terminated list beginning at $y$.
The fact that the list is acyclic is enforced by the separating conjunction, which implies that $x$ is not a node in $\listPred(y)$.

One can write similar inductive predicates to describe any other data structure that can be unrolled and decomposed into subpredicates that do not share nodes.
Sequential data structures such as sorted lists, doubly-linked lists, and binary search trees have been successfully verified in this manner.
However, concurrent data structures often use intricate sharing and overlays.

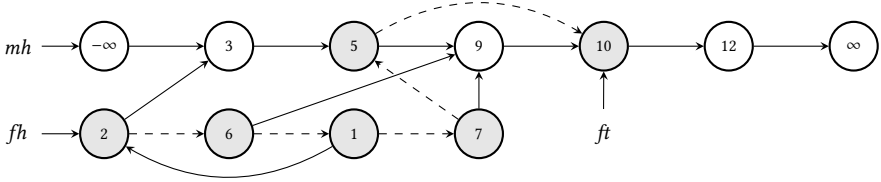
\begin{figure}[t]
  \centering
  \begin{tikzpicture}[>=stealth, font=\footnotesize, scale=0.8, every node/.style={scale=0.8}]
    % Nodes
    \node[stackVar] (hd) {$\mainListHead$};
    \node[unode, right=.5cm of hd] (inf) {$-\infty$};
    \node[unode, right= of inf] (n3) {$3$};
    \node[mnode, right= of n3] (n5) {$5$};
    \node[unode, right= of n5] (n9) {$9$};
    \node[mnode, right= of n9] (n10) {$10$};
    \node[unode, right= of n10] (n12) {$12$};
    \node[unode, right= of n12] (Inf) {$\infty$};
    
    %\node[stackVar] (fr) at ($(hd) + (1.5, 0)$) {$\freeListHead$};
    \node[stackVar, below=.5cm of hd] (fr) {$\freeListHead$};
    \node[mnode, right=.5cm of fr] (n2) {$2$};
    \node[mnode, right= of n2] (n6) {$6$};
    \node[mnode, right= of n6] (n1) {$1$};
    \node[mnode, right= of n1] (n7) {$7$};
    \node[stackVar, right= of n7] (ft) {$\freeListTail$};
    
    % Edges
    \draw[edge] (hd) to (inf);
    \draw[edge] (inf) to (n3);
    \draw[edge] (n3) to (n5);
    \draw[edge] (n5) to (n9);
    \draw[edge] (n9) to (n10);
    \draw[edge] (n10) to (n12);
    \draw[edge] (n12) to (Inf);

    \draw[edge] (fr) to (n2);
    \draw[fedge] (n2) to (n6);
    \draw[edge] (n2) to (n3);
    \draw[fedge] (n6) to (n1);
    \draw[edge] (n6) to (n9);
    \draw[fedge] (n1) to (n7);
    \draw[fedge] (n7) to (n5);
    \draw[edge] (n1) to[bend left=30] (n2);
    \draw[edge] (n7) to (n9);
    \draw[fedge] (n5) to[bend left=30] (n10);
    \draw[edge] (ft) to (n10);
   
  \end{tikzpicture}
  \caption{A potential state of the Harris list with explicit memory
    management. \code{fnext} pointers are shown with dashed edges, marked nodes are shaded gray, and null pointers are omitted for clarity.}
  \label{fig-harris}
\end{figure}

For instance, consider the concurrent non-blocking linked list algorithm due to \citet{DBLP:conf/wdag/Harris01}.
This algorithm implements a set data structure as a sorted list, and uses atomic compare-and-swap (CAS) operations to allow a high degree of parallelism.
\techreport{As with the sequential linked list, }Harris' algorithm inserts a new key $k$ into the list by finding nodes $k_1, k_2$ such that $k_1 < k < k_2$, setting $k$ to point to $k_2$, and using a CAS to change $k_1$ to point to $k$ only if it was still pointing to $k_2$.
However, a similar approach fails for the delete operation.
If we had consecutive nodes $k_1, k_2, k_3$ and we wanted to delete $k_2$ from the list (say by setting $k_1$ to point to $k_3$), there is no way to ensure with a single CAS that $k_2$ and $k_3$ are also still adjacent (another thread could have inserted or deleted in between them).

Harris' solution is a two step deletion: first atomically mark $k_2$ as deleted (by setting a mark bit on its successor field) and then later remove it from the list using a single CAS.
After a node is marked, no thread can insert or delete to its right, hence a thread that wanted to insert $k'$ to the right of $k_2$ would first remove $k_2$ from the list and then insert $k'$ as the successor of $k_1$.

As described so far, this data structure can still be specified and verified using inductive predicates as the underlying shape is still a standard linked list.
Things get complicated when we want to free deleted nodes.
Since marked nodes can still be traversed by other threads, there may still be suspended threads accessing a marked node even after we remove it from the list.
Thus, we cannot free the node immediately.
A common solution is the so-called drain technique which maintains a
second ``free list'' to which marked nodes are added before
they are unlinked from the main list. These nodes are then labeled with a
timestamp indicating when they were unlinked.
When the starting timestamp of all active threads is greater than the
timestamp of a node in the free list, the node can be safely freed.
This leads to the kind of data structure shown in \rF{fig-harris}, where each node has two pointer fields: a \code{next} field for the main list and an \code{fnext} field for the free list (shown as dashed edges).
Threads that have been suspended while holding
a reference to a node that was added to the free list can simply
continue traversing the \code{next} pointers to find their
way back to the unmarked nodes of the main list.

Even if our goal is to verify seemingly simple properties such as that the Harris list is memory safe and not leaking memory, the proof will rely on the following nontrivial invariants
\begin{enumerate}[(a)]
\item \label{it-two-lists} The data structure consists of two (potentially
  overlapping) lists: a list on \code{next} edges beginning at $\mainListHead$ and one on \code{fnext} edges beginning at $\freeListHead$.
\item \label{it-null-list} The two lists are null terminated and
  $\freeListTail$ is an element in the free list.
\item \label{it-free-list} The \code{next} edges from nodes in the
  free list point to nodes in the free list or main list.
\item \label{it-marked} All nodes in the free list are marked.
  %The free list consists of only marked nodes, main list nodes may or may not be marked, except for $\mainListHead$ and $\mainListTail$ which are always unmarked.
%\item \label{it-sorted} Unmarked nodes in the main list are sorted, with $\mainListHead$ having data value $-\infty$ and $\mainListTail$ having $\infty$.
\end{enumerate}

One may be tempted to describe the Harris list simply as two lists
(using the inductive predicate $\listPred$ defined above),
each using a different successor field. However, the two lists may
overlap and so cannot be described by a separating conjunction of the
two predicates. For the same reason, we can also not describe the
entire structure with a
single inductive predicate that uses separating conjunction because the
predicate must describe each node exactly
once\footnote{This is essentially because the disjointness of $*$
  makes $x \mapsto \_ * x \mapsto \_$ unsatisfiable}.
Here, we could have arbitrarily many nodes in the free list with \code{next} edges to the same node in the main list (imagine a series of successive deletions of the predecessor of a node, such as $6, 7, 5$ in \rF{fig-harris}).
Because of this, we cannot describe the list nodes using
separating conjunction as there is no way of telling if we are ``double-counting'' them.
%However, this would not capture the property that the \code{next} field of a node in the free list always points back to another node in the data structure (i.e. it is never null or dangling), which is essential for memory safety.
Also, using an inductive list predicate to describe the main list
would work well only for reasoning about traversals starting from $\mainListHead$, since that is the only possible unrolling of the predicate.
By contrast, threads may enter the main list at arbitrary points in the list by following \code{next} edges from the free list.

To further complicate matters, the \code{next} edges of free list
nodes can even point backward to predecessors in the free list (for e.g. $1$ points back to $2$ in \rF{fig-harris}, creating a cycle).
This means that even ramification-style approaches that use the
overlapping conjunction~\cite{Hobor:2013:RSD:2429069.2429131} would
not be able to specify this structure without switching to a coarse abstraction
of an arbitrary potentially-cyclic graph (akin to indexed separating conjunctions~\cite{SL2,Yang01ShorrWaite}). However, a coarse graph
abstraction could not easily capture key invariants such as that
every node must be reachable from $\mainListHead$ or $\freeListHead$, which
are critical for proving the absence of memory leaks.

In the rest of this paper, we show how to solve these problems and obtain a compositional SL-based framework for reasoning about complex data structures like the Harris list.
In \rSc{sec-interfaces} we define flows, show how they can be used to encode data structure invariants (both shape and data properties), build a flow-based semantic model for separation logic, and define an abstraction, flow interfaces, to reason compositionally about flow-based properties.
\rSc{sec-expressivity} demonstrates the expressivity of our flow-based approach by showing how to describe many intricate data structures.
We then extend a concurrent separation logic with flow interfaces in
\rSc{sec-logic}, show some lemmas to prove entailments in
\rSc{sec-logic-lemmas}, and use them in \rSc{sec-harris} to
verify memory safety and absence of memory leaks of a skeleton of the Harris
list insert operation that abstracts from some of the algorithm's
details but whose proof
still relies on the above invariants \ref{it-two-lists} to~\ref{it-marked}.
Our lemmas for reasoning about flows are
data-structure-agnostic, i.e., they generalize across
many different data structures.
%The \code{search} procedure in \rF{fig-harris-search} is annotated with intermediate assertions in our logic, and although the predicates have not been introduced, one can see that the proof is succinct and comparable to inductive-predicate-based proofs of a simple list traversal.
To demonstrate this feature, we show how to formalize the edgeset framework for concurrent dictionaries using flows in \rSc{sec-dictionaries} and obtain a template algorithm with an implementation-agnostic proof of memory safety and linearizability.

\section{Preliminaries}

\paragraph{Notation.}
The term $\ite(b, t_1, t_2)$ denotes $t_1$ if condition $b$ holds and
$t_2$ otherwise.  If $f$ is a function from $A$ to $B$, we write
$f[x \goesto y]$ to denote the function from $A \cup \set{x}$ defined
by $f[x \goesto y](z) \defeq \ite(z = x, y, f(z))$.  We use
$\set{x_1 \goesto y_1, \dotsc, x_n \goesto y_n}$ for pairwise
different $x_i$ to denote the function
$\emptyFn[x_1 \goesto y_1]\dotsm[x_n \goesto y_n]$, where $\emptyFn$
is the function on an empty domain.  If $f \colon A \times B \to C$,
then we also write $\dom_1(f), \dom_2(f)$ for $A, B$ respectively. For
$f \colon A \to B$ and $C \subseteq A$ we write $f|_C \colon C \to B$
for the function obtained from $f$ by restricting its domain to $C$.
For $f_1 \colon A \to B$ and $f_2 \colon C \to B$ if
$A \cap C = \emptyset$ then we write $f_1 \uplus f_2$ for the function
$f \colon A \cup C \to B$ given by
$f(x) \defeq \ite(x \in A, f_1(x), f_2(x))$.  Moreover, we denote by
$f(C)$ the set $\setcomp{f(c)}{c \in C}$.

\paragraph{Semirings, $\omega$-cpos, and continuous functions.} A
semiring $(\intDom, \intPlus, \intMult, 0, 1)$ is a set $\intDom$
equipped with binary operators
$\intPlus, \intMult \colon \intDom \times \intDom \to \intDom$. The
operation $\intPlus$ is called addition, and the operation $\intMult$
multiplication. The two operators must satisfy the following properties:
(1) $(\intDom, \intPlus, 0)$ is a commutative monoid with identity
$0$; (2) $(\intDom, \intMult, 1)$ is a monoid with identity $1$; (3)
multiplication left and right distributes over addition; and (4)
multiplication with $0$ annihilates $\intDom$, i.e., $0 \intMult d = d
\intMult 0 = 0$ for all $d \in \intDom$.

An $\omega$-complete partial order ($\omega$-cpo) is a set $\intDom$
equipped with a partial order $\sqsubseteq$ on $\intDom$ such that all
increasing chains in $\intDom$ have suprema in $\intDom$. A function
$f \colon D_1 \to D_2$ between two $\omega$-cpos $(D_1, \sqsubseteq_1)$ and
$(D_2, \sqsubseteq_2)$ is continuous if $f(\bigsqcup_1 X) =
\bigsqcup_2 f(X)$ for any increasing chain $X \subseteq D_1$. Here,
$\bigsqcup_i X$ denotes the supremum of a chain $X$ in $D_i$.

\section{The Flow Framework}
\label{sec-interfaces}

This section presents the formal treatment of our flow framework.
% We start by describing flow domains and the graphs over which we define flows.
% Using flows, we demonstrate how to model many different data structures and properties.
% We then build a separation algebra, consisting of so-called flow
% graphs, that preserves local conditions on the flow of a graph.
% We finish by defining abstractions of flow graphs, flow interfaces, that can be used in a logic for compositional reasoning.

Our theory is parameterized by the domain over which
flows range. This domain is equipped with operations for calculating
flows, which must satisfy certain algebraic properties.

\begin{definition}[Flow Domain]
  A \emph{flow domain}
  $(\intDom, \intLeq, \intJoin, \intPlus, \intMult, 0, 1)$ is a
  positive partially ordered semiring that is $\omega$-complete. That
  is,
  \begin{enumerate*}
  \item $(\intDom, \intLeq)$ is an $\omega$-cpo and $\intJoin$ its join;
  \item $(\intDom, \intPlus, \intMult, 0, 1)$ is a semiring;
  \item $\intPlus$ and $\intMult$ are continuous with respect to $\intLeq$; and
  \item $0$ is the smallest element of $\intDom$.
  \end{enumerate*}
  We identify a flow domain with its support set $\intDom$.
\end{definition}

% \begin{example}
% Any completely distributive lattice $(\intDom, \intLeq, \intJoin, \intMeet, \bot,
% \top)$ induces a flow domain $(\intDom, \intLeq, \intJoin, \intJoin, \intMeet, \bot,
% \top)$.
% \end{example}
\begin{example}
\label{ex-intDom-nat-infty}
The natural numbers extended with infinity, $\Nat^\infty \defeq (\Nat
\cup \set{\infty}, \leq, \max, +, \cdot, 0, 1)$, form a flow domain
and so does $(\intDom, \intLeq, \intJoin, \intJoin, \intMeet, \bot,
\top)$ for any completely distributive lattice $(\intDom, \intLeq, \intJoin, \intMeet, \bot, \top)$.
\end{example}

In our formalization, we consider graphs whose nodes are labeled from
some set of node labels $\nodelabelDom$ that encode relevant
information contained in each node (e.g. the node's content or the id
of the thread holding a lock on the node). We will later build
abstractions of graphs that compute summaries of these node
labels. For this purpose, we require that $\nodelabelDom$ is equipped
with a partial order $\nodelabelLeq$ that induces a join-semilattice
$(\nodelabelDom, \nodelabelLeq, \nodelabelJoin, \nodelabelBot)$ with
join $\nodelabelJoin$ and smallest element $\nodelabelBot$.

For the remainder of this section we fix a flow domain
$(\intDom, \intLeq, \intJoin, \intPlus, \intMult, 0, 1)$ and node
label domain $(\nodelabelDom, \nodelabelLeq, \nodelabelJoin, \nodelabelBot)$. 
We use the same symbols for the partial order and join operator
on the two domains. However, it will always be clear from the context
which one is meant. All definitions are implicitly parameterized by
$\intDom$ and $\nodelabelDom$.

\subsection{Flows}
\label{sec-flow-graphs}

\begin{figure}[t]
  \centering
  \begin{tikzpicture}[>=stealth, scale=0.7, every node/.style={scale=0.7}, font=\footnotesize]
    \def\ysep{-1.8}
    \def\xsep{1.5}
    \def\xshift{3.5}

    % G
    \node[gnode] (r) {$n_0\colon 1$};
    \node[gnode] (n0) at ($(r) + (\xsep, 0)$) {$n_1\colon 1$};
    \node[gnode] (n1) at ($(n0) + (\xsep, 0)$) {$n_2\colon 1$};
    \node[gnode] (n00) at ($(n0) + (0, \ysep)$) {$n_3\colon 1$};
    \node[gnode] (n01) at ($(n0) + (\xsep, \ysep)$) {$n_4\colon 1$};
    \node[gnode] (n010) at ($(n01) + (\xsep, 0)$) {$n_5\colon 1$};
    \node[gnode] (n10) at ($(n1) + (\xsep, 0)$) {$n_6\colon 1$};
    \draw[edge] (r) to (n0);
    \draw[edge] (r) to[bend left=45] (n1);
    \draw[edge] (n0) to (n00);
    \draw[edge] (n0) to (n01);
    \draw[edge] (n01) to (n010);
    \draw[edge] (n1) to (n10);
    \node[inflow] (r-in) at ($(r.north west) + (-.5, .5)$) {$1$};
    \draw[edge, bend right] (r-in) to (r);

    % Draw a region around G1 in G
    \begin{scope}[on background layer]
      \def\s{.2}
      \draw[draw=blue!50, rounded corners, thick, fill=blue!10]
        ($(n0.north west) + (-\s, \s)$) -- ($(n1.north east) + (\s, \s)$) -- ($(n01.south east) + (\s, -\s)$) -- ($(n01.south west) + (0, -\s)$) -- ($(n0.south west) + (-\s, 0)$) -- cycle;
    \end{scope}

    % G1
    \node[gnode, right=(\xshift) of n0] (n0') {$n_1\colon 1$};
    \node[gnode, right=(\xshift) of n1] (n1') {$n_2\colon 1$};
    \node[pnode, right=(\xshift) of n00] (n00') {};
    \node[gnode, right=(\xshift) of n01] (n01') {$n_4\colon 1$};
    \node[pnode, right=(\xshift) of n010] (n010') {};
    \node[pnode, right=(\xshift) of n10] (n10') {};
    \draw[edge] (n0') to (n00');
    \draw[edge] (n0') to (n01');
    \draw[edge] (n01') to (n010');
    \draw[edge] (n1') to (n10');
    \node[inflow] (n0-in') at ($(n0'.north west) + (-.5, .5)$) {$1$};
    \draw[edge, bend right] (n0-in') to (n0');
    \node[inflow] (n1-in') at ($(n1'.north west) + (-.5, .5)$) {$1$};
    \draw[edge, bend right] (n1-in') to (n1');

    % G2
    \node[gnode, right=2.5*\xshift of r] (r'') {$n_0\colon 1$};
    \node[pnode, right=2.5*\xshift of n0] (n0'') {};
    \node[pnode, right=2.5*\xshift of n1] (n1'') {};
    \node[gnode, right=2.5*\xshift of n00] (n00'') {$n_3\colon 1$};
    \node[gnode, right=2.5*\xshift of n010] (n010'') {$n_5\colon 1$};
    \node[gnode, right=2.5*\xshift of n10] (n10'') {$n_6\colon 1$};
    \draw[edge] (r'') to (n0'');
    \draw[edge] (r'') to[bend left=45] (n1'');
    \node[inflow] (r-in'') at ($(r''.north west) + (-.5, .5)$) {$1$};
    \draw[edge, bend right] (r-in'') to (r'');
    \node[inflow] (n00-in'') at ($(n00''.north west) + (-.5, .5)$) {$1$};
    \draw[edge, bend right] (n00-in'') to (n00'');
    \node[inflow] (n010-in'') at ($(n010''.north west) + (-.5, .5)$) {$1$};
    \draw[edge, bend right] (n010-in'') to (n010'');
    \node[inflow] (n10-in'') at ($(n10''.north west) + (-.5, .5)$) {$1$};
    \draw[edge, bend right] (n10-in'') to (n10'');

    % The equation
    \node[phantomNode, font=\huge] at ($($(n10)!0.5!(n010)$)!0.5!($(n0')!0.5!(n00')$)$) {$=$};
    \node[phantomNode, font=\huge] at ($($(n10')!0.5!(n010')$)!0.35!($(n0'')!0.5!(n00'')$)$) {$\graphInComp$};
  \end{tikzpicture}
  \caption{A decomposition of an inflow/graph pair
    $(\inflow, G)$ into the boxed blue region $(\inflow_1, G_1)$ and its context $(\inflow_2, G_2)$. All shown edges have edge label $1$, missing edges have label $0$, and the path-counting flow is used. Non-zero inflows are shown as curved arrows at each node, nodes are labeled with names and the resulting flows, and sink nodes are shown as dotted circles.}
  \label{fig-flow-graph-decomp}
\end{figure}
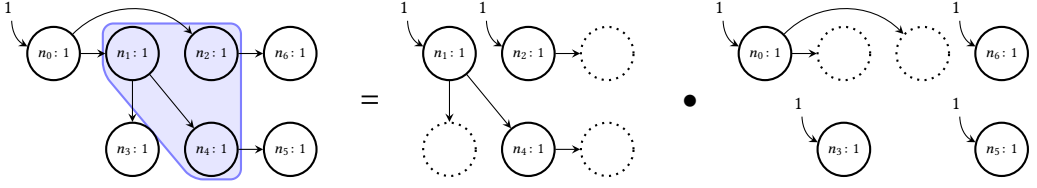

Flows are calculated over directed graphs that serve as an
intermediate abstraction
of heap graphs. The graphs are partial, i.e., they may have outgoing
edges to nodes that are not itself part of the graph. The edges are
labeled by elements of the flow domain and the nodes with node
labels. Formally, given a (potentially infinite) set of nodes $\nodeDom$, a \emph{(partial) graph}
$G = (N, N^o, \nodelabelling, \edgelabel)$ consists of a finite set of
nodes $N \subseteq \nodeDom$, a finite set of \emph{sink nodes} $N^o \subseteq \nodeDom$ disjoint from $N$,
a node labeling function $\nodelabelling \colon N \to \nodelabelDom$,
and an edge function
$\edgelabel \colon N \times (N \cup N^o) \rightarrow \intDom$. Note
that the edge function is total on $N \times N^o$. The absence
of an edge between two nodes $n,n'$ is indicated by
$\edgelabel(n,n') = 0$. We let $\dom(\graph) = N$ and sometimes identify
$\graph$ and $\dom(\graph)$ to ease notational burden. The (unique) graph
defined over the empty set of nodes and sinks is denoted by $\graphEmpty$.
We define the disjoint union of two
graphs $\graph_1 \graphPlus \graph_2$ in the expected way:
\begin{align*}
  \graph_1 \graphPlus \graph_2 &\defeq
                   \begin{cases}
                     (N_1 \cup N_2, (N_1^o \setminus N_2) \cup (N_2^o \setminus N_1), \nodelabelling_1 \uplus \nodelabelling_2, \edgelabel') & \text{if } N_1 \cap N_2 = \emptyset \\
                     \text{undefined} & \text{otherwise}
                   \end{cases} \\
  \text{where } \edgelabel'(n_1, n_2) &\defeq
                            \begin{cases}
                              \edgelabel_1(n_1, n_2) & n_1 \in N_1
                              \land n_2 \in N_1 \cup N_1^o \\
                              \edgelabel_2(n_1, n_2) & n_1 \in N_2
                              \land n_2 \in N_2 \cup N_2^o \\
                              0 & \text{otherwise.}
                            \end{cases}
\end{align*}
A graph $\graph_1$ is a subgraph of another graph $\graph$, denoted
$\graph_1 \subseteq \graph$, if and only if there exists $\graph_2$ such that
$\graph = \graph_1 \graphPlus \graph_2$.

\begin{example}
The graph $\graph_1$ in the middle of \rF{fig-flow-graph-decomp} is $(\set{n_1, n_2, n_4}, \set{n_3, n_5, n_6}, \nodelabelling_1, \edgelabel_1)$ for some $\nodelabelling_1$ and an $\edgelabel_1$ that sets all pairs in $\set{(n_1, n_3), (n_1, n_4), (n_2, n_6)}$ to $1$ and all other pairs to $0$.
Similarly, $\graph_2$ on the right is $(\set{n_0, n_3, n_5, n_6}, \set{n_1, n_2}, \nodelabelling_2, \edgelabel_2)$ for appropriate $\nodelabelling_2$ and $\edgelabel_2$, and the composed graph $\graph = \graph_1 \graphPlus \graph_2$ is shown on the left of \rF{fig-flow-graph-decomp}.
\end{example}

Let $\graph = (N, N^o, \nodelabelling, \edgelabel)$ be a graph. A \emph{flow} of
$\graph$ is a function $\netflow(\inflow, \graph) \colon N \to D$ that is calculated as a certain
fixpoint over $\graph$'s edge function starting from a given
\emph{inflow} $\inflow \colon N \to D$ going into $\graph$.
To motivate the definition of $\netflow(\inflow, \graph)$, let us use
the flow domain $\Nat^\infty$ from \rE{ex-intDom-nat-infty}. We can
use this flow domain to define a flow that counts the number of paths
between a fixed node and any other node.
First, view the edge function $\edgelabel$ as an adjacency matrix $E$ indexed by the nodes $N \cup N^o$.
If one looks at the matrix product $E \cdot E = E^2$ (using the addition and multiplication
operations of the flow domain), then the $(n, n')$-th entry corresponds to the sum over the products of edge labels on all paths of length 2 between $n$ and $n'$.
For $\graph$ from \rF{fig-flow-graph-decomp}, for example, the $(n_0, n_3)$-th entry of $E^2$ is $1$ as there is exactly $1$ non-zero length $2$ path, while the $(n_0, n_1)$-th entry is $0$ as there is no non-zero length $2$ path.
Extending this idea, for any graph where edges are labeled from $\set{0, 1}$, the $(n, n')$-th entry of matrix $E^l$ is the number of $l$ length paths from $n$ to $n'$.
Thus, the matrix $C = I + E + E^2 + \dots$ tells us the number of all paths between two nodes.
Given a root node $r$, we can look at the $r$-th row of this matrix to get the number of paths to any node from the root.
To generalize this to multiple roots, we can calculate the number of paths to any node by looking at $\vec{\inflow} \cdot C$ where $\vec{\inflow}$ is a row vector with a $1$ for every root node and $0$ otherwise.

Formally, first let the \emph{capacity}
$\capacity(\graph)\colon N \times (N \cup N^o) \rightarrow \intDom$ be the least fixpoint of the following equation:
\begin{align*}
  \capacity(\graph)(n, n') &=  \init(n, n') \intPlus \intBigPlus_{n'' \in \graph} \edgelabel(n, n'') \intMult \capacity(\graph)(n'', n') \quad
  \text{where} \quad \init(n, n') \defeq \ite(n = n', 1, 0).
\end{align*}
Note that the fixpoint exists because $D$ is $\omega$-complete and
$\intPlus$ and $\intMult$ are continuous. Then define
\begin{align*}
  \netflow(\inflow, \graph)(n) &\defeq \intBigPlus_{n' \in \graph} \inflow(n') \intMult \capacity(\graph)(n', n).
\end{align*}

\begin{example}
  In \rF{fig-flow-graph-decomp}, the capacity of $\graph_2$ is equal to its edge function, the capacity of $\graph_1$ is its edge function with $(n_1, n_5)$ additionally $1$, and $\capacity(\graph)$ is $1$ for all pairs $(n, n')$ such that $n$ is an ancestor of $n'$ in the tree.
  For the inflow $\inflow(x) = \ite(x = n_0, 1, 0)$, $\netflow(\inflow, \graph) = (\lambda x. 1)$.
\end{example}

Finally, a key property of flows is that they can be composed by
considering the product flow on the product of their flow domains.
%This is where the real power of flow-based reasoning shows itself: using local conditions on the product of flows, one can easily describe complex structures with overlays and sharing.
\begin{lemma}
  \label{lem-product-flow}
  Let $(\intDom_1, \intLeq_1, \intJoin_1, \intPlus_1, \intMult_1, 0_1, 1_1)$ and $(\intDom_2, \intLeq_2, \intJoin_2, \intPlus_2, \intMult_2, 0_2, 1_2)$ be flow domains.
  $(\intDom_1 \times \intDom_2, \intLeq, \intJoin, \intPlus, \intMult,
  (0_1, 0_2), (1_1, 1_2))$, where $\intLeq, \intJoin, \intPlus$, and $\intMult$ operate on each component respectively, is a flow domain.
  Moreover, the flow on the product domain is the product of the flows on each component.
\end{lemma}

\subsection{Expressivity}
\label{sec-expressivity}

We now give a few examples of flows to demonstrate the range of data structures that can be described using local conditions on flows.

The shape of all common structures that can be described with inductive predicates can be described using flows, e.g. lists (singly and doubly linked, cyclic), trees, and nested combinations of these.
By considering products with flows for data properties, we can also describe structures such as sorted lists, binary heaps, and search trees.
Going beyond inductive predicates, the flow framework can describe structures with unbounded sharing and irregular traversals (e.g. Harris list), overlays (e.g. threaded and B-link trees), as well as those with irregular shape (e.g. DAGs and arbitrary graphs).

In this section, we assume that the local condition on the flow of each node is specified by a \emph{good condition}.
This is a predicate $\goodCondition(n, \inflow, \nodelabel, \flow)$ taking as
arguments the node $n$, the inflow function $\inflow$ on the singleton graph containing $n$, the node label $\nodelabel$ of $n$, and the edge function $\flow$ of that singleton graph.
The singleton inflow of $n$ is equal to the flow at $n$ in the global graph representing the data structure of interest containing $n$.
This fact, as well as the reason for this formalization, will be seen in \rSc{sec-flow-interfaces}.
We also assume that we can specify the global inflow of our data structure.

\subsubsection{Path-counting Flow}
\label{example:pathcounting}
Consider the flow domain $\Nat^\infty$ from \rE{ex-intDom-nat-infty}, with an arbitrary node label domain.

\begin{lemma}
  For any graph $\graph$, if the edge label of every edge $(n, n')$ is interpreted as the number of edges between $n$ and $n'$, then $\capacity(\graph)(n, n')$ is the number of paths from $n$ to $n'$.
  If the inflow $\inflow$ takes values in $\set{0,1}$ then the path-counting flow $\netflow(\inflow, \graph)(n)$ is the number of paths from nodes in $\setcomp{n'}{\inflow(n') = 1}$ to $n$.
\end{lemma}

We can use this flow to describe shapes such as lists and trees.
For instance, a graph is a tree if and only if every node has exactly one path from the root.
Thus, if a graph $\graph$ satisfies the good condition
\[ \goodCondition_{\m{tr}}(n, \inflow, \_, \_) \defeq \inflow(n) = 1\]
at all nodes $n$, given a global inflow satisfying $\inflow(n) = \ite(n = n_0, 1, 0)$, then $\graph$ must be a tree rooted at $n_0$.
As a singly-linked list is a tree where every node has exactly one outgoing edge, we can describe it using the good condition
\[ \goodCondition_{\m{ls}}(n, \inflow, \_, \flow) \defeq \exists n'.\; \inflow(n) = 1 \land (\flow = \emptyFn \lor \flow = \set{(n, n') \goesto 1}) \]
where $\emptyFn$ stands for the empty function (i.e. the node has no
outgoing edges). % and $\_$ is a fresh existential variable.
A graph $\graph$ satisfying $\goodCondition_{\m{ls}}$ under a global
inflow $\inflow$ with $\inflow(n) = \ite(n = n_0, 1, 0)$ for all $n$ must be a list beginning at $n_0$.
To model a null-terminated list, we must further require that the entire graph $\graph$ has no outgoing edges.
This can be done using our flow interface abstraction in \rSc{sec-flow-interfaces}.
If we want to model a list that terminates at a specific node, say
$\freeListTail$ as in
the free list shown in Fig.~\ref{fig-harris}, we can replace the last clause of $\goodCondition_{\m{ls}}$ with $\flow = \ite(n = \freeListTail, \emptyFn, \set{(n, n') \goesto 1})$.
% If $\zerofun$ is the function mapping everything
% to the $0$ value of the flow domain, then $\graphPred_{\gamma_1}(\interface)$ with
% $\inflowOfInt{\interface}.\zerofun=\set{h \goesto 1}.\zerofun$ and
% $\flowmapOfInt{\interface}=\emptyFn$ denotes a list beginning at $h$
% and terminating at $null$\footnote{This is because: (1) Every node in
%   the graph has one path from $h$ and at most one outgoing edge,
%   implying the graph is a list. (2) Since the flow map has no sink
%   nodes and since we only consider finite graphs,
%   the list must terminate at some node with no outgoing edges. (3)
%   Assuming the abstraction from a concrete heap to a flow graph only
%   removes $null$ edges, the last node must point to null.} so we can use
% this to model the free list (replacing $h$ with $\freeListHead$).
% \footnote{\tw{main?} Sid: No, I meant free - since the main list ends at tl. Used another variable for the head.}
% If we wanted to model a list that terminates at a specific node, for instance the main list, we can define the flow map of every node to be $\flow = \ite(n = \mainListTail, \emptyFn, \set{(n, n') \goesto 1})$.

Similarly, one can describe cyclic lists by requiring that each node has exactly one outgoing edge and that the path count of each node is $\infty$.
To describe a tree of a particular arity, like a binary tree, one can
constrain the path count to be $1$ and add a further condition restricting the number of outgoing edges at
each node appropriately.
For overlaid structures like the threaded tree, which is a tree where all elements also form a list (in an arbitrary order), we can use the product of two path-counting flows, one from the root of the tree and one from the head of the list.

\subsubsection{Last-edge Flow}
We next show a flow domain where the flow along a path is equal to the label of the last edge on the path.

\begin{lemma}
Given an $\omega$-cpo $(\intDom, \intLeq, \intJoin)$ with smallest element $0$, let $1$ be a fresh element. There exists a flow domain $(\intDom \uplus \set{1}, \intLeq, \intJoin, \intJoin, \rtimes, 0, 1)$, where multiplication is the projection operator, defined as $d_1 \rtimes d_2 \defeq \ite(d_1 = 0, 0, \ite(d_2 = 1, d_1, d_2))$ and the ordering is extended with $0 \intLeq 1$.
\end{lemma}

This flow has many applications; for instance, we can start with a lattice of types and use the last-edge flow to enforce that all reachable nodes are well-typed in the presence of type-unsafe operations.
Another example is doubly linked lists, which can be described by using a product of two path-counting flows, one counting paths from the head and the other from the tail, and a last-edge flow to ensure that for any node \code{n}, \code{n.next.prev} is equal to \code{n}.

We can also use this flow to encode nested combinations of structures such as trees of lists.
We start with a path-counting flow to enforce that the entire graph forms a tree and then take a product with the last-edge flow domain built from the $\omega$-cpo on $\set{0, d_T, d_L}$ where $0$ is the smallest element and the others are unordered.
We then use the good condition to label all edges to tree nodes with $d_T$, all edges to list nodes with $d_L$, and enforce that the label on the incoming edge matches the type of the node and that list nodes have no edges to tree nodes.

\subsubsection{Upper and Lower-bound Flows}
We next turn to data properties like sortedness.
Consider the lower-bound flow domain formed by the integers extended
with positive and negative infinities: $(\Int \cup \set{-\infty,
  \infty}, \geq, \min, \min, \max, \infty, -\infty)$.
Assume further that the node label domain consists of sets of data
values and that each node is labeled with the singleton set of its content.
If we label each edge of a list with the value at the source node and start with an inflow of $-\infty$ at the root $h$, then the flow at a node is the maximum value of its predecessors.
This value is like a lower bound for the value of the current node; if the value at every node reachable from $h$ is greater than its lower-bound flow, then the list beginning at $h$ is sorted in ascending order.
To do this, we use the product of the path-counting and the lower-bound flow domains.
We can use the good condition
\[ \goodCondition_{\m{sls}}(n, \inflow, \nodelabel, \flow) \defeq \exists n', l, k.\; \inflow(n) = (1, l) \land \nodelabel = \set{k} \land \flow = \ite(n' = \mathit{null}, \emptyFn, \set{(n, n') \goesto k}) \land l \leq k \]
and a global inflow $\inflow(n) = \ite(n = n_0, (1, -\infty), (0, \infty))$ to describe a sorted list beginning at $n_0$.

Note that, unlike in the definitions of inductive SL predicates, the data constraints are
not tied to the shape constraints -- we could just as easily use the
lower-bound flow condition of sortedness and the path-counting flow
condition of being a tree to describe a min-heap.

Analogously, we can use the upper-bound flow domain $(\Int \cup \set{-\infty, \infty}, \leq, \max, \max, \min, -\infty, \infty)$ to describe lists sorted in descending order and max-heaps.
We can also use a product of a lower-bound and an upper-bound flow to describe a binary search tree by enforcing that each node propagates to its child the appropriate bounds on values that can be present in the child's subtree.

We will see a more general flow, the inset flow, that can be used to specify data properties of dictionary data structures that subsumes the examples above in \rSc{sec-dictionaries}.

\subsection{Flow Graph Algebras}
\label{sec-flow-interface-algebra}

We next show how to define a separation algebra that can be used to give semantics to separation logic
assertions in a way that allows us to maintain flow-dependent
invariants.
This will enable us to reason locally about graphs and their flows using the $*$ operator of SL.

\begin{definition}[Separation Algebra~\cite{DBLP:conf/lics/CalcagnoOY07}]
  A \emph{separation algebra} is a cancellative, partial commutative
  monoid $(\Sigma, \bullet, u)$. That is, $\bullet \colon \Sigma \times
  \Sigma \pto \Sigma$ and $u \in \Sigma$ such that the following
  properties hold:
  \begin{enumerate}
  \item Identity: $\forall \sigma \in \Sigma.\; \sigma \bullet u = \sigma$.
  \item Commutativity: $\forall \sigma_1, \sigma_2 \in \Sigma.\; \sigma_1 \bullet \sigma_2 = \sigma_2 \bullet \sigma_1$.
  \item Associativity: $\forall \sigma_1, \sigma_2, \sigma_3 \in \Sigma.\; \sigma_1 \bullet (\sigma_2 \bullet \sigma_3) = (\sigma_1 \bullet \sigma_2) \bullet \sigma_3$.
  \item Cancellativity: $\forall \sigma, \sigma_1, \sigma'_1, \sigma_2 \in \Sigma.\; \sigma = \sigma_1 \bullet \sigma_2 \land \sigma = \sigma'_1 \bullet \sigma_2 \impl \sigma_1 = \sigma'_1$.
  \end{enumerate}
  Here, equality means that either both sides are defined and equal,
  or both sides are undefined.
\end{definition}

A simple way of obtaining a separation algebra of graphs is to
define the composition $\bullet$ in terms of disjoint
union. Effectively, this yields the standard model of separation
logic.
However, suppose we have a graph $\graph$ satisfying some local condition on the flow at each node that we want to decompose into two subgraphs $\graph_1$ and $\graph_2$, say to verify a program that traverses the Harris list.
To infer that the subgraph $\graph_1$ also satisfies the same local condition on the flow, we need to be able to compute the flow of nodes in $\graph_1$ without looking at the entire graph.
We thus want a stricter version of composition that
allows us to decompose a flow on $\graph$ into flows
on $\graph_1$ and $\graph_2$.
Our solution is to use graphs equipped with inflows as the elements of our separation algebra.

Since we allow our graphs to have cycles, it is hard to define the composed inflow of two inflows on subgraphs.
Instead, we start with an inflow/graph pair $(\inflow, \graph)$ and show how to decompose it into an inflow for its constituent subgraphs $\graph_i$.
To do this, we \emph{project} the inflow $\inflow$ to obtain a new
inflow function $\inflow_i$ for $\graph_i$ that results in the same
flow at each node of $\graph_i$ as $\inflow$ on $\graph$.
The projected inflow $\inflow_1(n)$ of a node in $\graph_1$ is equal to $\inflow(n)$ plus the contribution of $\graph_2$ to the flow of $n$ (and vice versa for $\inflow_2$ on $\graph_2$).
We will use projection to define the composition of inflow/graph pairs.

\begin{definition}
The projection of an inflow $\inflow$ on a graph $\graph = \graph_1 \graphPlus \graph_2$ onto $\graph_1$ is a function $\projection{(\inflow, \graph)}(\graph_1) \colon \dom(\graph_1) \rightarrow \intDom$ defined as:
\begin{align*}
  \projection{(\inflow, \graph)}(\graph_1)(n) &\defeq \inflow(n) \intPlus
                                             \intBigPlus_{n' \in \graph \setminus \graph_1}
                                             \netflow(\inflow, \graph)(n') \intMult \edgelabel(n', n).
\end{align*}
\end{definition}

The following lemma tells us that the flow induced by
the projection of an inflow $\inflow$ onto a subgraph is the same as the flow
induced by $\inflow$ on the larger graph.
In other words, we can view the inflow as containing all information about the context of a graph that is needed for calculation of the flow of nodes in the graph.

\begin{lemma}
  \label{lem-projection-impl-inset-equal}
  If $\inflow$ is an inflow on $\graph$ and $\inflow_1 =
  \projection{(\inflow, \graph)}(\graph_1)$ for some $\graph_1 \subseteq
  \graph$, then $\netflow(\inflow, \graph)|_{\graph_1} =
  \netflow(\inflow_1, \graph_1)$.
\end{lemma}

\begin{example}
\rF{fig-flow-graph-decomp} uses the path-counting flow and shows a split of a graph $\graph$ into the boxed blue region $\graph_1$ and its context $\graph_2$.
The inflow on $\graph$ is $\inflow(n) = \ite(n = n_0, 1, 0)$, and its projections onto $\graph_1$ and $\graph_2$ are shown as curved arrows entering the nodes in each graph.
\end{example}

We next define a first approximation of a composition operator on inflow/graph pairs that preserves the flow at each node.
The composition of two inflow/graph pairs
$(\inflow_1, \graph_1)$ and $(\inflow_2, \graph_2)$, denoted
$(\inflow_1, \graph_1) \graphInComp (\inflow_2, \graph_2)$, is the set
of all pairs $(\inflow, \graph)$ such that $\graph_1$ and $\graph_2$ compose to $\graph$
and the component inflows $\inflow_1$ and $\inflow_2$ are the projections of
$\inflow$ onto the respective subgraphs:
\[
  (\inflow, \graph) \in (\inflow_1, \graph_1) \graphInComp (\inflow_2, \graph_2) \defiff
  \graph = \graph_1 \uplus \graph_2 \land \forall i \in \set{1,2}.\; \projection(\inflow, \graph)(\graph_i) = \inflow_i.
\]
\rL{lem-projection-impl-inset-equal} then ensures that the flow is the same in the composite graph as in the subgraphs.
Note that this set may be empty even if $\graph_1$ and $\graph_2$ are disjoint
since the inflows of the two subgraphs may not be compatible.

Unfortunately, $(\inflow_1, \graph_1) \graphInComp (\inflow_2, \graph_2)$ may
contain more than one element since there may be many possible
composite inflows $\inflow$ whose projection onto the subgraphs yields
$\inflow_1$ and $\inflow_2$.
For example\techreport{\ftodo{Promote to figure if we have space.}}, let $\graph_1, \graph_2$ be singleton graphs on nodes $n_1, n_2$ respectively that each have an edge with label $\infty \in \Nat^\infty$ to the other.
If $\inflow_1 = \set{n_1 \goesto \infty}$ and $\inflow_2 = \set{n_2 \goesto \infty}$, then the composite inflow can be $\set{n_1 \goesto 1}$ or $\set{n_2 \goesto 1}$ (among many others).
%\ftodo{Should we mention that this is a
%  non-deterministic monoid and so may be okay to build separation
%  logics, as done by \cite{TODO} for Boolean-BI algebras?}
So we cannot immediately use $\graphInComp$ to define the partial
monoid for our separation algebra. However, as we shall see, all
inflows $\inflow$ satisfying the condition of the composition induce
the same flow on $\graph$.  Thus, let us define an equivalence relation
$\inflowEquiv_\graph$ on inflows for a given graph $\graph$. The
relation $\inflowEquiv_\graph$ relates all inflows that induce the
same flows:
\[ \inflow \inflowEquiv_\graph \inflow' \defiff \netflow(\inflow,
  \graph) = \netflow(\inflow', \graph). \]
Given an inflow $\inflow$, we denote its equivalence class with
respect to $\inflowEquiv_\graph$ as $[\inflow]_\graph$ and extend
$\netflow$ from inflows to their equivalence classes in the expected way.
We also write $\dom(\inflowClass) \defeq \dom(\inflow)$ for $\inflowClass = [\inflow]_\graph$, and use $\inflowDom$ to denote the set of all inflow equivalence classes.

Using Lemma~\ref{lem-projection-impl-inset-equal} we can show that the
equivalence relation $\inflowEquiv$ is a congruence on $\graphInComp$:
\begin{lemma}
  \label{lem-inflow-equiv-congruence}
  The following two implications hold:
  \begin{enumerate}
   \item
     $(\inflow, \graph) \in (\inflow_1, \graph_1) \graphInComp (\inflow_2, \graph_2) \land \inflow_1 \inflowEquiv_{\graph_1} \inflow'_1 \land \inflow_2 \inflowEquiv_{\graph_2} \inflow'_2 \land (\inflow', \graph) \in (\inflow'_1, \graph_1) \graphInComp (\inflow'_2, \graph_2) \Rightarrow \inflow \inflowEquiv_\graph \inflow'$.
   \item $(\inflow, \graph) \in (\inflow_1, \graph_1) \graphInComp (\inflow_2, \graph_2) \land \inflow \inflowEquiv_{\graph} \inflow' \Rightarrow$\\
     $\exists \inflow'_1, \inflow'_2.\; \inflow_1 \inflowEquiv_{\graph_1} \inflow'_1 \land \inflow_2 \inflowEquiv_{\graph_2} \inflow'_2 \land (\inflow', \graph) \in (\inflow'_1, \graph_1) \graphInComp (\inflow'_2, \graph_2)$.
   \end{enumerate}
\end{lemma}

Lemma~\ref{lem-inflow-equiv-congruence} implies that
$\graphInComp$ yields a partial function on pairs of graphs $\graph$ and
their inflow equivalence classes. This suggests the following
definition for our separation algebra.

\begin{definition}[Flow Graph Algebra]
  The \emph{flow graph algebra}
  $(\graphInDom, \graphInComp, \graphInEmpty)$ for flow domain
  $\intDom$ and node label domain $\nodelabelDom$ is defined by
  \begin{align*}
    \graphIn \in \graphInDom \defeq & \setcomp{([\inflow]_\graph, \graph)}{\inflow \colon \dom(\graph) \to \intDom}\\
  ([\inflow_1]_{\graph_1}, \graph_1) \graphInComp
    ([\inflow_2]_{\graph_2}, \graph_2)
    \defeq &
  \begin{cases}
    ([\inflow]_\graph, \graph) & \graph = \graph_1 \graphPlus \graph_2
    \land \forall i \in \set{1,2}.\; \projection(\inflow, \graph)(\graph_i) \in [\inflow_i]_{\graph_i} \\
    \text{undefined} & \text{otherwise}
  \end{cases}\\
    \graphInEmpty \defeq & ([\inflowEmpty]_{\graphEmpty}, \graphEmpty)
  \end{align*}
  where $\graphEmpty$ is the empty graph and $\inflowEmpty$ the
  inflow on an empty domain. We call the elements $\graphIn \in \graphInDom$
  \emph{flow graphs}. We again let $\dom(\graphIn) = \dom(\graph)$ and write $\graphIn$ for $\dom(\graphIn)$ when it is clear from the context.
\end{definition}

\begin{theorem}
  \label{thm-flow-graphs-separation-algebra}
  The flow graph algebra $(\graphInDom, \graphInComp, \graphInEmpty)$ is a separation algebra.
\end{theorem}

\subsection{Abstracting Flow Graphs with Flow Interfaces}
\label{sec-flow-interfaces}

We can now use flow graphs to give semantics to separation
logic assertions. However, we also want to be able to use predicates
in the logic that describe sets of such graphs so that the flow of
each graph in the set satisfies certain local invariants at each of
its nodes. In the following, we introduce semantic objects, which we
call \emph{flow interfaces}, that provide such abstractions. We will
then lift the composition operator $\graphInComp$ from flow 
graphs to flow interfaces and prove important properties about flow
interface composition. In \rSc{sec-logic}, we will
use these properties to justify general lemmas for proving
entailments in a separation logic with flow interfaces.

Our target abstraction must have certain properties.
Recall from \rSc{sec-expressivity} that to describe data structure invariants as conditions on the flow we fixed a particular inflow to the global graph.
Thus, our abstraction must fix the inflow of the abstracted graph.
Secondly, our abstraction must be sound with respect to flow graph composition: if we have a graph $\graph = \graph_1 \graphPlus \graph_2$, and $\graph_1$ and $\graph'_1$ satisfy the same interface, then $\graph'_1 \graphPlus \graph_2$ must be defined and satisfy the same interface as $\graph$.

The second property is a challenging one, for the effect of a local modification in a flow graph are not necessarily local: small changes in $\graph_1$ may affect the flow of nodes in $\graph_2$.
For example, in \rF{fig-flow-graph-decomp} if we modified $\graph_1$ by adding a new outgoing edge from $n_1$ to $n_6$, then this would increase the path count of $n_6$ to $2$ and we can no longer compose with $\graph_2$ as it expects an inflow of $1$ at $n_6$.
We need to find a characterization of modifications to $\graph_1$ that preserve the inflow of every node in $\graph_2$.

Intuitively, the internal structure of $\graph_1$ should be irrelevant for computing the flow of $\graph_2$; from the perspective of $\graph_2$, $\graph_1$ should be replaceable by any other flow graph that provides $\graph_2$ with the same inflow.
One may try to define an outflow quantity, analogous to inflow, that specifies the flow that a graph gives to each of its sink nodes.
If the interface abstracting $\graph_1$ also fixes its outflow, then one expects the flow of $\graph_2$ to be preserved.
However, this is not true in the presence of cycles.
\techreport{\ftodo{Add figure and counter example.}}

On the other hand, if we preserve the capacity between every pair of node and sink node in the modification of $\graph_1$ to $\graph'_1$, then the flow of all nodes in $\graph_2$ stays the same.
This is a consequence of the semiring properties of the flow domain and the fact that $\intPlus$ and $\intMult$ are continuous.
Moreover, as $\graph_2$ can only route flow into nodes in $\graph_1$ with a non-zero inflow, we only need to preserve the capacity from these \emph{source} nodes.
We capture this idea formally by defining the \emph{flow map} of $\graphIn = (\inflowClass, \graph)$ as the restriction of $\capacity(\graph)$ onto $\graph$'s source-sink pairs as specified by $\inflowClass$:
\begin{equation*}
  \flowmap((\inflowClass, \graph)) \defeq \capacity(\graph)|_{\setcomp{n \in N}{\exists \inflow \in \inflowClass.\; \inflow(n) > 0} \times N^o}.
\end{equation*}

This gives us sufficient technical machinery to define our abstraction of flow graphs.

\begin{definition}[Flow Interface]
  Given a flow graph $(\inflowClass, \graph) \in \graphInDom$, its \emph{flow interface} is the tuple consisting of its inflow equivalence class, the join of all its node labels, and its flow map:
  \begin{align*}
    \interfaceFn(\graphIn) &\defeq 
    \paren{\inflowClass, \nodelabelJoin_{n \in \graphIn} \nodelabelling(n), \flowmap(\graphIn)} \quad \text{where } \graphIn = (\inflowClass, \graph) \text{ and } \graph = (N, N^o, \nodelabelling, \edgelabel).
  \end{align*}
  The set of all flow interfaces is $\interfaces \defeq \setcomp{\interfaceFn(\graphIn)}{\graphIn \in \graphInDom}$.
  The denotation of a flow interface is a set of flow graphs defined as $\denotation{\interface} \defeq \setcomp{\graphIn \in \graphInDom}{\interfaceFn(\graphIn) = \interface}$.
\end{definition}

\begin{example}
  The interface of $\graphIn_1 = ([\inflow_1]_{\graph_1}, \graph_1)$ from \rF{fig-flow-graph-decomp} is $\interface_1 = (\set{\inflow_1}, \nodelabelBot, \flow_1)$ where $\flow_1$ is the flow map that sets $\set{(n_1, n_3), (n_1, n_5), (n_2, n_6)}$ to $1$ and every other pair to $0$ (assuming all nodes are labeled with $\nodelabelBot$).
  If we now modified $\graphIn_1$ to $\graphIn'_1$ by removing the edge $(n_1, n_4)$ and adding an edge $(n_1, n_5)$, then $\graphIn'_1$ will also be in the denotation $\denotation{\interface_1}$.
  Note that $\graphIn'_1$ still composes with the flow graph on $\graph_2$, and the flow at all nodes in $\graph_2$ is the same in this composition.
\end{example}

If we look at a flow interface $\interface = (\inflowClass, \nodelabel, \flow)$ in the rely-guarantee setting, a flow graph satisfying $\interface$ relies on getting some inflow in $\inflowClass$ from its context, and guarantees the flow map $\flow$ to its context.
The third component, $\nodelabel$, provides a summary of the information contained in the nodes of a graph, which will be useful, for example, to reason about the contents of data structures.
We identify the domain of an interface with the domain of the graphs it abstracts, i.e. $\dom((\inflowClass, \nodelabel, \flow)) = \dom(\inflowClass)$.

Note that both the inflow class and flow map of a flow interface
$\interface$ fully determine the domains of the graphs being described
by the interface. That is, $\interface$ does not abstract from the
identity of the nodes in the graphs $\denotation{\interface}$, which
is intentional as it allows us to define separation logic predicates
in terms of flow interfaces that have precise semantics, without
encoding specific traversal patterns of the underlying data structure
and without restricting sharing as other methods do.

If we consider the equivalence relation on flow graphs induced by having the same flow interface, then we have the required property that this relation is a congruence on $\graphInComp$.

\begin{lemma}
  \label{lem-interface-composition-congruence}
  If $\graphIn_1, \graphIn'_1 \in \denotation{\interface_1}$, $\graphIn_2, \graphIn'_2 \in \denotation{\interface_2}$ and $\graphIn_1 \graphInComp \graphIn_2 \in \denotation{\interface}$, then $\graphIn'_1 \graphInComp \graphIn'_2 \in \denotation{\interface}$.
\end{lemma}

We can now lift flow graph composition to flow interfaces, denoted $\interface = \interface_1 \intComp \interface_2$, as follows:
\begin{align*}
  \interface = \interface_1 \intComp \interface_2 \defiff{}& \exists \graphIn \in \denotation{\interface}, \graphIn_1 \in \denotation{\interface_1}, \graphIn_2 \in \denotation{\interface_2}.\; \graphIn = \graphIn_1 \graphInComp
             \graphIn_2.
\end{align*}
Note that $\intComp$ is also a partial function for the domains of the graphs may overlap, or the inflows may not be compatible.
The identity flow interface $\interfaceEmpty$ is defined by
$\interfaceEmpty \defeq (\set{\inflowEmpty}, \nodelabelBot, \flowEmpty)$, where $\flowEmpty$ is the flow map on an empty domain.
% Does that imply both components empty? No.
A simple corollary of \rL{lem-interface-composition-congruence} is
that the flow interface algebra $(\interfaces, \intComp,
\interfaceEmpty)$ is also a separation algebra.

As we saw in \rSc{sec-expressivity}, to express properties of different data structures we need to use abstractions of flow graphs that additionally satisfy certain conditions on the flow at each node.
To formalize this, we assume the node condition is specified by a predicate
$\goodCondition(n, \inflow, \nodelabel, \flow)$ which takes a node $n$, the inflow $\inflow$ of the singleton flow graph containing $n$, the node label $\nodelabel$ of $n$, and the edge function $\flow$ of that singleton graph as arguments.
The denotation of a flow interface $\interface$ with respect to a good
condition $\goodCondition$ is defined as
\begin{align*}
  \denotation{\interface}_{\goodCondition} & \defeq \setcomp{(\inflowClass, (N, N^o, \nodelabelling, \edgelabel)) \in \denotation{\interface}}{\forall n \in N.\; \goodCondition(n, \inflow_n, \nodelabelling(n), \edgelabel_n)}
\end{align*}
where $\inflow_n \defeq {} \set{n \goesto \netflow(\inflowClass, \graph)(n)}$ and $\edgelabel_n \defeq {} \setcomp{(n, n') \goesto \edgelabel(n, n')}{n' \in N \cup N^o, \edgelabel(n, n') \neq 0}$.
Note that the interface of the singleton flow graph containing $n$ is $(\set{\inflow_n}, \nodelabelling(n), \ite(\inflow_n(n) = 0, \emptyFn, \edgelabel_n))$, but we choose the above formulation for notational convenience.

The following lemma tells us that imposing a good condition on flow interfaces is essentially reducing the domain of flow graphs to those that satisfy the good condition.
In other words, the equivalence on flow graphs induced by the denotation of interfaces with respect to a good condition is also a congruence on $\graphInComp$.

\begin{lemma}
  \label{lem-good-quotient-congruence}
   For all good conditions $\goodCondition$ and interfaces $\interface = \interface_1 \intComp \interface_2$, if $\graphIn_1 \in \denotation{\interface_1}_{\goodCondition}$ and $\graphIn_2 \in \denotation{\interface_2}_{\goodCondition}$ then $\graphIn_1 \graphInComp \graphIn_2$ is defined and in $\denotation{\interface}_{\goodCondition}$.
\end{lemma}

When reasoning about programs, we will sometimes need to modify a flow graph region in a way that will increase the set of inflows.
For instance, if we were to add a node $n$ to a cyclic list, then the equivalence class of inflows will now also contain inflows that have $n$ as a source.
As the flow map is defined from all source nodes, this modification will also change the flow map.
However, as long as the flow map from the existing sources is preserved, then the modified flow graph should still be able to compose with any context.
Formally, we say an interface $(\inflowClass, \nodelabel,
\flow)$ is \emph{contextually extended} by $(\inflowClass', \nodelabel',
\flow')$, written $(\inflowClass, \nodelabel, \flow) \intLessEquiv
(\inflowClass', \nodelabel', \flow')$, if and only if $\inflowClass \subseteq
\inflowClass'$ and
\begin{align*}
  & \forall n \in \dom_1(\flow).\; \exists \inflow \in \inflowClass.\; \inflow(n) \neq 0 \impl \flow(n, \cdot) = \flow'(n, \cdot).
  % & \forall n \in \dom_1(\flow).\; \exists \inflow \in \inflowClass.\; \inflow(n) \neq 0 \impl (\lambda x. \flow(n, x)) \intPlus \zerofun = (\lambda x. \flow'(n, x)) \intPlus \zerofun.
\end{align*}
where $\flow(n, \cdot)$ is the function that maps a node $n'$ to $\flow(n, n')$.
The following theorem states that contextual extension preserves
composability and is itself preserved
under interface composition.
\begin{theorem}[Replacement]
  \label{thm-replacement}
  If $\interface = \interface_1 \intComp \interface_2$, and $\interface_1 \intLessEquiv \interface'_1$, then there exists $\interface' = \interface'_1 \intComp \interface_2$ such that $\interface \intLessEquiv \interface'$.
\end{theorem}

%%% Local Variables:
%%% mode: latex
%%% TeX-master: "main"
%%% End:

\section{A Concurrent Separation Logic with Flow Interfaces}
\label{sec-logic}

We next show how to extend a concurrent separation logic with flow
interfaces. We use RGSep for this purpose as it
suffices to prove the correctness of interesting programs, yet is
relatively simple compared to other CSL flavors. However, many
concurrent separation logics, including RGSep, are parametric in the
separation algebra over which the semantics of programs is defined. So
the technical development in this section can be easily transferred to
these other logics.

RGSep is parametric in the program states (states must form a
separation algebra), the language of the assertions (the original
paper used the standard variant of separation logic), and the basic
commands of the programming language (their semantics need to satisfy
a \emph{locality} property to obtain the frame rule of separation
logic). For the most part, we piggyback on the development
in~\cite{viktor-thesis}. We denote the resulting logic by
$\logic$. Once we have shown that $\logic$ is a valid instantiation of
RGSep, the soundness of all RGSep proof rules immediately carries over
to $\logic$. The only extra work that we will need to do before we can
apply the logic is to derive a few generic lemmas for
proving entailments between assertions that involve flow interfaces.

\subsection{Reasoning about the Actual Heap}
\label{sec-model-dirty-regions-etc}

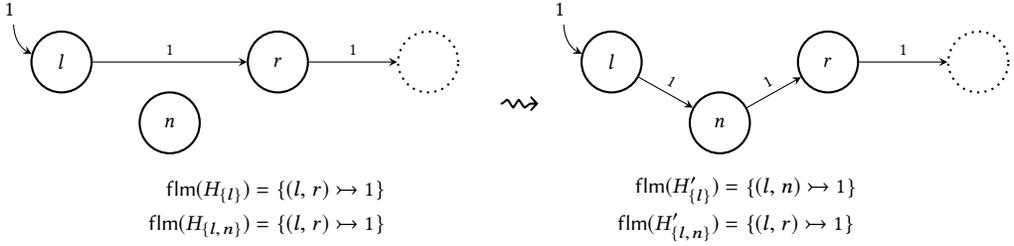
\begin{figure}[t]
  \centering
  \begin{tikzpicture}[>=stealth, scale=0.8, every node/.style={scale=0.8}, font=\footnotesize]
    \tikzstyle{gnode}=[circle, draw=black, thick, minimum size=1cm, font=\normalsize]

    \def\xsep{2}
    \def\ysep{-1}
    \def\xshift{6.5}

    % Before
    \node[gnode] (l) {$l$};
    \node[gnode] (r) at ($(l) + (1.8*\xsep, 0)$) {$r$};
    \node[gnode] (n) at ($(l) + (.9*\xsep, \ysep)$) {$n$};
    \node[inflow] (l-in) at ($(l.north west) + (-.5, .5)$) {$1$};
    \node[pnode] (x) at ($(r) + (2.5, 0)$) {};

    \draw[edge, bend right] (l-in) to (l);
    \draw[edge] (l) to node[above] {$1$} (r);
    \draw[edge] (r) to node[above] {$1$} (x);

    % After
    \node[gnode, right=\xshift of l] (l') {$l$};
    \node[gnode] (r') at ($(l') + (1.8*\xsep, 0)$) {$r$};
    \node[gnode] (n') at ($(l') + (.9*\xsep, \ysep)$) {$n$};
    \node[inflow] (l'-in) at ($(l'.north west) + (-.5, .5)$) {$1$};
    \node[pnode] (x') at ($(r') + (2.5, 0)$) {};

    \draw[edge, bend right] (l'-in) to (l');
    \draw[edge] (l') to node[sloped, above] {$1$} (n');
    \draw[edge] (n') to node[sloped, above] {$1$} (r');
    \draw[edge] (r') to node[above] {$1$} (x');

    % Squiggly arrow
    \node[phantomNode, font=\huge] at ($($(x)!0.5!(l')$) + (0, 0.7*\ysep)$) {$\rightsquigarrow$};
  \end{tikzpicture}
  \begin{minipage}{.40\linewidth}
    \footnotesize
    \begin{align*}
      \flowmap(\graphIn_{\set{l}}) &= \set{(l, r) \goesto 1} \\
      \flowmap(\graphIn_{\set{l, n}}) &= \set{(l, r) \goesto 1}
    \end{align*}
  \end{minipage}%
  \begin{minipage}{.05\linewidth}
    $\quad$
  \end{minipage}%
  \begin{minipage}{.40\linewidth}
    \footnotesize
    \begin{align*}
      \flowmap(\graphIn'_{\set{l}}) &= \set{(l, n) \goesto 1} \\
      \flowmap(\graphIn'_{\set{l, n}}) &= \set{(l, r) \goesto 1}
    \end{align*}
  \end{minipage}
  \caption{Inserting a new node $n$ into a list $\graphIn$ between existing
    nodes $l$ and $r$ obtaining a new list $\graphIn'$. Edges are labeled
    with edge labels for path counting, and the nodes are labeled with their names. The flow maps of certain subgraphs before and after the modification are shown below.}
  \label{fig-list-insert}
\end{figure}

We first discuss the insert procedure on a singly-linked list as an example to motivate the design of our logic.
The insert procedure first traverses the list to find two adjacent nodes $l$ and $r$, with some appropriate properties, between which to insert the given node $n$.
It then sets $n$'s \code{next} field to point to $r$, and finally swings $l$'s \code{next} pointer to $n$.
A pictorial representation of the crucial step in this procedure is shown in the top portion of \rF{fig-list-insert}.

Using flow graphs directly as our program state, as presented thus far, poses three key challenges.
Firstly, a key assumption SL makes on the programming language is that commands are \emph{local}~\cite{SL2}: informally, that if a command runs successfully on a state $\sigma_1$ and we run it on a composed state $\sigma_1 \cdot \sigma_2$ then it leaves $\sigma_2$ unchanged.
However, even a basic operation on a flow graph, for instance changing the edge $(l, r)$ to an edge $(l, n)$, can potentially change the flows of all downstream nodes and is hence not local.
We have seen that if we can find a region around the modification where the flow map is preserved, then the change will not affect any nodes outside the region.
But which region do we choose as the footprint of a given command?
Secondly, the programs that we wish to verify may temporarily violate the good condition during an update or a maintenance operation.
For example, at the state on the left of \rF{fig-list-insert} the new node $n$ has no paths from the root, which violates the good condition describing a list.
Thus we also need a way to describe these intermediate states.
Finally and most immediately, the programs that we wish to verify operate on a program state that are more akin to standard heap representations than to flow graphs.
We thus need a way to abstract the heap to a flow graph.

Our solution to all three problems is based on using the product of heaps and flow graphs as program states.
The heap component represents the concrete state, and the standard commands such as heap allocation or mutation affect only the heap.
Thus, these standard commands inherit locality from the standard semantics.
The graph component is a ghost state and serves as an abstraction of the heap.
Since standard commands only modify the heap, this abstraction may not be up-to-date with the heap.
We add a new ghost command to \emph{sync} a region of the graph with the current state of the heap, i.e. change the graph component to one that abstracts the heap component (using some abstraction as described below).
We say the resulting region of the state is \emph{in sync}.
This $\syncCommand$ command takes as an argument the new interface of
the heap region to sync,
and is only successful if the new interface contextually extends the old
interface of that region. 
For instance, in \rF{fig-list-insert}, we cannot sync the region $\set{l}$ as $l$'s outgoing edges have changed, we would have to sync the region $\set{l, n}$ (note that $n$ is not a source node in this subgraph).
Recall that \rT{thm-replacement} guarantees that this new region will still compose with the old context, making $\syncCommand$ a local command and solving our first problem.

The fact that the graph portion of the state is not up-to-date with the heap also gives us a way to solve the problem of temporary violations of the good condition.
When executing a series of commands, we delay syncing the graph portion of the state until the point when the good condition has been re-established.
Thus, the graph component can be thought of as an abstraction of the heap at the last time it was in a good state.
In the concurrent setting, one may wonder if we might expose such \emph{dirty} regions (where the graph is not in sync with the heap) to other threads; however, if a thread is breaking some data structure invariants in a region, then it must be locked or marked off from interference from other threads in some manner.
We add a \emph{dirty predicate} $\bracket{\phi}_\interface$ to our logic to describe a region where the heap component satisfies $\phi$ but the flow graph component is potentially out of sync and satisfies the interface $\interface$.
So to verify the list insert procedure, we would use a dirty predicate to describe $l$ and $n$ until we have modified both edges and are ready to use $\syncCommand$.

Finally, to solve the problem of relating the concrete heap to the abstract graph, we make use of the good condition.
We define $\goodCondition$ by means of an SL predicate $\gamma$ that not only describes the good condition on the flow in the graph component, but also describes the relation between a good node in the graph and its representation on the heap.
For instance, here is the $\gamma$ that we use to describe a singly-linked list:
\begin{align*}
 \gamma(n, \inflow, \_, \flow) &\defeq \exists n'.\; n \mapsto n' \land \inflow(n) = 1 \land \flow = \ite(n' = \mathit{null}, \emptyFn, \set{(n, n') \goesto 1}).
\end{align*}
Note that this is essentially $\goodCondition_{\m{ls}}$ from
\rSc{sec-expressivity} with the additional constraint that the edge
label in the graph depends on $n'$, the next pointer of $n$ on the heap.
We will show how to define a good graph predicate $\graphPred(\interface)$ using such a $\gamma$ to ensure that the heap and the graph are in sync.

\subsection{Program States}
\label{sec-logic-states}

We model the heap $h$ as usual in standard separation logic. That is,
the heap is a partial map from addresses $\addrs$ to values
$\values$. Values include
addresses, as well as the flow and node label domains and some auxiliary values for representing ghost state related to flow interfaces:
\[\values \supseteq  \addrs \cup \intDom \cup \nodelabelDom \cup
  (\addrs \pto \intDom) \cup \inflowDom \cup (\addrs
    \times \addrs \pto \intDom)  \cup
    \interfaces\]

A state augments the heap with an auxiliary ghost state in the form of a
flow graph $\graphIn$ that maintains an abstract view of the
heap. Each graph node $n \in \dom(\graphIn)$ is an abstraction of
potentially multiple heap nodes, one of which is chosen as the
representative $n$. That is, we require
$\dom(\graphIn) \subseteq \dom(h)$. Also, we have a node map
$\nodeMap$ that labels each heap node with the graph node that
abstracts it. This map is partial, i.e., we allow heap nodes that are
not described by the flow graph.

\iffalse
For memory safety, we
require that the nodes in $\graphIn$ range over all addresses except
for the special $\nullVal$ address.
\fi

\begin{definition}
  \label{def-state}
  A \emph{state}
  $\sigma = (h, \graphIn, \nodeMap) \in \states$
  consists of a heap $h \colon \addrs \pto \values$, a flow interface
  graph $\graphIn \in \graphInDom$ such that $\dom(\graphIn) \subseteq \dom(h)$, and a node map
  $\nodeMap \colon \dom(h) \pto \dom(\graphIn)$ such that $\nodeMap|_{\dom(r) \cap \dom(\graphIn)} = (\lambda x. x)$.
\end{definition}

The composition operator on states is a partial function $\sigma_1 \cdot
\sigma_2$, defined as
\begin{align*}
  (h_1, \graphIn_1, \nodeMap_1) \cdot (h_2, \graphIn_2,
  \nodeMap_2) &\defeq \begin{cases}
    (h_1 \cdot h_2, \graphIn_1 \graphInComp
    \graphIn_2, \nodeMap_1 \uplus \nodeMap_2)
    & \text{if } h_1 \cdot h_2 \text{ and } \graphIn_1 \graphInComp
    \graphIn_2 \text{ defined } \\
    \text{undefined} & \text{otherwise}.
  \end{cases}
\end{align*}
Here, $h_1 \cdot h_2$ denotes the standard disjoint union of
heaps. The empty state $\stateEmpty$ is defined as
$\stateEmpty \defeq (\heapEmpty, \graphInEmpty,
\nodeMapEmpty)$ where $\heapEmpty$ and $\nodeMapEmpty$ are
the empty heap and empty node map, respectively.
The fact that states form a separation algebra easily follows
from~Theorem~\ref{thm-flow-graphs-separation-algebra}.

\begin{theorem}
  The structure $(\states, \cdot, \stateEmpty)$ is a separation algebra.
\end{theorem}

For simplicity we do not use the variable as resource
model~\cite{Bornat:2006:VRS:1706629.1706801} for the treatment of
stack variables. Instead, we follow~\cite{viktor-thesis} and assume
that all stack variables are local to a thread, and that all global
variables are read-only. Stack variables are allocated on the local
heap and global variables on the shared heap. We assume that each
program variable $\code{x}$ is located at a fixed address
$\code{\&x}$.  In all assertions occurring in program specifications
throughout the rest of the paper we have an implicit
$\code{\&x} \mapsto x$ for every program variable $\code{x}$ in scope.

\subsection{Assertions}

The syntax of assertions is summarized in
Fig.~\ref{fig-assertion-syntax}. We assume an infinite set of
(implicitly sorted) logical variables $\vars$. The meta variable $v$ stands
for a variable of an arbitrary sort and the meta variable $x$ for a
variable of a sort denoting addresses. The sorts of other variables
are as expected. Assertions are built from
sorted terms -- one sort per semantic domain of interest. The terms
$T_\intVar$ denote flow domain values, 
$T_\nodelabel$ node label values, etc. For inflow terms $T_\inflow$,
the constant $\epsilon$ denotes the empty inflow $\inflowEmpty$,
% $T.T$ denotes function extension,
$T_\inflow \intPlus T_\inflow$ denotes the lifting of $\intPlus$ to partial functions\footnote{$(f \intPlus g)(x) \defeq \ite(x \in \dom(f), \ite(x \in \dom(g), f(x) + g(x), f(x)), \ite(x \in \dom(g), g(x), \text{undefined}))$}, and the constant $\zerofun$ denotes the total
inflow that maps all addresses to $0$. Terms for flow maps are
constructed similarly.

\begin{figure}
  \centering
  \begin{align*}
     T & \Coloneqq x \mid (x_1, x_2) \mid T_{\intVar} \mid T_{\nodelabel} \mid T_{\inflow} \mid T_{\inflowClass} 
         \mid T_{\flow} \mid T_{\interface}
         \quad T_{\nodelabel} \Coloneqq \nodelabel \mid \nodelabelDom \mid T_\nodelabel \nodelabelJoin T_\nodelabel
    \quad v, x, \intVar, \inflow, \inflowClass, \flow, \nodelabel, \interface \in \vars \\
    T_{\intVar} & \Coloneqq \intVar \mid \intDom \mid T_{\inflow}(x) \mid T_{\flow}(x_1, x_2) \mid T_{\intVar} \intPlus
                  T_{\intVar} \mid T_{\intVar} \intMult T_{\intVar} 
                  \qquad
                  T_{\inflow} \Coloneqq \inflow \mid \mapEmpty \mid \set{x \goesto T_d, \dotsc}
                %\mid T_{\inflow} . T_{\inflow}
                \mid T_{\inflow} \intPlus T_{\inflow} \mid \zerofun \\
    T_{\inflowClass} & \Coloneqq \inflowClass \mid \set{T_\inflow}    \quad                        
    T_{\flow} \Coloneqq \flow \mid \mapEmpty \mid \set{(x_1, x_2) \goesto
                       T_d, \dotsc} %\mid T_{\flow} . T_{\flow}
                \mid T_{\flow} \intPlus T_{\flow} \mid \zerofun \quad
                 T_{\interface} \Coloneqq \interface \mid (T_\inflowClass, T_\nodelabel, T_\flow)
                  \mid T_\interface \intComp T_\interface \\
   P & \Coloneqq T = T \mid T_\intVar \intLeq T_\intVar \mid T_\nodelabel
       \nodelabelLeq T_\nodelabel \mid x \in T_\interface \mid (x_1, x_2)
       \in T_f \mid T_\inflow \in T_\inflowClass \mid 
        T_\interface \intLessEquiv T_\interface \\
    %\psi & \Coloneqq P \mid \emp \mid \true \mid \false \mid x \mapsto E \mid \psi * \psi \mid \psi \land \psi \mid \psi \lor \psi \dots \\
    \phi & \Coloneqq P \mid \emp \mid \true \mid x \mapsto T \mid \graphPred(\interface) \mid
           [\phi]_{\interface} \mid \phi * \phi \mid \phi \magicwand
           \phi \mid \phi \land \phi \mid \exists v.\, \phi \mid \neg \phi
  \end{align*}
  \caption{Syntax of $\logic$ terms and assertions.}
  \label{fig-assertion-syntax}
\end{figure}
\techreport{
\begin{figure}
  \centering
  \begin{align*}
    \begin{array}{rl c l l}
      & \phi[\_] 
      & \defeq &
      \exists v.\, \phi[v] & \text{ $v$ fresh in $\phi$}\\[.4em]
      & \phi[\inflowClassOfInt{\interface}]
      &\defeq &
      \exists \inflowClass.\, \interface = (\inflowClass, \_, \_) \land \phi[\inflowClass] \\
      & & & \multicolumn{2}{l}{\text{where $\inflowClass$ fresh in $\phi$, similarly for
          $\phi[\nodelabelOfInt{\interface}]$ and
          $\phi[\flowmapOfInt{\interface}]$}}\\[.4em]
      & \phi[\inflowClassOfInt{\interface}(x)] & \defeq
      & \exists \inflow .\;  \inflow \in \inflowClassOfInt{\interface} \land \phi[\inflow(x)] \\[.4em]
      & \phi_1 \sepincl \phi_2
      & \defeq &
      (\phi_1 * \mathit{true}) \land \phi_2 \\[.4em]
      & \nodePred(x, \interface) &
      \defeq &
      \graphPred(\interface) \land \set{x \goesto \_} \in \inflowClassOfInt{\interface} \\[.4em]
      & \interface \intEquiv \interface'
      & \defeq &
      \interface \intLessEquiv \interface' \land \interface' \intLessEquiv \interface
      % & \interface_1 \intIncl \interface & \defeq & \interface \in \interface_1 \intComp \_ \\[.4em]
      % & \nodelabel_1 \nodelabelLeq \nodelabel & \defeq & \nodelabel = \nodelabel_1 \nodelabelJoin \nodelabel
      % \\[.4em]
    \end{array}
  \end{align*}
  \caption{Some syntactic shorthands for $\logic$ assertions.}
  \label{fig-syntactic-shorthands}
\end{figure}
}

The atomic pure predicates $P$ include equalities, domain membership
tests for interfaces, flow maps, and inflow equivalence classes, as well as contextual extension of flow interfaces. The
actual separation logic assertions $\phi$ include the standard SL assertions
as well as the extended assertions $\graphPred(\interface)$ and
$[\phi]_\interface$ related to flow interfaces.

The good graph predicate $\graphPred(\interface)$ describes a heap
region that is abstracted by a good flow graph satisfying
the flow interface $\interface$. The semantics of the logic is
parametric in what constitutes a good graph. For simplicity, we 
consider only one type of good graph predicate, i.e., all instances of good
graph predicates imply the same data structure
invariant. \techreport{The logic can be easily extended to support
  more than one such predicate or, if we were to extend higher-order
  separation logic~\cite{Jung:2015:IMI:2676726.2676980,
    DBLP:conf/esop/Krebbers0BJDB17, Dodds:2016:VCS:2866613.2818638},
  the parametrization could be pushed into the assertion language
  itself.}

We also need to permit certain regions of the heap that are under
modification to be not good.
These regions are described using the dirty predicate
$[\phi]_\interface$.
This describes a state where the flow graph satisfies $\interface$ but the heap satisfies $\phi$.

We use some syntactic shorthands for $\logic$ assertions.
For instance, $\nodePred(x, \interface)$ denotes a singleton good graph $\graphPred(\interface)$ containing $x$, and the separating inclusion operator $\phi_1 \sepincl \phi_2$ defined as $(\phi_1 * \mathit{true}) \land \phi_2$.
For an interface $\interface = (\inflowClass, \nodelabel, \flow)$ we write $\inflowClassOfInt{\interface}, \nodelabelOfInt{\nodelabel}$ and $\flowmapOfInt{\interface}$ for $\inflowClass, \nodelabel$ and $\flow$.
\moreless{
These, and some additional syntactic shorthands for $\logic$ assertions, are defined in \rF{fig-syntactic-shorthands}.
}{
The rest will be defined as and when they are used.
}

\iffalse
\begin{figure}[t]
  \centering
  \begin{align*}
    \begin{array}{rl rl rl}
      \denotation{v}_\interp & \defeq \interp(v) & \dots\\
      \denotation{\epsilon_\inflow}_\interp & \defeq \inflowEmpty &
      \denotation{\zerofun_\inflow}_\interp & \defeq \lambda n \in \addrs.\, 0
      \denotation{\set{x \goesto \intVar}}_\interp & \defeq \set{\denotation{x}_\interp
                                             \goesto \denotation{\intVar}_\interp}\\
      \denotation{\epsilon_\flow}_\interp & \defeq \flowEmpty &
      \denotation{\zerofun_\flow}_\interp & \defeq \lambda (n,n') \in
                                            \addrs \times \addrs.\, 0 &
      \denotation{\set{(x_1,x_2) \goesto \intVar}}_\interp & \defeq
                                                             \set{(\denotation{x_1}_\interp, \denotation{x_2}_\interp)
                                             \goesto
                                                             \denotation{\intVar}_\interp}\\
      \denotation{(\inflow, \nodelabel, \flow)}_\interp &
      \defeq (\denotation{\inflow}_\interp, \denotation{\nodelabel}_\interp, \denotation{\flow}_\interp)
    \end{array}
  \end{align*}
  \caption{Semantics of terms.}
  \label{fig-terms-semantics}
\end{figure}
\fi

\begin{figure}[t]
  \centering
  \begin{align*}
    \begin{array}{rl c l}
      %(h, \graphIn, \nodeMap), \interp & \models T_1 = T_2 & \iff & \denotation{T_1}_{\interp} = \denotation{T_2}_{\interp} \\[.4em]
      %
      (h, \graphIn, \nodeMap), \interp & \models T \in T_{\interface,\flow} & \iff & \denotation{T}_{\interp} \in \dom(\denotation{T_{\interface,\flow}}_{\interp}) \\[.4em]
      (h, \graphIn, \nodeMap), \interp &\models  T_{\inflow} \in T_{\inflowClass}  & \iff & \denotation{T_{\inflow}}_\interp \in \denotation{T_{\inflowClass}}_\interp\\[.4em]
      (h, \graphIn, \nodeMap), \interp &\models  T_{\interface} \intLessEquiv T'_{\interface} & \iff & \denotation{T_{\interface}}_\interp \intLessEquiv \denotation{T'_{\interface}}_\interp \\[.4em]
      (h, \graphIn, \nodeMap), \interp & \models \emp & \iff & h = \heapEmpty \land \graphIn = \graphInEmpty \\[.4em]
      (h, \graphIn, \nodeMap), \interp & \models x \mapsto T & \iff &
                                                                     \dom(h) = \{\denotation{x}_\interp\}
                                                                     \land h(\denotation{x}_\interp) = \denotation{T}_{\interp}  \land \graphIn = \graphInEmpty \\[.4em]
      (h, \graphIn, \nodeMap), \interp & \models \graphPred(\interface)
                                                          & \iff &
                                                                   \dom(h) = \dom(r) \land \graphIn \in \denotation{\denotation{\interface}_\interp}_{\goodCondition_{h, r}} \\[.4em]
      (h, \graphIn, \nodeMap), \interp & \models [\phi]_{\interface} &
                                                                      \iff
                                                                 &
                                                                 \dom(h) = \dom(r) \land \graphIn \in \denotation{\denotation{\interface}_\interp} \land \exists \graphIn', \nodeMap'.\; (h, \graphIn', \nodeMap'), \interp \models \phi \\[.4em]
      \sigma, \interp & \models \phi_1 * \phi_2 & \iff & \exists \sigma_1, \sigma_2.\; (\sigma = \sigma_1 \cdot \sigma_2) \land (\sigma_1, \interp \models \phi_1) \land (\sigma_2, \interp \models \phi_2) \\
      %
      %\sigma, \interp & \models \phi_1 \magicwand \phi_2 & \iff & \forall \sigma_1, \sigma_2.\; (\sigma \cdot \sigma_1 = \sigma_2) \land (\sigma_1, \interp \models \phi_1) \impl (\sigma_2, \interp \models \phi_2) \\[.4em]
      & & \dots
    \end{array}
  \end{align*}
  \caption{Semantics of assertions.}
  \label{fig-assertion-semantics}
\end{figure}

Terms are interpreted with respect to an interpretation
$\interp\colon \vars \to \values$ that maps the logical variables to
values. For a variable $v \in \vars$, we write $\denotation{v}_\interp$ to mean
$\interp(v)$ and also extend this function to terms $T$. We omit
the definition of this partial function as it is straightforward.
Note that some terms, for instance interface composition $T_\interface
\intComp T_\interface$, may be undefined; we define the semantics of
any atom containing an undefined term to be false.
The semantics of assertions $\phi$ is given by the satisfaction relation
$\sigma, \interp \models \phi$ and we denote the induced entailment
relation by $\phi_1 \models \phi_2$. The satisfaction relation is defined in
Fig.~\ref{fig-assertion-semantics}. Again, we omit some of the obvious
cases. Most of the cases are straightforward. In particular, the
standard SL assertions have their standard semantics. We only discuss
the cases for $\graphPred(\interface)$ and $[\phi]_\interface$.

The predicate $\graphPred(\interface)$ describes a region whose heap is abstracted by a good flow graph.
We tie the concrete representation in the heap $h$ to the
flow graph $\graphIn$ of the state using the good node
condition, which we assume is defined by an SL predicate
$\gamma(x, \inflow, \nodelabel, \flow)$. This predicate specifies the
heap representation of a node $x$ in the graph as well as any
invariants that $x$ must satisfy on $\netflow(\graphIn)(x)$.
We restrict $\gamma$ to standard SL assertions, i.e., $\gamma$ is not
allowed to include occurrences of the predicates
$\graphPred(\interface)$ and $[\phi]_\interface$.

\iffalse
Moreover, we require
that $\gamma$ satisfies the following property:
\[\gamma(x, \_, \nodelabel, \flow) \land \gamma(x, \_, \nodelabel',
  \flow') \models \nodelabel = \nodelabel' \land \flow = \flow'.\]
This property states that the abstraction of a node does not depend on the inflow parameter of $\gamma$, ensuring that the abstraction of a heap to a flow graph is unique.
\fi

The predicate $\gamma$ implicitly defines a good node condition
$\goodCondition_{h, r}$ on flow interfaces:
\[ \goodCondition_{h, r}(n, \inflow, \nodelabel, \flow)
  \defiff (h|_{r^{-1}(n)}, \graphInEmpty, \nodeMapEmpty), \emptyFn \models \gamma(n, \inflow, \nodelabel, \flow)
\]
The semantics of $\graphPred(\interface)$ then uses this condition to
tie the flow graph $\graphIn$ to the heap $h$ and ensures
that each node in the graph satisfies its local invariant on the
flow. The semantics of $[\phi]_\interface$ is slightly more
complicated. It states that the current
flow graph $\graphIn$ satisfies $\interface$ and there is some flow interface
graph $\graphIn'$ that satisfies $\phi$ in the
current heap.
%The side condition that the domain of the heap coincides
%with the domain of the current node map $\nodeMap$ is important to
%ensure that the predicate has precise semantics, even if $\phi$ itself
%is not precise.

A precise assertion is one that if it
holds for any substate, then it holds for a unique substate.
To use our new predicates in actions to describe interference, we require them to be precise.

%An SL assertion is called precise if for all states
%$\sigma,\sigma_1,\sigma_2$ and interpretations $\interp$, $\sigma

\begin{lemma}
  The graph predicates $\graphPred(\interface)$ and $[\phi]_\interface$ are precise.
\end{lemma}

% \subsection{Specifications}

% These are the same definitions as in \cite{viktor-thesis}, but included here for easy reference.

% \paragraph{Actions}
% Semantics of an action $\phi_1 \rightsquigarrow \phi_2$:
% \begin{equation}
%   \label{eqn-action-semantics}
%   [\phi_1 \rightsquigarrow \phi_2] \defeq \setcomp{(\sigma_0 \cdot \sigma_1, \sigma_0 \cdot \sigma_2) \in \states \times \states}{\sigma_1 \models \phi_1 \land \sigma_2 \models \phi_2}
% \end{equation}

% \todo{Modify this definition to match RGSep definition:}
% \begin{definition}
%   Given an assertion $\phi$ and an action $R$, we say $\stable{\phi}{R}$ if $\sigma \models \phi$ and $(\sigma, \sigma') \in R$ implies $\sigma' \models \phi$.
% \end{definition}

% \begin{definition}
%   We say $\rgTriple{R}{\graph}{\phi}{c}{\psi}$ if for every state $\sigma \models \phi$ and execution of \code{c} on $\sigma$ under the environment $R$,
%   \begin{itemize}
%   \item the execution does not abort,
%   \item every program action satisfies the guarantee $\graph$, and
%   \item the final state satisfies the postcondition, $\sigma' \models \psi$.
%   \end{itemize}
% \end{definition}

\subsection{Programming Language and Semantics}

The programming language is very similar to the one used in RGSep and
its semantics is also mostly identical to the one given in Section 3.2 of
\cite{viktor-thesis}:
\begin{equation*}
  \begin{array}{r l}
    C \in \cmds \Coloneqq & \skipCommand \mid c \mid C_1; C_2 \mid C_1 + C_2 \mid C^* \mid \tuple{C} \mid C \| C \\
    c \Coloneqq & \code{x :| } \phi(\code{x}) \mid \syncCommand(\code{I})
            \mid \markCommand(\code{x}, \code{y}) \mid
            \unmarkCommand(\code{x}) \mid \dots
  \end{array}
\end{equation*}
The standard commands include the empty command ($\skipCommand$), basic
commands, sequential composition, non-deterministic choice,
looping, atomic commands, and parallel composition. Apart from the standard basic
commands, we add several special ghost commands to the language, which
we describe next. Their axiomatic specifications are given in
Fig.~\ref{fig-ghost-command-specs}.

\begin{figure}[t]
  \centering
  \footnotesize
    \begin{align*}
      & \annot{\code{\&x} \mapsto \_ * \exists y.\; \phi(y)}\, \code{x :| }\phi(\code{x})\,
        \annot{\code{\&x} \mapsto x * \phi(x)}
        &&
        \annot{
        \interface \intLessEquiv \interface' \land
        {[\graphPred(\interface')]}_\interface}\,
        \code{\syncCommand(I')}\,
        \annot{\graphPred(\interface')} \\[.2em]
      & \annot{x \mapsto v *
        [\phi]_\interface \land y \in \interface}\,
        \code{\markCommand(x,y)}\,
        \annot{{[x \mapsto v * \phi]}_\interface}
        &&
        \annot{x = y \land x \mapsto v} \,
        \code{\markCommand(x,y)}\,
        \annot{
        [x \mapsto v]_{(\set{\set{x \goesto 0}},\nodelabelBot, \epsilon)}} \\[.2em]
      & \annot{{[x \mapsto v * \phi]}_\interface \land x \notin
        \inflowClassOfInt{\interface}}\,
        \code{\unmarkCommand(x)}\,
        \annot{x \mapsto v * [\phi]_\interface}
        &&
        \annot{{[x \mapsto v]}_\interface \land
           \inflowClassOfInt{\interface} =
        \set{\set{x \goesto 0}}}\,
        \code{\unmarkCommand(x)}\,
        \annot{x \mapsto v}
  \end{align*}
 \caption{Specifications of ghost commands}
  \label{fig-ghost-command-specs}
\end{figure}

\techreport{
\begin{figure}
  \begin{align*}
  (\code{x :| } \phi(\code{x}), (h, \graphIn, \nodeMap)) &\reduces{}{}
  \begin{cases}
    (\skipCommand, (h[l \to v], \graphIn, \nodeMap)) & \text{if }
    h(\&\code{x}) = l \land \phi(x) \text{ is precise} \land {} \\
    & \quad \exists i.\,  (h, \graphIn, \nodeMap), i[x \to v] \models \phi(x) * \true \\
    \abort & \text{otherwise.}
  \end{cases}\\
  (\syncCommand(\code{I'}), (h_1, \graphIn_1, \nodeMap_1) \cdot \sigma_2) &\reduces{}{}
  \begin{cases}
    (\skipCommand, (h_1, \graphIn'_1, \nodeMap'_1) \cdot \sigma_2) &
    \text{if }
    \exists \interface, \interface', i.\; \interface \intLessEquiv \interface' \land \sigma_2 \models \&\code{I'} \mapsto \interface' * \true \land {} \\
    & \quad (h_1, \graphIn_1, \nodeMap_1) \models \bracket{\true}_{\interface} \land (h_1, \graphIn'_1, \nodeMap'_1) \models \graphPred(\interface') \\
    \abort & \text{otherwise.}
  \end{cases}\\
    (\markCommand(\code{x}, \code{y}), (h, \graphIn, \nodeMap)) & \reduces{}{}
  \begin{cases}
    (\skipCommand, (h, \graphIn, \nodeMap[l \goesto n]))
    & \text{if }
    l = h(\& \code{x}) \in \dom(h)
    \setminus \dom(\graphIn) \land {} \\
    & \quad n = h(\&\code{y}) \in \dom(\graphIn) \\
    (\skipCommand, (h, \graphIn', \nodeMap[n \goesto n]))
    & \text{if } n = h(\& \code{x}) = h(\&\code{y}) \in \dom(h)
    \setminus \dom(\graphIn) \land {} \\
    & \quad \exists \graphIn_n \in \denotation{(\set{\set{n \goesto 0}}, \nodelabelBot, \emptyFn)}.\; \graphIn' = \graphIn \graphInComp \graphIn_n \\
    \abort & \text{otherwise.}
  \end{cases}\\
  (\unmarkCommand(\code{x}), (h, \graphIn, \nodeMap)) & \reduces{}{}
  \begin{cases}
    (\skipCommand, (h, \graphIn, \nodeMap|_{\dom(\nodeMap) \setminus \set{l}}))
    & \text{if } l = h(\& \code{x}) \land l \in \dom(r) \setminus \dom(\graphIn)\\
    (\skipCommand, (h, \graphIn', \nodeMap|_{\dom(\nodeMap) \setminus \set{n}}))
    & \text{if } n = h(\& \code{x}) \land r^{-1}(n) = \set{n} \land {} \\
    & \quad \exists \graphIn_n \in \denotation{(\set{\set{n \goesto 0}}, \_, \_)}.\; \graphIn = \graphIn' \graphInComp \graphIn_n \\
    \abort & \text{otherwise.}
  \end{cases}
  \end{align*}
  \caption{Semantics of ghost commands.}
  \label{fig-ghost-command-semantics}
\end{figure}
}

The \emph{wishful assignment} \code{x :|} $\phi(\code{x})$ is a ghost
command that allows us to bring a witness to an existentially
quantified assertion onto the stack. This helps us to get a handle on the updated
flow interface of a modified region so that we can pass it to the $\syncCommand$.
The special ghost command $\syncCommand$ is used to bring a
flow graph region back to sync with the heap.  The $\markCommand$
command\footnote{not to be confused with the notion of marking logically
deleted nodes in Harris' list} is used to formally associate a
(potentially new) heap node with the graph node that abstracts it.
This is useful in cases where many heap nodes are abstracted by the
same graph node.  Similarly, $\unmarkCommand$ removes the association
between a heap and a graph node.
%$\markCommand$ and $\unmarkCommand$ commands add, respectively, remove a node from the node map and flow graph.
Each command has
two successful cases, one where the node being (un)marked is the
representative of a graph node, and one where it is represented
by another node.

\moreless{
The operational semantics of commands $C$ in RGSep are given by a
reduction relation
\[\reduces{}{} \; \subseteq \; (\cmds \times \states) \times ((\cmds \times
  \states) \uplus \set{\abort}).\]
The semantics of the new ghost commands are shown in \rF{fig-ghost-command-semantics}.
}{
The semantics of the new ghost commands are given in the companion technical report~\cite{techreport}.
}
The semantics of all other commands are inherited
from~\cite{viktor-thesis}. That is, these commands only affect the
heap component of the state and leave the flow graph and
node map components unchanged.
If a standard command $C$ modifies an unmarked heap location, then we can use the standard SL specification to reason about it.
If it modifies a heap location that is marked, then we can lift $C$'s SL specification to an $\logic$ specification using the following rule
\begin{align*}
  \infer{\hoareTriple{\bracket{\psi_1}_{\interface}}{C}{\bracket{\psi_2}_{\interface}}}
  {\hoareTriple{\psi_1}{C}{\psi_2}} % \tag{\sc{Lift}}\label{rule-lift-sl}
\end{align*}
where $\psi_{\_}$ denotes an assertion that does not contain the predicates $\graphPred(\interface)$ and $[\phi]_\interface$.
We shall see generic lemmas to convert a good graph node into a dirty region in \rSc{sec-logic-lemmas}.

\moreless{
Each RGSep command must be \emph{local}. A command $C$ is local if for all
$C', \sigma, \sigma', \sigma_1$, and $\sigma_2$:
\begin{enumerate}
\item If $(C, \sigma_1 \cdot \sigma_2) \reduces{}{} (C', \sigma')$,
  then either $(C,
  \sigma_1) \reduces{}{} \abort$, or there exists $\sigma'_1$ such
  that $\sigma' =
  \sigma'_1 \cdot \sigma_2$ and  $(C, \sigma_1) \reduces{}{} (C', \sigma'_1)$.
\item If $(C, \sigma) \reduces{}{} \abort$ and $\sigma =
  \sigma_1 \cdot \sigma_2$, then $(C, \sigma_1) \reduces{}{} \abort$.
\end{enumerate}

\begin{theorem}
  The commands \code{:|}, $\syncCommand$, $\markCommand$,
  $\unmarkCommand$ are local and their semantics satisfies the
  specifications given in Fig.~\ref{fig-ghost-command-specs}.
\end{theorem}

}{
We show in~\cite{techreport} that the new ghost commands are local.
}
Our logic $\logic$ thus satisfies all the requirements for a correct
instantiation of $\rgsep$.

% We also add semantics for galloc
% \begin{equation*}
%   \infer{(E_1 = \mallocCommand(E_2), (h, \inflowClass, \graph, \nodeMap)) \reduces{}{} (\skipCommand, (h \cdot \set{l \goesto \_}, \inflowClass, \graph, \nodeMap \uplus \set{l \goesto n}))}
%   {l \not\in \dom(h) & n = \denotation{E_2}}
% \end{equation*}

\subsection{Proving Entailments}
\label{sec-logic-lemmas}

\begin{figure}
  \begin{align}
    \graphPred(\interface) \land x \in \interface 
    \;\; &\models \;\;
    \exists \interface_1, \interface_2.\; \nodePred(x, \interface_1) * \graphPred(\interface_2) \land \interface = \interface_1 \intComp \interface_2
    \tag{\sc{Decomp}}\label{rule-decomp} \\
    (\graphPred(\interface_1) * \true) \land \graphPred(\interface)
    \;\; &\models \;\;
    \exists \interface_2.\; \graphPred(\interface_1) * \graphPred(\interface_2)
    \tag{\sc{GrDecomp}}\label{rule-gr-decomp} \\
    \bracket{Q_1(\interface'_1)}_{\interface_1} * \bracket{Q_2(\interface'_2)}_{\interface_2} \land \interface' = \interface'_1 \intComp \interface'_2
    \;\; &\models \;\;
    \bracket{Q_1(\interface'_1) * Q_2(\interface'_2)}_{\interface_1 \intComp \interface_2}
    \tag{\sc{Comp}}\label{rule-comp} \\
    \graphPred(\interface_1) * \graphPred(\interface_2)
    \;\; &\models \;\;
    \graphPred(\interface_1 \intComp \interface_2)
    \tag{\sc{GrComp}}\label{rule-gr-comp} \\
    \paren{\graphPred(\interface_1) * \graphPred(\interface_2)} \land \paren{\nodePred(x, \interface_x) * \graphPred(\interface_3)}
    \;\; &\models \;\;
    \exists \interface_4.\; \graphPred(\interface_1) * \nodePred(x, \interface_x) * \graphPred(\interface_4)
    \tag{\sc{Disj}}\label{rule-disj} \\
     {} \land \nodelabelOfInt{\interface_x} \not\nodelabelLeq \nodelabelOfInt{\interface_1} \;\; & \nonumber \\
    % End Disj
    \nodePred(x, \interface)
    \;\; &\equiv \;\;
    \exists \inflow, \nodelabel, \flow.\; \bracket{\gamma(x, \inflow, \nodelabel, \flow)}_{\interface} \land
            \inflow = \set{x \goesto \_}
    \tag{\sc{Conc}}\label{rule-conc} \\
    & \quad \quad \land \interface = (\set{\inflow}, \nodelabel, \ite(\inflow(x) = 0, \emptyFn, \flow))  \nonumber \\
    % End Conc
    \bracket{\nodePred(x, \interface')}_{\interface}
    \;\; &\equiv \;\;
    \exists \inflow, \nodelabel, \flow.\; \bracket{\gamma(x, \inflow, \nodelabel, \flow)}_{\interface} \land \inflow = \set{x \goesto \_}
    \tag{\sc{Abs}}\label{rule-abs} \\
    & \quad \quad \land \interface' = (\set{\inflow}, \nodelabel, \ite(\inflow(x) = 0, \emptyFn, \flow)) \;\; & \nonumber \\
    % End Abs
    \bracket{true}_\interface \land \bracket{true}_{\interface'}
    \;\; &\models \;\;
    \interface = \interface'
    \tag{\sc{Uniq}}\label{rule-uniq} \\
    \interface' = \interface \intComp (\set{\set{x \goesto 0}}, \_, \_) \land \inflow \in \interface
    \;\; &\models \;\;
    \exists \inflow' \in \inflowClassOfInt{\interface'}.\; \inflow \intPlus \zerofun = \inflow' \intPlus \zerofun
    \tag{\sc{AddIn}}\label{rule-add-in} \\
    \interface' = \interface \intComp (\_, \_, \emptyFn) \land \flowmapOfInt{\interface} = \emptyFn
    \;\; &\models \;\;
    \flowmapOfInt{\interface'} = \emptyFn
    \tag{\sc{AddF}}\label{rule-add-f} \\
    \interface_1 \intLessEquiv \interface'_1
    \;\; &\models \;\;
    (\interface_1 \intComp \interface_2) \intLessEquiv (\interface'_1 \intComp \interface_2)
    \tag{\sc{Repl}}\label{rule-repl} \\
    \interface \intLessEquiv \interface' \land \inflow \in \interface
    \;\; &\models \;\;
    \inflow \in \inflowClassOfInt{\interface'}
    \tag{\sc{ReplIn}}\label{rule-repl-in} \\
    \interface \intLessEquiv \interface' \land \flowmapOfInt{\interface} = \emptyFn
    \;\; &\models \;\;
     \flowmapOfInt{\interface'} = \emptyFn
    \tag{\sc{ReplF}}\label{rule-repl-f} \\
    \interface = \interface_1 \intComp \interface_2 \land (x, y) \in \flowmapOfInt{\interface_1} \land \flowmapOfInt{\interface} = \emptyFn
    \;\; &\models \;\;
    y \in \interface_2
    \tag{\sc{Step}}\label{rule-step}
%    \interface = \interface_1 \intComp \_ \land d \intLeq \inflowClassOfInt{\interface}(x) \land x \in \interface_1
%    \;\; &\models \;\;
%    d \intLeq \inflowClassOfInt{\interface_1}(x)
%    \tag{\sc{FlowProj}}\label{rule-flow-proj} \\
\techreport{
  \\
    \interface = \interface_1 \intComp \_ \land x \in \interface_1
    \;\; &\models \;\;
    x \in \interface
    \tag{\sc{Climb}}\label{rule-climb} \\
    % \graphPred(\interface_1) \sepincl \graphPred(\interface)
    % \;\; &\models \;\;
    % \exists \interface_2.\; \graphPred(\interface_1) * \graphPred(\interface_2) \land \interface = \interface_1 \intComp \interface_2
    % \tag{\sc{Incl2Star}}\label{rule-incl2star} \\
    % \graphPred(\interface) \land x \in \interface
    % \;\; &\models \;\;
    % \nodePred(x, \interface_x) \sepincl \graphPred(\interface)
    % \tag{\sc{Decomp}$\sepincl$}\label{rule-decomp-sepincl} \\
    \nodePred(x, \interface_x) \sepincl \graphPred(\interface) \land (x, y) \in \flowmapOfInt{\interface_x} \land \flowmapOfInt{\interface} = \emptyFn
    \;\; &\models \;\;
    \exists \interface_y.\; (\nodePred(x, \interface_x) * \nodePred(y, \interface_y)) \sepincl \graphPred(\interface)
    \tag{\sc{Step}$\sepincl$}\label{rule-step-sepincl}
}
  \end{align}
  \caption{Generic lemmas for proving entailments}
  \label{fig-proof-rules}
\end{figure}

Finally, we show several generic lemmas for proving
entailments about SL assertions with flow interfaces. We only focus on
lemmas that involve flow interfaces and omit lemmas about simple
pure assertions involving inflows and flow maps. We also omit lemmas
about the algebraic properties of flow interface composition, such as
associativity. The lemmas are shown in Fig.\ref{fig-proof-rules}.
\techreport{\todo{If Climb is not used, remove it from POPL. Check if we ever use $\bracket{\phi}_{\interface} \equiv \bracket{\bracket{\phi}_{\interface}}_{\interface}$}}
% as well as some derived rules in Fig.~\ref{fig-derived-proof-rules}.
% Each lemma is a valid entailment.

For example, the lemma (\ref{rule-decomp}) implies that we can pull an arbitrary node $\nodePred(x, \interface_1)$ from a good graph $\graphPred(\interface)$ containing $x$, obtaining a flow interface constraint $\interface = \interface_1 \intComp \interface_2$ that remembers that the two parts composed to $\interface$.
Lemma (\ref{rule-comp}) uses $Q(\interface)$ to denote either $\graphPred(\interface)$ or $\bracket{\phi}_{\interface}$ for some $\phi$ and tells us that we can combine two dirty regions if the interfaces they are expecting ($\interface'_1$ and $\interface'_2$) are compatible.
The lemma (\ref{rule-gr-comp}) states that we can always abstract two good graphs
$\graphPred(\interface_1) * \graphPred(\interface_2)$ to a larger good
graph $\graphPred(\interface)$ that satisfies the composite interface
$\interface = \interface_1 \intComp \interface_2$. The validity of this entailment follows from \rL{lem-good-quotient-congruence}.
(\ref{rule-conc}) and (\ref{rule-abs}) allow us to shift our reasoning from graph nodes to their heap representation and back.
(\ref{rule-add-in}) and (\ref{rule-add-f}) can be used to reason about adding new nodes to the state.

One advantage of our logic is that a lemma like (\ref{rule-gr-comp}) is very general and can be used to compose two lists, two sorted lists, or two min-heaps.
The good condition will ensure, for instance, that the composition of sorted lists is itself sorted.

%%% Local Variables:
%%% mode: latex
%%% TeX-master: "main"
%%% End:

\section{Application 1: The Harris List}
\label{sec-harris}

\begin{figure}[t]
  \begin{lstlisting}
procedure insert() { $\annot{\boxed{\Phi}}$
  // Where: $\textcolor{blue}{\Phi \defeq \exists \interface.\; \graphPred(\interface) \land \varphi(\interface), \quad \varphi(\interface) \defeq \exists \inflow \in \interface^{\inflow}.\; \inflow \intPlus \zerofun = \set{\mainListHead \goesto (1, 0)} \intPlus \set{\freeListHead \goesto (0, 1)} \intPlus \zerofun \land \interface^\flow = \emptyFn}$@\label{code-insert-pre}@
  var l := mh; $\annot{\boxed{\nodePred(l, \interface_l) \sepincl \Phi} \land l = \mainListHead}$@\label{code-insert-decomp}@
  var r := getUnmarked(l.next);
  while (r != null && nondet()) { $\annot{\boxed{(\nodePred(l, \interface_l) * \nodePred(r, \interface_{r})) \sepincl \Phi}}$@\label{code-insert-pre-loop}@
    l := r; r := getUnmarked(l.next);@\label{code-insert-read-next}@ $\annot{\boxed{(\nodePred(l, \interface_l)) \sepincl \Phi} \land \paren{r \neq null \impl \boxed{(\nodePred(l, \interface_l) * \nodePred(r, \interface_{r})) \sepincl \Phi}}}$
  }
  if (!isMarked(r)) { $\annot{\boxed{\nodePred(l, \interface_l) \sepincl \Phi} \land \neg M(r)}$@\label{code-insert-if-not-marked}@
    var n := new Node(r, null);@\label{code-insert-new}@ $\annot{\boxed{\nodePred(l, \interface_l) \sepincl \Phi} * n \mapsto r, \nullVal \land \neg M(r)}$@\label{code-insert-alloc}@
    mark(n, n); $\annot{\boxed{\nodePred(l, \interface_l) \sepincl \Phi} * \bracket{n \mapsto r, \nullVal}_{\interface_n} \land \neg M(r)}$@\label{code-insert-mark}@ // Where: $\textcolor{blue}{\interface_n := (\set{n \goesto (0, 0)}, \bot, \emptyFn)}$
    atomic {@\label{code-insert-cas-start}@    // CAS(l.next, r, n)
      if (l.next == r) {
        $\annot{\paren{\bracket{l \mapsto r, \_}_{\interface_l} \land \nodePred(l, \interface_l)} * \bracket{n \mapsto r,  \nullVal}_{\interface_n} * \graphPred(\interface_2) \land \nodelabelOfInt{\interface_l} = \unmarked \land \interface' = \interface_l \intComp \interface_n \intComp \interface_2 \land \varphi(\interface')}$@\label{code-insert-atomic-pre}@
        l.next := n;
        $\annot{\bracket{\nodePred(l, \interface'_l) * \nodePred(n, \interface'_n)}_{\interface_1} * \graphPred(\interface_2) \land \nodelabelOfInt{\interface'_l} = \unmarked \land \nodelabelOfInt{\interface'_n} = \unmarked \land \interface' = \interface_l \intComp \interface_n \intComp \interface_2 \land \varphi(\interface') \land \interface_1 = \interface_l \intComp \interface_n \intEquiv \interface'_l \intComp \interface'_n}$@\label{code-insert-sync-pre}@
        var I$'_1$ :| $\bracket{\graphPred(\interface'_1)}_{\interface_1} \land \set{l \goesto \_, n \goesto \_} \in \inflowClassOfInt{\interface_1}$;@\label{code-insert-wish}@
        sync(I$'_1$);@\label{code-insert-sync}@ $\annot{\nodePred(l, \interface'_l) * \nodePred(n, \interface'_n) * \graphPred(\interface_2) \land \nodelabelOfInt{\interface'_l} = \unmarked \land \nodelabelOfInt{\interface'_n} = \unmarked \land \interface'' = \interface'_l \intComp \interface'_n \intComp \interface_2 \land \varphi(\interface'')}$
        var b := true; $\annot{(\nodePred(l, \interface'_l) * \nodePred(l, \interface'_l)) \sepincl \Phi \land \nodelabelOfInt{\interface'_l} = \unmarked \land \nodelabelOfInt{\interface'_n} = \unmarked \land b}$@\label{code-insert-post}@
      } else var b := false; $\annot{\boxed{\nodePred(l, \interface_l)
          \sepincl \Phi} * \bracket{n \mapsto r, \nullVal}_{\interface_n}
        \land \neg b}$
    } ...@\label{code-insert-cas-end}@ // If CAS failed, unmark and free n, and call insert() again
  }
}
\end{lstlisting}
    % if (!b) {$\annot{\boxed{\nodePred(l, \interface_l) \sepincl \Phi} * \bracket{n \mapsto r, 0}_{\interface_n}}$
    %   unmark(n);
    %   free(n);
    %   insert(...);
    % }
\caption{A fragment from the insert procedure on the Harris list.}
\label{fig-harris-insert}
\end{figure}

We now demonstrate our flow framework by describing and verifying the Harris list.
For simplicity of presentation, we consider a data structure without keys and abstract the algorithm to one that non-deterministically chooses where to insert a node or which node to delete.
This is without loss of generality, as we are only proving memory
safety and absence of memory leaks.
Our proof can be extended to prove these properties on the full Harris list, and our framework can also prove the key invariants needed to prove linearizability.

We now describe the flows and good conditions we use for this proof.
To describe the structural properties, we use the product of two path-counting flows, one counting paths from the head of the main list $\mainListHead$ and one from the head of the free list $\freeListHead$.
To reason about marking, we use the node domain $\Nat \cup \set{\unmarked, \top}$ under the ordering where $\unmarked$ is the smallest element, $\top$ is the largest element, and all other elements are unordered.
We label unmarked nodes with $\unmarked$, and marked nodes with the thread ID $t \in \Nat$ of the thread that marked the node.
This is in order to enforce that only the thread that marks a node may link it to the free list, a property needed to prove that the free list is acyclic.

Our good condition is specified by the following predicate:
\begin{align*}
  & \gamma(n, \inflow, \nodelabel, \flow) \defeq \exists n', n''.\; n \mapsto n', n'' \land \nodelabel \neq \top \land (M(n') \iff \nodelabel \neq \unmarked) \land (0, 0) < \inflow(n) \leq (1, 1) \\
  & \quad \land (\inflow(n) \geq (0, 1) \impl \nodelabel \neq \unmarked) \land (n = \freeListTail \impl \inflow(n) \geq (0, 1)) \land (\inflow(n) \leq (1, 0) \impl n'' = \nullVal) \\
  & \quad \land \flow = \ite(u(n') = \nullVal, \emptyFn, \set{(n, u(n')) \goesto (1, 0)}) \intPlus \ite(n'' = \nullVal, \emptyFn, \set{(n, n'') \goesto (0, 1)}).
\end{align*}
% NOTE: right now this doesn't enforce that all free list nodes have next pointers.
This condition expresses that every node $n$ is a heap cell containing two pointers $n'$ (for the \code{next} field) and $n''$ (for the \code{fnext} field).
$n$ is either unmarked or marked with a thread ID
($\nodelabel \neq \top$), but only if the value $n'$
of field \code{next} has its mark bit set (encoded using the predicate $M$).
We next use the path-counting flows to say that $n$ has exactly 1 path on \code{next} edges from $\mainListHead$, on \code{fnext} edges from $\freeListHead$, or both.
This enforces that all nodes are in at least one of the two lists, and this is how we establish absence of memory leaks.
The next few conjuncts say that all nodes in the free list are marked ($\inflow \geq (0, 1) \impl \nodelabel \neq \unmarked$), that $\freeListTail$ is a node in the free list, and that main list nodes have no \code{fnext} edges.
The final line describes the edges: $n$ has a
\code{next} edge (encoded with the edge label $(1, 0)$) to $u(n')$
(the unmarked version of $n'$, i.e. the actual pointer obtained from $n'$ by masking the mark bit), and an \code{fnext} edge (label $(0, 1)$) to $n''$, but only if they are not null.

We use the global data structure invariant $\Phi$ shown in line \ref{code-insert-pre} of \rF{fig-harris-insert}, which gives the appropriate non-zero inflow to $\freeListHead$ and $\mainListHead$.
It is easy to see that with good condition $\gamma$, $\Phi$ describes a structure satisfying properties \ref{it-two-lists} to \ref{it-marked} of the Harris list from \rSc{sec-overview}.

% TODO: continue making actions.
% \begin{align}
%   \nodePred(l, \interface_l) \land \nodelabelOfInt{\interface_l} = \unmarked
%   & \rightsquigarrow
%   \begin{aligned}
%     & \nodePred(l, \interface'_l) * \nodePred(n, \interface_n) \land \nodelabelOfInt{\interface'_l} = \unmarked \land \nodelabelOfInt{\interface_n} = \unmarked \\
%     & \land \exists \interface_1.\; \interface_l \intLessEquiv \interface_1 \land \interface_1 \in \interface'_l \land \interface_n
%   \end{aligned}
%   \tag{Insert}
% \end{align}

We now describe the $\logic$ proof for the \code{insert} procedure, shown in \rF{fig-harris-insert}.
Variables that are not program variables in the annotations are implicitly existentially quantified.
All the entailments in this proof sketch can be proved with the help of our library of lemmas from \rSc{sec-logic-lemmas}.

The procedure starts in a state satisfying $\Phi$ and sets a variable \code{l} to equal the head of the main list.
Since \code{l} is in the domain of the inflow $\inflowOfInt{\interface}$ we use (\ref{rule-decomp}) to decompose $\Phi$ and get a single node ($\nodePred(l, \interface_l)$) included in the larger graph satisfying $\Phi$ (line \ref{code-insert-decomp}).
We then read \code{l.next}, store the unmarked version in \code{r}, and enter a loop.
At the beginning of the loop (line \ref{code-insert-pre-loop}), since $r \neq \nullVal$, we know by $\gamma$ that \code{l} has a \code{next} edge to $r$.
But since $\Phi$ implies $\interface^\flow = \emptyFn$, we can use \moreless{(\ref{rule-step-sepincl})}{(\ref{rule-gr-decomp}), (\ref{rule-step}), and (\ref{rule-decomp})} to extract the node corresponding to $r$ and obtain the annotation on line \ref{code-insert-pre-loop}.

Inside the loop, we move \code{l} to \code{r} and then again read \code{l.next}, unmark it, and set it to \code{r} (line \ref{code-insert-read-next}).
This dereference is memory safe because we know that we have access permission to \code{l} (by $\nodePred(r, \interface_r)$ on line \ref{code-insert-pre-loop}).
Finally, to establish the annotation on line \ref{code-insert-read-next}, we ``drop'' the node predicate corresponding to the old value of \code{l} and absorb it into $\Phi$.
The second conjunct is derived similar to the annotation on line \ref{code-insert-pre-loop}.
Note that the annotation at the end of the loop implies the annotation at the beginning of the loop if the loop check succeeds.

The loop terminates when the non-deterministic loop condition fails -- presumably at the correct position to insert the new node.
Since we only want to insert new nodes into the main list, we ensure \code{l} is unmarked (line \ref{code-insert-if-not-marked}).
We then create the new node with \code{next} field $r$ (which may equal null) in line \ref{code-insert-alloc}.
As this allocation only modifies the heap, the resulting node is described using a standard SL points-to predicate in the local state of the thread.
To create a corresponding node in the graph we use the $\markCommand$ ghost command, which creates a new graph node with a zero inflow, the bottom node label, and no outgoing edges.
The resulting state is described with a dirty predicate (line \ref{code-insert-mark}) as it is neither in sync nor a good node.

The next step is to swing $l$'s \code{next} pointer from $r$ to $n$ using a CAS operation.
The CAS is expanded into an atomic block (lines \ref{code-insert-cas-start} to \ref{code-insert-cas-end}) in order to show the intermediate proof steps.
Note that inside the atomic block there is no need for separation of shared and local state, as in $\rgsep$ one can reason sequentially about atomic executions.
If the compare portion of the CAS succeeds, then we know that $l$ is unmarked ($\nodelabelOfInt{\interface_l} = \unmarked$) and its next field equals $r$ (we use a dirty predicate around $l$ to express this).
We use (\ref{rule-gr-decomp}), (\ref{rule-comp}), (\ref{rule-add-in}), and (\ref{rule-add-f}) to decompose $\Phi$ at this point and infer that $\interface'$, the global interface extended with $n$, also satisfies the global conditions $\varphi(\interface')$ (line \ref{code-insert-atomic-pre}).
The modification to \code{l.next} in the next line is memory safe since we have access permission to $l$ (by $l \mapsto r, \_$).

Before we can bring the graph abstraction back to sync, we must
establish that the change to the interface of some region is a contextual extension.
As we saw with the example of inserting in a singly-linked list in \rF{fig-list-insert}, we must consider the region containing $\set{l, n}$.
Some pure reasoning about path-counts in this region can be used to infer that $\interface_l \intComp \interface_n \intEquiv \interface'_l \intComp \interface'_n$ (line \ref{code-insert-sync-pre})\footnote{$\interface \intEquiv \interface'$ is shorthand for $\interface \intLessEquiv \interface' \land \interface' \intLessEquiv \interface$.}.
We then use a wishful assignment (line \ref{code-insert-wish}) to take
a snapshot of the new interface $\code{I}'_1$ of the region we wish to sync.

The other fact to prove before syncing the graph is to show that the region under modification is itself a good state.
This is indicated in line \ref{code-insert-sync-pre} by using good node predicates inside the dirty region.
To establish that $n$ is a good node, for instance, we use a new interface $\interface'_n$ which gives it an inflow of $(1, 0)$, a node abstraction of $\unmarked$, and a flow map of one edge to $r$ with label $(1, 0)$.
We then check that the heap representation of $n$ along with this interface satisfies $\gamma$.
Similarly, we check that $l$, with its new edge to $n$ is still a good node.
As the entire dirty region is a good state with a contextually extended interface, we can now use $\syncCommand$ (line \ref{code-insert-sync}) to update the graph.
We can then use the (\ref{rule-repl}), (\ref{rule-repl-in}), (\ref{rule-repl-f}), and (\ref{rule-comp}) rules to establish that the new global interface satisfies the invariant $\varphi(\interface'')$.

The final state (line \ref{code-insert-post}) is once again a good state, which means that we have shown both memory safety and absence of memory leaks of this procedure.
We have omitted the rest of the insert procedure due to space constraints, which frees the node $n$ and restarts if the CAS failed, but note that this can be proved in a similar fashion.
We have also omitted the reasoning about interference by other threads, which is done in RGSep by checking that each intermediate assertion is stable under the action of other threads.
This is a syntactic check that can be done with the help of the entailment lemmas presented in \rF{fig-proof-rules}.
The full proof can be found in \moreless{\rSc{sec-harris-simple-appendix}}{the technical report \cite[\S{}A]{techreport}}.
One can also prove linearizability -- for the Harris list, this requires a technique such as history variables since linearization points are dynamic -- but the invariants of the data structure needed to show linearizability are expressible using flows.

This example shows that reasoning about programs using flow graphs and interfaces is as natural as with inductive predicates, and we can use similar unrolling and abstraction lemmas to reason about traversals of the data structure.
Note that the intermediate assertions would have looked almost
identical if we did not have a free list -- the only place where we
reason about the free list is in the global interface on line
\ref{code-insert-pre} and when we prove the local property $\interface_l \intComp \interface_n \intEquiv \interface'_l \intComp \interface'_n$ on line \ref{code-insert-sync-pre}.

%%% Local Variables: 
%%% mode: latex
%%% TeX-master: "main"
%%% End: 

\section{Application 2: Dictionaries}
\label{sec-dictionaries}

In this section, we use the flow framework to verify a large class of concurrent dictionary implementations.
We base our approach on the \emph{edgeset framework} of~\citet{dennis-keyset} that provides invariants, in terms of reachability properties of sets of keys, for proving linearizability of dictionary operations (search, insert, and delete).
Linearizability means, informally, that each operation appears to happen atomically, i.e. at a single instant in time, and that if operation $o_1$ finishes before $o_2$ begins, then $o_1$ will appear to happen before $o_2$.

By encoding the edgeset framework using flows, we obtain a method of proving linearizability as well as memory safety.
More importantly, we can use the power of flows to encode the data constraints independently of the shape.
Thus our encoding will be data-structure-agnostic, meaning that we can verify \emph{any} dictionary implementation that falls within the edgeset framework.

We first briefly describe the original edgeset framework, then show how to encode it using flows.
We then give an abstract algorithm template and specifications for the dictionary operations that can be instantiated to concrete implementations.
In \moreless{\rSc{sec-btree}}{the technical report~\cite[\S{}B]{techreport}} we describe such an instantiation to a nontrivial implementation based on B+ trees, and show how we can verify it using this framework.

\subsection{The Edgeset Framework}

A dictionary is a key-value store that implements three basic operations: search, insert, and delete.
For simplicity of exposition, we ignore the data values and treat the dictionary as containing only keys.
We refer to a thread seeking to search for, insert, or delete a key $k$ as an operation on $k$, and to $k$ as the operation's query key.
Let $\KS$ be the set of possible keys, e.g., all integers.

The edgeset framework describes certain invariants on the data layout in dictionary implementations and an abstract algorithm that will be correct if the invariants are maintained.
These invariants do not specify the shape of the data structure to be say a tree, list, or hash table, and instead describe properties of an abstract graph representation of the data structure (rather like the flow graph).
The nodes in the graph can represent, for instance, an actual heap node (in the case of a list), an array cell (in the case of a hash table), or even a collection of fields and arrays (in the case of a B-tree).
Nodes are labeled with the set of keys stored at the node (henceforth, the contents of the node).
Each edge is labeled by a set of keys that we call the \emph{edgeset}, and is defined as follows.

When a dictionary operation arrives at a node $n$, the set of query keys for which the operation traverses an edge $(n, n')$ is called the edgeset of $(n, n')$.
For example, in a BST, an operation on $k$ moves from a node to its left child $n_l$ if $k$ is less than the key contained at the node, hence the edgeset of $(n, n_l)$ is $\setcomp{k}{k < n.key}$.
Note that $k$ can be in the edgeset of $(n, n')$ even if $n$ is not reachable from any root; the edgeset is the set of query keys for which an operation would traverse $(n, n')$ assuming it somehow found itself at $n$.

The \emph{pathset} of a path between nodes $n_1$ and $n_2$ is defined as the intersection of edgesets of every edge on the path, and is thus the set of keys for which operations starting at $n_1$ would arrive at $n_2$ assuming neither the path nor the edgesets along that path change.
For example, in a sorted list, the pathset of a path from the head of a list to a node $n$ is equal to the edgeset of the edge leading into $n$ (i.e. the set $\setcomp{k}{n'.key < k}$ where $n'$ is $n$'s predecessor).
With this, we define the \emph{inset} of a node $n$ as the union of the pathset of all paths from the root node to $n$.
\footnote{If there are multiple roots, then each edge has a different edgeset depending on the root, and hence the definitions of pathset and inset also depend on the particular root. The formalism can be extended to handle multiple roots in this manner.}
% {\em Dennis suggestion: remove this sentence. It's not longer needed:} The inset of $n$ is the set of keys $k$ such that $n$ is reachable from some root by an operation looking for $k$.
If we take the inset of a node $n$, and remove all the keys in the union of edgesets of edges leaving $n$, we get the \emph{keyset} of $n$.

There are three desirable conditions on graphs representing dictionary data structures:
\begin{enumerate}[label=(GS\arabic{enumi})]
\item\label{item-gs1} The keysets of two distinct nodes are disjoint\footnote{The original paper required the keysets to partition $\KS$, but we note that this weaker condition is sufficient for linearizability.}.
\item\label{item-gs2} The contents of every node are a subset of the keyset of that node.
\item\label{item-gs3} The edgesets of two distinct edges leaving a node are disjoint.
\end{enumerate}
Intuitively, \ref{item-gs1} and \ref{item-gs2} tell us we can treat the keyset of $n$ as the set of keys that $n$ can potentially contain.
In this case, $k$ is in the inset of $n$ if and only if operations on $k$ pass through $n$, and $k$ is in the keyset of $n$ if and only if operations on $k$ end up at $n$.
\ref{item-gs3} requires that there is a deterministic path that operations follow, which is a desirable property that is true of all data structures in common use.
A dictionary state that satisfies these conditions is called a \emph{good state}.

The Keyset Theorem of~\citet{dennis-keyset} states, informally, that if every atomic operation preserves the good state property and $k$ is in the keyset of a node $n$ at the point when the operation looks for, inserts, or deletes $k$ at $n$, then the algorithm is linearizable.
Intuitively, if $k$ is in the keyset of $n$, then since the keysets are
disjoint we know that no other thread is performing an operation on $k$ at any other node.
And once this operation acquires a lock on $n$ (or establishes exclusive access in another way, e.g. through a compare and swap), we know that operations on $n$ will be atomic.

The challenge with using this framework for a formal proof in an SL-based program logic is that the invariants depend on quantities, like the inset, that are not local.
The inset of a node $n$ depends both on the global root as well as all paths in the data structure from the root to $n$.
% Thus, it is not easy to obtain a compositional proof using these ideas.
We show next how to convert the good state conditions into local properties of nodes using flows.

\subsection{Encoding the Edgeset Framework using Flows}

\begin{figure}
  \centering
  \begin{lstlisting}
procedure dictionaryOp(Key k) { $\annot{\boxed{\Phi}}$@\label{give-up-pre}@
  var c := r; $\annot{\boxed{\nodePred(c, \interface_c) \sepincl \Phi}}$
  while (true) { $\annot{\boxed{\nodePred(c, \interface_c) \sepincl \Phi}}$ @\label{give-up-loop-inv}@
    lock(c); $\annot{\boxed{\nodePred(c, \interface_c) \sepincl \Phi} \land \nodelabelOfInt{\interface_c} = (\_, \set{t})}$
    var n;
    if (inRange(c, k)) { $\annot{\boxed{\nodePred(c, \interface_c) \sepincl \Phi} \land \nodelabelOfInt{\interface_c} = (\_, \set{t}) \land k \in \inflowClassOfInt{\interface_c}(c)}$ @\label{give-up-inRange}@
      n := findNext(c, k); @\label{give-up-findNext}@$\annot{
        & \boxed{\nodePred(c, \interface_c) \sepincl \Phi} \land \nodelabelOfInt{\interface_c} = (\_, \set{t}) \land k \in \inflowClassOfInt{\interface_c}(c) \\
        & \land (n \neq \nullVal \land k \in \flowmapOfInt{\interface_c}(c, n) \lor n = \nullVal \land \forall x \in \flowmapOfInt{\interface_c}.\; k \not\in \flowmapOfInt{\interface_c}(c, x))}$
      if (n == null) break; @\label{give-up-break}@
      $\annot{\boxed{(\nodePred(c, \interface_c) * \nodePred(n, \interface_n)) \sepincl \Phi} \land \nodelabelOfInt{\interface_c} = (\_, \set{t})}$
    } else {
      n := r; @\label{give-up-give-up}@$\annot{\boxed{\nodePred(c, \interface_c) \sepincl \Phi} \land \nodelabelOfInt{\interface_c} = (\_, \set{t}) \land n = r}$
    } $\annot{\boxed{(\nodePred(c, \interface_c) * \nodePred(n, \interface_n)) \sepincl \Phi} \land \nodelabelOfInt{\interface_c} = (\_, \set{t}) \lor \boxed{\nodePred(c, \interface_c) \sepincl \Phi} \land \nodelabelOfInt{\interface_c} = (\_, \set{t}) \land c = n = r}$
    unlock(c);
    c := n; $\annot{\boxed{\nodePred(c, \interface_c) \sepincl \Phi}}$
  } $\annot{\boxed{\nodePred(c, \interface_c) \sepincl \Phi} \land \nodelabelOfInt{\interface_c} = (\_, \set{t}) \land k \in \inflowClassOfInt{\interface_c}(c) \land \forall x \in \flowmapOfInt{\interface_c}.\; k \not\in \flowmapOfInt{\interface_c}(c, x)}$@\label{give-up-decisive-pre}@
  var res := decisiveOp(c, k); $\annot{\boxed{\nodePred(c, \interface_c) \sepincl \Phi} \land \nodelabelOfInt{\interface_c} = (\_, \set{t})}$
  unlock(c); $\annot{\boxed{\nodePred(c, \interface_c) \sepincl \Phi}}$
  return res; $\annot{\boxed{\Phi}}$
}
\end{lstlisting}
  \vspace*{-1em}
  \caption{The give-up template, with proof annotations in our logic.}
  \label{fig-give-up-template}
\end{figure}

\begin{figure}
  \footnotesize
  \begin{align*}
    &\annot{\boxed{\nodePred(c, \interface_c) \sepincl \Phi}}\; \code{lock(c);} \; \annot{\boxed{\nodePred(c, \interface'_c) \sepincl \Phi} \land \nodelabelOfInt{\interface'_c} = (\_, \set{t}) \land \interface_c \intEquiv \interface'_c} \\
    &\annot{\boxed{\nodePred(c, \interface_c) \sepincl \Phi} \land \nodelabelOfInt{\interface_c} = (\_, \set{t})}\; \code{unlock(c); } \; \annot{\boxed{\nodePred(c, \interface'_c) \sepincl \Phi} \land \interface_c \intEquiv \interface'_c} \\
    &\annot{\boxed{\nodePred(c, \interface_c) \sepincl \Phi} \land \nodelabelOfInt{\interface_c} = (\_, \set{t})}\; \code{res := inRange(c, k); } \; \annot{\boxed{\nodePred(c, \interface_c) \sepincl \Phi} \land (res \impl k \in \inflowClassOfInt{\interface_c}(c)) } \\
    &\annot{\boxed{\nodePred(c, \interface_c) \sepincl \Phi} \land
      \nodelabelOfInt{\interface_c} = (\_, \set{t})}\; \code{n :=
      findNext(c, k); } \; \annot{
      & \boxed{\nodePred(c, \interface_c) \sepincl \Phi} \land \big(n \neq \nullVal \land k \in \flowmapOfInt{\interface_c}(c, n) \\
      & \quad \lor n = \nullVal \land \forall x \in \flowmapOfInt{\interface_c}.\; k \not\in \flowmapOfInt{\interface_c}(c, x)\big)} \\
    &\annot{
      & \boxed{\nodePred(c, \interface_c) \sepincl \Phi} \land \nodelabelOfInt{\interface_c} = (C, \set{t}) \land k \in \inflowClassOfInt{\interface_c}(c) \\
      & \land \forall x \in \flowmapOfInt{\interface_c}.\; k \not\in
      \flowmapOfInt{\interface_c}(c, x)}\; \code{res := decisiveOp(c, k); }
    \; \annot{\boxed{\nodePred(c, \interface'_c) \sepincl \Phi} \land \interface_c \intEquiv \interface'_c \land \Psi} \\
    & \text{where } \Psi \defeq
    \begin{cases}
      \nodelabelOfInt{\interface'_c} = (C, \set{t}) \land res \iff k \in C & \text{for \code{member}} \\
      \nodelabelOfInt{\interface'_c} = (C \cup \set{k}, \set{t}) \land res \iff k \not\in C & \text{for \code{insert}} \\
      \nodelabelOfInt{\interface'_c} = (C \setminus \set{k}, \set{t}) \land res \iff k \in C & \text{for \code{delete}} \\
    \end{cases}
  \end{align*}
  \vspace{-.8em}
  \caption{Specifications for helper functions.}
  \label{fig-helper-specs}
\end{figure}

Given a potentially-infinite set of keys, $\KS$, the set of subsets of $\KS$ forms a flow domain: $(2^{\KS}, \subseteq, \cup, \cup, \cap, \emptyset, \KS)$.
The node domain must contain a set of keys to keep track of the contents of each node, but we also wish to reason about locking.
The challenge here is that threads may modify a locked node using a series of atomic operations such that the node does not satisfy the good state conditions in between operations, thus making the Keyset Theorem inapplicable.
We do not want to use the dirty predicate to reason about such nodes, as this will complicate the global shared state invariant.
Instead, we label nodes with elements of $\Nat \uplus
\overline{\Nat}$, where $\overline{\Nat} \defeq
\setcomp{\overline{x}}{x \in \Nat}$. Here, $0$ denotes an unlocked node, $t$ denotes a node locked by thread $t$, and $\overline{t}$ denotes a node locked by thread $t$ whose heap representation is out of sync.
Formally, we use a product of sets of keys and sets of augmented thread IDs as the node domain: $(2^\KS \times 2^{\Nat \uplus \overline{\Nat}}, \subseteq, \cup, (\emptyset, \emptyset))$, where $\subseteq$ and $\cup$ are lifted component-wise.
In the following, all uses of $t$ implicitly assert that the label is
not $\overline{x}$ for any $x \in \Nat$.

\begin{lemma}
  \label{lem-flow-equals-inset}
  For any graph $\graph$, $\capacity(\graph)(n, n')$ is the set of keys $k$ for which there exists a path from $n$ to $n'$, every edge of which contains $k$ in its edge label.
  In particular, if $\inflow(x) = \ite(x = r, \KS, \emptyset)$ is an inflow on $\graph$, then $\netflow(\inflow, \graph)(n)$ is the set of keys $k$ for which such a path exists from $r$ to $n$.
\end{lemma}

If we label each graph edge with its edgeset, then by \rL{lem-flow-equals-inset}, the flow at each node is the inset of that node.
We can encode \ref{item-gs3} easily using the good condition at each node, and by using sets of keys as the node domain, we can also encode \ref{item-gs2}.
We further observe that for graphs with one root, edgesets leaving a node being pairwise disjoint implies that the keysets of every pair of nodes is disjoint\footnote{For graphs with several roots, each edgeset is defined with respect to each root and then this property holds for each root.}.
We can thus use a global data structure invariant that the inflow of the root is $\KS$ and all other nodes $\emptyset$ to obtain \ref{item-gs1}.

This motivates us to use the following good condition
\begin{align*}
   \gamma(x, \inflow, (C, T), \flow) & \defeq \exists t.\; (\gamma_g(x, \inflow, C, t, \flow) \land T = \set{t}\, \lor\, \gamma_b(x, t) \land t \neq 0 \land T = \{\overline{t}\}) \\
   & \quad \land C \subseteq \inflow(x) \land \forall y.\; (C \cap \flow(x, y) = \emptyset \land \forall z.\; \flow(x, y) \cap \flow(x, z) = \emptyset)
\end{align*}
where $\gamma_g$ and $\gamma_b$ are user-specified SL predicates.
$\gamma_g$ is to be instantiated with the heap implementation of a node in sync, and $\gamma_b$ with a description of a node that may not be in sync.
% $\lockedPred$ must additionally satisfy $\lockedPred(t) \land \lockedPred(t') \models t = t'$.
%% Note: we don't need the above because all locking can be done only in the graph component. It just wouldn't be a ghost state anymore.
The global data structure invariant enforces that the graph has no outgoing edges and that only the root gets a non-zero inflow:
\[\Phi \defeq \exists \interface, \inflow.\; \graphPred(\interface) \land \inflow \in \inflowClassOfInt{\interface} \land \inflow \intPlus \zerofun = \set{r \goesto \KS} \intPlus \zerofun \land \flowmapOfInt{\interface} = \emptyFn\]

\begin{lemma}
  The flow graph component of any state satisfying $\Phi$ satisfies the good state conditions \ref{item-gs1} to \ref{item-gs3}.
\end{lemma}

The original edgeset paper~\cite{dennis-keyset} set out three template algorithms, based on different locking disciplines, along with invariants that implied their linearizability.
Of these, we formalize the template based on the \emph{give-up} technique here.
Our method can be easily extended to the lock-coupling template, but we note that most practical algorithms use more fine-grained locking schemes.
For the link technique, the third template, we need to encode the \emph{inreach}, an inductive quantity depending on the keyset and edgesets, viz. $k$ is in the inreach of a node $n$ if $k$ is in the keyset of $n$ or $k$ is in the keyset of $n'$ and for every edge $e$ in the path from $n$ to $n'$, $k$ is in the edgeset of $e$.
Being able to express the inreach would also give a simpler way to prove linearizability of the full Harris list.
We would need to extend our framework to support second-order flows to define the inreach, and we leave this for future work.

%\ftodo{\em Dennis suggests taking a look at the Stratos paper which handles both give-up and link. If you find that too general, then Dennis suggests defining inRange and making the invariant that inRange is always a subset of inset. This is guaranteed for both splits and merges.  For the Harris list, inreach never decreases for any node until the node is freed (Stratos paper has invariants on restructuring operations that may be useful). Finally, Dennis suggests that you prove the Harris algorithm using this machinery.}

We now describe the give-up template algorithm, shown in \rF{fig-give-up-template}.
This template can be used to build implementations of all three dictionary operations by defining the function \code{decisiveOp} as described below.
The basic idea is that every node stores a \emph{range}, which is an under-approximation (i.e. always a subset) of its inset.
An operation on key $k$ proceeds by starting at the root and enters a loop at line \ref{give-up-loop-inv} where it follows edges that contain $k$ in their edgeset.
Between two nodes, the algorithm has no locks on either node, to allow for more parallelism.
Because of this, when it arrives at some node $c$, the first thing it
does after locking it is to check $c$'s range field (the call to
\code{inRange} in line \ref{give-up-inRange}) to ensure that the
dictionary operation is at a node whose inset\footnote{$\inflowClassOfInt{\interface}(x)$ is shorthand for $\inflow(x)$ for some $\inflow \in \inflowClassOfInt{\interface}$.} contains $k$. %, and hence on the path to a node whose keyset contains $k$.
If the check succeeds, then it calls another helper, \code{findNext} at line \ref{give-up-findNext}, that checks if there exists a node $n$ such that $k$ is in the edgeset of $(c, n)$.
If there is no such node, then we know that $k$ must be in the keyset of $c$, and we break from the loop (line \ref{give-up-break}) while holding the lock on $c$.
If the \code{inRange} check fails, then the algorithm gives up and starts again from the root $r$ (line 13)\footnote{We could, in theory, jump to any ancestor of that node. This will require a variable to keep track of the ancestor, and similar reasoning can be used.}.
If the search continues, the algorithm first unlocks $c$ on line 15 before reassigning $c$ to the next node $n$.

When the algorithm breaks out of the loop, $k$ must be in the keyset of $c$ (line \ref{give-up-decisive-pre}), so it calls \code{decisiveOp} to perform the operation on $c$.
It then unlocks $c$, and returns the result of \code{decisiveOp}.

\paragraph{Proof}
The proof annotations for this template (also shown in \rF{fig-give-up-template}) are fairly straightforward, and entailments between them can be derived using the lemmas in \rF{fig-proof-rules}, assuming that the user-provided implementations of the helper functions satisfy the specifications in \rF{fig-helper-specs}.
Note that these specifications are in terms of the entire shared state, even though the helper functions only modify the current node $c$.
This is a limitation of $\rgsep$, and one can obtain local specifications for these functions by switching to a more advanced logic such as \cite{DBLP:conf/popl/Feng09}.
To reason about interference, we use the following actions to specify the modifications to the shared state allowed by a set of thread IDs $T$:
\begin{align*}
  t \in T \land \nodePred(x, (\inflowClass, (C, \set{0}), \flow)) & \rightsquigarrow \nodePred(x, (\inflowClass, (C, T'), \flow)) \land T' \subseteq \{t, \overline{t}\} \tag{Lock} \label{eqn-action-lock} \\
  % t \in T \land \nodePred(x, \inflow, (C, \set{t}), \flow) & \rightsquigarrow \nodePred(x, \inflow, (C, \set{0}), \flow) \tag{Unlock} \label{eqn-action-unlock} \\
  t \in T \land \emp & \rightsquigarrow \nodePred(x, (\set{\set{x \goesto 0}}, (\emptyset, \{\overline{t}\}), \emptyFn)) \tag{Alloc} \label{eqn-action-alloc} \\
  t \in T \land \graphPred(\interface) \land \nodelabelOfInt{\interface} \nodelabelLeq (\_, \{t, \overline{t}\}) & \rightsquigarrow \graphPred(\interface')  \land \nodelabelOfInt{\interface'} \nodelabelLeq (\_, \{0, t, \overline{t}\}) \land \interface \intLessEquiv \interface' \tag{Sync} \label{eqn-action-sync}
\end{align*}
(\ref{eqn-action-lock}) allows a thread $t \in T$ to lock an unlocked node; (\ref{eqn-action-alloc}) allows $t$ to add new nodes with no inflow, contents, or outgoing edges; and (\ref{eqn-action-sync}) allows $t$ to modify a locked region arbitrarily, as long as the new interface contextually extends the old one.
The last action also allows it to unlock nodes it has locked.
The guarantee of thread with id $t_0$ is made up of the above actions with $T = \set{t_0}$, while the rely constitutes the above actions with $T = \Nat \setminus \set{t_0}$.
To complete the proof that the template algorithm is memory safe and preserves the global data structure invariant, one must show the stability of every intermediate assertion in \rF{fig-give-up-template}.
This can be done syntactically using our lemmas, and an example of such a proof can be seen in~\moreless{\rSc{sec-stability-proof}}{\cite[\S{}C]{techreport}}.

\paragraph{Proving Linearizability}
To prove that the template algorithm is linearizable, we adapt the Keyset Theorem to the language of flows.
%By doing this, we explicitly formalize the implicit argument in~\cite{dennis-keyset} used to justify such algorithms.
To show that all atomic operations maintain the good state conditions, we require that every intermediate assertion in our proof implies $\Phi$.
This is true of assertions in our template proof, but also needs to be true in the intermediate assertions used to prove that implementations of helper functions meet their specification.
The condition about the query key $k$ being in the keyset of the node at which a thread performs its operation is captured by the specifications in \rF{fig-helper-specs}.
We additionally require that each dictionary operation only modifies the contents of the global once, and that no other operation modifies the contents.
We thus obtain the following re-statement of the Keyset Theorem:
\begin{theorem}
  \label{thm-give-up-template}
  An implementation of the give-up template from \rF{fig-give-up-template} is memory safe and linearizable if the following conditions hold:
  \begin{enumerate}
  \item The helper functions satisfy the specification in \rF{fig-helper-specs} under the rely and guarantee specified above, with all intermediate assertions implying $\boxed{\Phi} * \true$.
  \item Every execution of \code{decisiveOp} must have at most one call
    to $\syncCommand$ that changes the contents of the flow graph.
  \item All other calls to $\syncCommand$, including those in maintenance operations, do not change the contents of the graph region on which they operate.
  \end{enumerate}
\end{theorem}

Alternatively, our $\logic$ proof has established sufficient invariants to directly prove linearizability.
One way to do this~\cite{DBLP:conf/vmcai/Vafeiadis09} is to use auxiliary variables to track the abstract state of the data structure (we already have this as the set of contents in the global interface) and atomically execute the specification of each operation on this abstract state at the linearization point.
One can then use a write-once variable to store the result of the abstract operation, and at the end of the operation prove that the implementation returns the same value as the specification.

This template algorithm can now be instantiated to any concrete implementation by providing predicates $\gamma_g$ and $\gamma_b$ to describe the heap layout of the data structure and implementing the helper functions such that they satisfy the conditions of \rT{thm-give-up-template}.
If there are any maintenance operations, for example splitting or merging nodes, these must also satisfy the invariants in the theorem under the given rely and guarantee.
An example implementation of this template, the B+ tree, can be seen in~\moreless{\rSc{sec-btree}}{\cite[\S{}B]{techreport}}.

%%% Local Variables:
%%% mode: latex
%%% TeX-master: "main"
%%% End:

\section{Related Work}
\label{sec-related-work}

\paragraph{Abstraction Mechanisms in Separation Logic.}
The prevalent mechanism for abstracting unbounded heap regions in
separation logic is based on inductive predicates defined by
separating conjunctions~\cite{SL2,
  BerdineETAL04DecidableFragmentSeparationLogic,
  CooketALFragmentSepLog, DBLP:conf/cade/BrotherstonDP11,
  DBLP:conf/atva/IosifRV14, DBLP:conf/pldi/PekQM14,
  DBLP:conf/atva/EneaSW15}. For simple inductive data structures whose
implementation follows regular traversal patterns, inductive
predicates can be easier to work with than flow-based
abstractions. Certain abstractions, for instance abstracting a list as
a sequence of values, while possible to encode in flows, would be more
natural with inductive predicates. However, as discussed
in~\rSc{sec-overview}, inductive predicates are often ill-suited for
abstracting concurrent data structures. The reasons include the
dependence on specific traversal patterns in the inductive definitions
of the predicates, and the implied restrictions on expressing sharing
and data structure overlays with separating conjunction.

Ramifications~\cite{DBLP:conf/vstte/MehnertSBS12,
  Hobor:2013:RSD:2429069.2429131} offer an alternative approach to
reasoning about overlaid data structures and data structures with
unrestricted sharing. These approaches express inductively defined
graph abstractions by using overlapping conjunctions of subheaps
instead of separating conjunction.  However, this necessitates
complicated ramification entailments involving magic wand to reason
about updates.  We do not need to use combinations of separating
conjunction and overlapping conjunction or reason about entailments
involving the notorious magic wand.  Instead, we shift the reasoning
to the composition and decomposition of inflows and flow maps.  While
showing that a flow graph satisfies conditions on the flow in general
requires computing fixpoints over the graph, in proofs we only need to
reason about them when a heap region is modified.  These computations
are typically easy because concurrent algorithms modify a bounded number of
nodes at a time in order to minimize interference.  Finally, we also
obtain a uniform, and decoupled, way to reason about both shape and
data properties, yielding abstractions that generalize over a wide
variety of data structures such as in our encoding of the edgeset
framework.

Yet another alternative for abstracting arbitrary graphs in SL is to
use iterated separating conjunction~\cite{SL2,Yang01ShorrWaite,
  DBLP:conf/cav/0001SS16, DBLP:conf/aplas/RaadHVG16}. Similar to flow
interfaces, such abstractions are not tied to specific traversal
patterns and can capture invariants that are expressible as local
conditions on nodes such as that the graph is closed. However, unlike
flow interfaces, iterated separating conjunctions cannot capture
inductive properties of a graph (e.g. that the reachable nodes from
a root form a tree). In essence, flow interfaces occupy a sweet spot
between inductive predicates and iterated separating conjunctions.

\paragraph{Nondeterministic Monoidal Models of Separation Logic.}
Our notion of flow graph composition naturally yields a
nondeterministic monoidal model of SL where separating conjunction $*$
is interpreted as a ternary relation. However, the conventional
meta-theory of SL~\cite{DBLP:conf/lics/CalcagnoOY07,
  Dockins:2009:FLS:1696759.1696777} requires $*$ to be
partial-deterministic. We overcome this mismatch here by defining an
appropriate equivalence relation on flow graphs (or, more precisely,
their inflows) to enforce a functional interpretation of $*$. However,
this solution leads to a slightly stronger model than is strictly
necessary. It also adds some artificial complexity to the logic as the
inflow equivalence classes must be reasoned about at the syntactic
level. We believe that both of these issues can be avoided by
considering a nondeterministic monoidal semantics as the starting
point for the development of the program logic. Such semantics have
been studied in the context of substructural logics, Boolean
BI~\cite{Galmiche2006}, and more recently to obtain a more general
proof theory of propositional abstract
SL~\cite{Hou:2014:PSP:2535838.2535864}. To our knowledge, flow graphs
constitute the first example of a separation algebra with
nondeterministic monoidal structure that has practical applications in
program verification.

\paragraph{Concurrent Separation Logics.}

CSL was introduced by \citet{DBLP:conf/concur/OHearn04}. A recent article by~\citet{DBLP:journals/siglog/BrookesO16} provides a survey of the development of CSLs since
then. Among the many
improvements that have been developed are the idea of fractional
permissions to reason about shared
reads~\cite{Bornat:2005:PAS:1040305.1040327,
  DBLP:conf/vmcai/HeuleLMS13}, combinations of rely/guarantee
reasoning~\cite{DBLP:conf/concur/VafeiadisP07,
  DBLP:conf/esop/FengFS07}, and abstraction mechanisms for reasoning
about different aspects of concurrency such as synchronization
protocols for low-level lock
implementations~\cite{DBLP:conf/esop/NanevskiLSD14}, and atomicity
abstractions~\cite{DBLP:conf/ecoop/Dinsdale-YoungDGPV10,
  DBLP:conf/ecoop/PintoDG14, DBLP:conf/esop/XiongPNG17}. \techreport{We only
discuss the most closely related works and otherwise refer the reader
to the detailed discussion in~\cite{DBLP:journals/siglog/BrookesO16}.

Vafeiadis and Parkinson introduced
$\rgsep$~\cite{DBLP:conf/concur/VafeiadisP07} to marry rely/guarantee
reasoning with separation logic. $\rgsep$ allows the actions of
threads to be described by SL assertions, thereby abstracting from the
specifics of the synchronization mechanism between threads. We chose
$\rgsep$ as the basis for our logic because it provides a
good trade-off between flexibility and simplicity when reasoning about
concurrent data structure implementations. However, the logic also has
limitations. One limitation is that it does not allow the shared state
to be split and \emph{framed out}. This is less of an issue in our
work as the global shared state can always be described by a single
$\graphPred$ predicate and carrying this predicate through all proofs
does not add too much complexity. This limitation of $\rgsep$ was
addressed in local $\rgsep$~\cite{Feng:2009:LRR:1480881.1480922} and
our results can be easily adapted to that logic.

Another limitation of $\rgsep$ is that it lacks abstraction mechanisms
that enable compositional reasoning about data structure clients. A
number of solution to this problem have been proposed, including the
views framework~\cite{DBLP:conf/popl/Dinsdale-YoungBGPY13}, Concurrent
Abstract Predicates (CAP)~\cite{DBLP:conf/ecoop/Dinsdale-YoungDGPV10},
and TaDA~\cite{DBLP:conf/ecoop/PintoDG14}.} The abstractions provided
by \moreless{these logics}{the latter} are orthogonal to
the ones developed here. These logics have been used, e.g., to verify
specifications of concurrent dictionary implementations based on
B-trees and skip lists that enable compositional client
verification~\cite{DBLP:conf/oopsla/PintoDDGW11,
  DBLP:conf/esop/XiongPNG17}. We believe that the proofs developed
in~\cite{DBLP:conf/esop/XiongPNG17} can be further simplified by
introducing flow interfaces as an intermediate abstraction of the
considered data structures.

Higher-order concurrent separation
logic~\cite{Jung:2015:IMI:2676726.2676980,
  DBLP:conf/esop/Krebbers0BJDB17} can express ghost state within the
assertion language of the logic itself. This feature can be used to
eliminate the restriction of $\logic$ that the semantics of the flow
interface predicates is defined on the meta level and that it does not
support nesting of flow interface abstractions (i.e., cases where a
node of a flow graph should abstract from another flow graph contained
in the node). Similarly, higher-order CSL can express complex
linearizability proofs directly without relegating a part of the proof
argument to the meta level.

\paragraph{Invariant Inference.}
There is a large body of work on inferring invariants for
heap-manipulating programs (e.g.,
~\cite{DBLP:journals/toplas/SagivRW02}), including techniques based on
separation logic~\cite{DBLP:conf/tacas/DistefanoOY06,
  DBLP:conf/popl/CalcagnoDOY09, DBLP:conf/vmcai/Vafeiadis10}. Many of
these approaches rely on some form of abstract
interpretation~\cite{CousotCousot77AbstractInterpretation}. We believe
that the least fixpoint characterization of flows in flow interfaces
lends itself well to abstract interpretation techniques.

\section{Conclusion}
\label{sec-conclusion}

We have introduced flow interfaces as a novel approach to the
abstraction of unbounded data structures in separation logic. The
approach avoids several limitations of common solutions to such
abstraction, allows unrestricted sharing and arbitrary traversals of
heap regions, and provides a uniform treatment of data constraints. We
have shown that flow interfaces are particularly well suited for
reasoning about concurrent data structures and that they hold great
promise for developing automated techniques for reasoning about
implementation-agnostic abstractions.

\bibliography{references}

\techreport{
\newpage

%% Appendix
\appendix
\section{The Harris List: Complete Proof Annotations}
\label{sec-harris-simple-appendix}

In the following proof sketches, $t$ is the thread ID of the current thread.
We use an augmented global shared state invariant that also enforces that $\freeListTail$ is a node in the shared state:
\[\Phi \defeq \exists \interface.\; \graphPred(\interface) \land \varphi(\interface), \quad \varphi(\interface) \defeq \exists \inflow \in \interface^{\inflow}.\; \inflow \intPlus \zerofun = \set{\mainListHead \goesto (1, 0)} \intPlus \set{\freeListHead \goesto (0, 1)} \intPlus \zerofun \land \freeListTail \in \interface \land \interface^\flow = \emptyFn\]

\todo{Have a lemma saying $\Phi$ is preserved by Repl?}

\begin{lstlisting}
procedure insert() {
  $\annot{\boxed{\Phi}}$
  var l := mh; $\annot{\boxed{\nodePred(l, \interface_l) \sepincl \Phi} \land l = \mainListHead}$
  var r := getUnmarked(l.next);
  while (r != null && nondet()) { $\annot{\boxed{(\nodePred(l, \interface_l) * \nodePred(r, \interface_{r})) \sepincl \Phi}}$
    l := r; r := getUnmarked(l.next); $\annot{\boxed{(\nodePred(l, \interface_l)) \sepincl \Phi} \land \paren{r \neq null \impl \boxed{(\nodePred(l, \interface_l) * \nodePred(r, \interface_{r})) \sepincl \Phi}}}$
  }
  if (!isMarked(r)) { $\annot{\boxed{\nodePred(l, \interface_l) \sepincl \Phi} \land \neg M(r)}$
    var n := new Node(r, null); $\annot{\boxed{\nodePred(l, \interface_l) \sepincl \Phi} * n \mapsto r, \nullVal \land \neg M(r)}$
    mark(n, n); $\annot{\boxed{\nodePred(l, \interface_l) \sepincl \Phi} * \bracket{n \mapsto r, \nullVal}_{\interface_n} \land \neg M(r)}$ // Where: $\textcolor{blue}{\interface_n := (\set{n \goesto (0, 0)}, \bot, \emptyFn)}$
    atomic {    // CAS(l.next, r, n)
      if (l.next == r) {
        $\annot{\paren{\bracket{l \mapsto r, \_}_{\interface_l} \land \nodePred(l, \interface_l)} \sepincl \Phi * \bracket{n \mapsto r, \nullVal}_{\interface_n} \land \nodelabelOfInt{\interface_l} = \unmarked}$ // @(\ref{rule-conc})@
        $\annot{\paren{\bracket{l \mapsto r, \_}_{\interface_l} \land \nodePred(l, \interface_l)} * \graphPred(\interface_2) * \bracket{n \mapsto r, \nullVal}_{\interface_n} \land \nodelabelOfInt{\interface_l} = \unmarked \land \interface = \interface_l \intComp \interface_2 \land \varphi(\interface)}$ // @(\ref{rule-gr-decomp})@
        $\annot{\paren{\bracket{l \mapsto r, \_}_{\interface_l} \land \nodePred(l, \interface_l)} * \bracket{n \mapsto r, \nullVal}_{\interface_n} * \graphPred(\interface_2) \land \nodelabelOfInt{\interface_l} = \unmarked \land \interface' = \interface_l \intComp \interface_n \intComp \interface_2 \land \varphi(\interface')}$ // @(\ref{rule-comp}), (\ref{rule-add-in}), (\ref{rule-add-f})@
        l.next := n;
        $\annot{\bracket{\nodePred(l, \interface'_l)}_{\interface_l} * \bracket{n \mapsto r, \nullVal}_{\interface_n} * \graphPred(\interface_2) \land \nodelabelOfInt{\interface_l} = \unmarked \land \interface' = \interface_l \intComp \interface_n \intComp \interface_2 \land \varphi(\interface')}$
        $\annot{\bracket{\nodePred(l, \interface'_l) * \nodePred(n, \interface'_n)}_{\interface_1} * \graphPred(\interface_2) \land \nodelabelOfInt{\interface'_l} = \unmarked \land \nodelabelOfInt{\interface'_n} = \unmarked \land \interface' = \interface_l \intComp \interface_n \intComp \interface_2 \land \varphi(\interface') \land \interface_1 = \interface_l \intComp \interface_n \intEquiv \interface'_l \intComp \interface'_n}$
        var I$'_1$ :| $\bracket{\graphPred(\interface'_1)}_{\interface_1} \land \set{l \goesto \_, n \goesto \_} \in \inflowClassOfInt{\interface_1}$;
        sync(I$'_1$);
        $\annot{\nodePred(l, \interface'_l) * \nodePred(n, \interface'_n) * \graphPred(\interface_2) \land \nodelabelOfInt{\interface'_l} = \unmarked \land \nodelabelOfInt{\interface'_n} = \unmarked \land \interface' = \interface_l \intComp \interface_n \intComp \interface_2 \land \varphi(\interface') \land \interface_l \intComp \interface_n \intEquiv \interface'_l \intComp \interface'_n}$
        $\annot{\nodePred(l, \interface'_l) * \nodePred(n, \interface'_n) * \graphPred(\interface_2) \land \nodelabelOfInt{\interface'_l} = \unmarked \land \nodelabelOfInt{\interface'_n} = \unmarked \land \interface'' = \interface'_l \intComp \interface'_n \intComp \interface_2 \land \varphi(\interface'')}$
        var b := true; $\annot{(\nodePred(l, \interface'_l) * \nodePred(l, \interface'_l)) \sepincl \Phi \land \nodelabelOfInt{\interface'_l} = \unmarked \land \nodelabelOfInt{\interface'_n} = \unmarked \land b}$
      } else var b := false; $\annot{\boxed{\nodePred(l, \interface_l)
          \sepincl \Phi} * \bracket{n \mapsto r, \nullVal}_{\interface_n}
        \land \neg b}$
    }
    if (!b) { $\annot{\boxed{\nodePred(l, \interface_l) \sepincl \Phi} * \bracket{n \mapsto r, \_}_{\interface_n} \land \neg b}$
      unmark(n); $\annot{\boxed{\Phi} * n \mapsto r, \_}$
      free(n); $\annot{\boxed{\Phi}}$
      insert(); $\annot{\boxed{\Phi}}$
    }
  }
}

procedure delete() {
  $\annot{\boxed{\Phi}}$
  var l := mh; $\annot{\boxed{\nodePred(l, \interface_l) \sepincl \Phi} \land l = \mainListHead}$
  var r := getUnmarked(l.next);
  while (r != null && nondet()) { $\annot{\boxed{(\nodePred(l, \interface_l) * \nodePred(r, \interface_r)) \sepincl \Phi}}$
    l := r; r := getUnmarked(l.next); $\annot{\boxed{\nodePred(l, \interface_l) \sepincl \Phi} \land (r \neq \nullVal \impl \boxed{(\nodePred(l, \interface_l) * \nodePred(r, \interface_r)) \sepincl \Phi})}$
  }
  if (r == null) return;
  var x := r.next;
  if (isMarked(x)) delete();
  $\annot{\boxed{(\nodePred(l, \interface_l) * \nodePred(r, \interface_r)) \sepincl \Phi} \land \neg M(x)}$
  atomic {  // CAS(r.next, x, getMarked(x))
    if (r.next == x) { $\annot{\paren{\nodePred(l, \interface_l) * \paren{\bracket{r \mapsto x, \_}_{\interface_r} \land \nodePred(r, \interface_r)}} \sepincl \Phi \land \nodelabelOfInt{\interface_r} = \unmarked \land \inflowClassOfInt{\interface_r} = (1, 0)}$
      r.next := getMarked(x); $\annot{\paren{\nodePred(l, \interface_l) * \bracket{r \mapsto x, \_ \land \nodePred(r, \interface'_r)}_{\interface_r}} \sepincl \Phi \land \interface_r \intEquiv \interface'_r \land \nodelabelOfInt{\interface'_r} = t \land \inflowClassOfInt{\interface'_r} = (1, 0)}$
      var I$'_r$ :| $\bracket{\nodePred(r, \interface'_r)}_{\interface_r} \land \set{r \goesto \_} \in \inflowClassOfInt{\interface_r}$;
      sync(I$'_r$); $\annot{\paren{\nodePred(l, \interface_l) * \paren{\bracket{r \mapsto x, \_}_{\_} \land \nodePred(r, \interface'_r)}} \sepincl \Phi \land \nodelabelOfInt{\interface'_r} = t \land \inflowClassOfInt{\interface'_r}(r) = (1, 0)}$
      var b := true;
    } else var b := false;
  }
  if (b) { $\annot{\boxed{\paren{\bracket{r \mapsto x, \_}_{\_} \land \nodePred(r, \interface_r)} \sepincl \Phi} \land l \in \interface \land \nodelabelOfInt{\interface_r} = t \land \inflowClassOfInt{\interface_r}(r) = (1, 0) \land r \neq \freeListTail}$
    while (true) { $\annot{\boxed{\paren{\nodePred(\freeListTail, \interface_f) * \paren{\bracket{r \mapsto x, \_}_{\_} \land \nodePred(r, \interface_r)}} \sepincl \Phi} \land l \in \interface \land \nodelabelOfInt{\interface_r} = t \land \inflowClassOfInt{\interface_r}(r) = (1, 0)}$
      atomic {  // CAS(ft.fnext, null, r)
        if (ft.fnext == null) {
          $\annot{\paren{\paren{\bracket{\freeListTail \mapsto \_, \nullVal}_{\_} \land \nodePred(\freeListTail, \interface_f)} * \paren{\bracket{r \mapsto x, \_}_{\_} \land \nodePred(r, \interface_r)}} \sepincl \Phi \land l \in \interface \land \nodelabelOfInt{\interface_r} = t \land \inflowClassOfInt{\interface_r}(r) = (1, 0)}$
          ft.fnext := r; $\annot{\bracket{\nodePred(\freeListTail, \interface'_f) * \paren{\bracket{r \mapsto x, \_}_{\_} \land \nodePred(r, \interface'_r)}}_{\interface_1} \sepincl \Phi \land l \in \interface \land \interface_1 = \interface_f \intComp \interface_r \intEquiv \interface'_f \intComp \interface'_r \land \inflowClassOfInt{\interface'_r} = (1, 1)}$
          var I$'_1$ :| $\bracket{\graphPred(\interface'_1)}_{\interface_1} \land \set{\freeListTail \goesto \_, r \goesto \_} \in \inflowClassOfInt{\interface_1}$;
          sync(I$'_1$); $\annot{\paren{\nodePred(\freeListTail, \interface'_f) * \paren{\bracket{r \mapsto x, \_}_{\_} \land \nodePred(r, \interface'_r)}} \sepincl \Phi \land l \in \interface \land \inflowClassOfInt{\interface'_r} = (1, 1)}$
          var c := true;
        } else var c := false;
      }
      if (c) { $\annot{\boxed{\paren{\nodePred(\freeListTail, \interface_f) * \paren{\bracket{r \mapsto x, \_}_{\_} \land \nodePred(r, \interface_r)}} \sepincl \Phi} \land l \in \interface \land \inflowClassOfInt{\interface_r} = (1, 1)}$
        ft := r; $\annot{\boxed{\paren{\bracket{r \mapsto x, \_}_{\_} \land \nodePred(r, \interface_r)} \sepincl \Phi} \land l \in \interface}$
        break;
      }
    } $\annot{\boxed{\paren{\bracket{r \mapsto x, \_}_{\_} \land \nodePred(r, \interface_r)} \sepincl \Phi} \land l \in \interface}$
    atomic {  // CAS(l.next, r, x)
      if (l.next == r) { $\annot{\paren{\paren{\bracket{l \mapsto r, \_}_{\_} \land \nodePred(l, \interface_l)} * \paren{\bracket{r \mapsto x, \_}_{\_} \land \nodePred(r, \interface_r)}} \sepincl \Phi}$
        l.next := x; $\annot{\bracket{\nodePred(l, \interface'_l) * \nodePred(r, \interface'_r)}_{\interface_1} \sepincl \Phi \land \interface_1 = \interface_l \intComp \interface_r \intEquiv \interface'_l \intComp \interface'_r}$
        var I$'_1$ :| $\bracket{\graphPred(\interface'_1)}_{\interface_1} \land \set{l \goesto \_, r \goesto \_} \in \inflowClassOfInt{\interface_1}$;
        sync(I$'_1$);  $\annot{\paren{\nodePred(l, \interface'_l) * \nodePred(r, \interface'_r)} \sepincl \Phi}$
      }
    } $\annot{\boxed{\Phi}}$
  } else { $\annot{\boxed{\Phi}}$
    delete(); $\annot{\boxed{\Phi}}$
  }
}
\end{lstlisting}

%%% Local Variables: 
%%% mode: latex
%%% TeX-master: "main"
%%% End: 

\section{The B+ Tree}
\label{sec-btree}

\tikzset{%
  array/.style={matrix of nodes,nodes={draw, minimum size=5mm, anchor=center},column sep=-\pgflinewidth, row sep=-\pgflinewidth, nodes in empty cells,anchor=center},
  ptr/.style={*->, shorten <=-(1.8pt+1.4\pgflinewidth)},
  edge/.style={->},
  btnodeold/.style={circle, draw=black, thick, minimum size=5mm},
  btnode/.style={draw=black, rounded corners, minimum width=3.6cm},
  phantomNode/.style={circle, fill=none, inner sep=0pt, minimum size=0pt}
}

\begin{figure}[t]
  \begin{minipage}{1.0\linewidth}
    \centering
    \begin{tikzpicture}[>=stealth, font=\footnotesize, scale=0.8, every node/.style={scale=0.8}]
      % Grid
      % \draw[step=1cm,gray,very thin] (0, 0) grid (16, 8);

      % Nodes

      % r
      \matrix[array] (rkeys) {3 & & \\};
      \node (rran) at ($(rkeys.east) + (.6, 0)$) {$(-\infty, \infty)$};
      \matrix[array] (rptrs) at ($(rkeys.south) + (0, -.2)$) { & & & \\};
      \node[minimum size=5mm]  (rlen) at ($(rptrs.east) + (.4, 0)$) {$1$};
      \node[btnode, fit={(rkeys) (rptrs) (rlen) (rran)}] (r) {};

      % c
      \matrix[array, below=2 of rkeys] (ckeys) {5 & 7 & \\};
      \node (cran) at ($(ckeys.east) + (.6, 0)$) {$(3, \infty)$};
      \matrix[array] (cptrs) at ($(ckeys.south) + (0, -.2)$) { & & & \\};
      \node[minimum size=5mm]  (clen) at ($(cptrs.east) + (.4, 0)$) {$2$};
      \node[btnode, fit={(ckeys) (cptrs) (clen) (cran)}] (c) {};

      % n
      \matrix[array, below=2 of ckeys] (nkeys) {5 & & \\};
      \node (nran) at ($(nkeys.east) + (.6, 0)$) {$(5, 7)$};
      \matrix[array] (nptrs) at ($(nkeys.south) + (0, -.2)$) { & & & \\};
      \node[minimum size=5mm]  (nlen) at ($(nptrs.east) + (.4, 0)$) {$1$};
      \node[btnode, fit={(nkeys) (nptrs) (nlen) (nran)}] (n) {};

      % y0
      \matrix[array, left=2.5 of nkeys] (y0keys) {3 & 4 & \\};
      \node (y0ran) at ($(y0keys.east) + (.6, 0)$) {$(3, 5)$};
      \matrix[array] (y0ptrs) at ($(y0keys.south) + (0, -.2)$) { & & & \\};
      \node[minimum size=5mm]  (y0len) at ($(y0ptrs.east) + (.4, 0)$) {$2$};
      \node[btnode, fit={(y0keys) (y0ptrs) (y0len) (y0ran)}] (y0) {};

      % y2
      \matrix[array, right=2.5 of nkeys] (y2keys) {7 & 8 & \\};
      \node (y2ran) at ($(y2keys.east) + (.6, 0)$) {$(7, \infty)$};
      \matrix[array] (y2ptrs) at ($(y2keys.south) + (0, -.2)$) { & & & \\};
      \node[minimum size=5mm]  (y2len) at ($(y2ptrs.east) + (.4, 0)$) {$2$};
      \node[btnode, fit={(y2keys) (y2ptrs) (y2len) (y2ran)}] (y2) {};

      % cl
      \node[phantomNode] (cl) at (y0 |- c.north) {$\dotsc$};

      % Edges
      \draw[ptr] (rptrs-1-2.center) to (c);
      \draw[ptr] (rptrs-1-1.center) to (cl);
      \draw[ptr] (cptrs-1-1.center) to (y0.north);
      \draw[ptr] (cptrs-1-2.center) to (n);
      \draw[ptr] (cptrs-1-3.center) to (y2.north);
    \end{tikzpicture}
    \subcaption{}\label{fig-b-tree-a}
  \end{minipage}

    \vspace{1cm}

    \begin{minipage}{1.0\linewidth}
    \centering
    \begin{tikzpicture}[>=stealth, font=\footnotesize, scale=0.8, every node/.style={scale=0.8}]
      \tikzstyle{gnode}=[circle, draw=black, thick, minimum size=1cm]

      % Nodes
      \node[gnode] (r) {$\set{}$};
      \node[gnode, below= of r] (c) {$\set{}$};
      \node[gnode, below=1.2cm of c ] (n) {$\set{5}$};
      \node[gnode, left=1.3cm of n] (y0) {$\set{3, 4}$};
      \node[gnode, right=1.3cm of n] (y2) {$\set{7, 8}$};
      \node[phantomNode] (cl) at (y0 |- c.north) {$\dotsc$};

      % Edges
      \draw[edge] (r) to node[right] {$\set{3 \leq k < \infty}$} (c);
      \draw[edge] (r) to (cl);
      \draw[edge] (c) to node[above left] {$\set{3 \leq k < 5}$} (y0);
      \draw[edge] (c) to node[fill=white] {$\set{5 \leq k < 7}$} (n);
      \draw[edge] (c) to node[above right] {$\set{7 \leq k < \infty}$} (y2);
    \end{tikzpicture}
    \subcaption{}\label{fig-b-tree-b}
    \end{minipage}
  \caption{A B+ tree \subref{fig-b-tree-a} and its representation as a flow graph \subref{fig-b-tree-b}. The nodes in the B+ tree contain an array of keys on the top row, an array of pointers in the bottom row, the range in the top right, and $l$ (number of keys) in the bottom right. The flow graph edges are labelled with the edgesets and nodes with their contents. Lock information is omitted from this figure.}
  \label{fig-b-tree}
\end{figure}
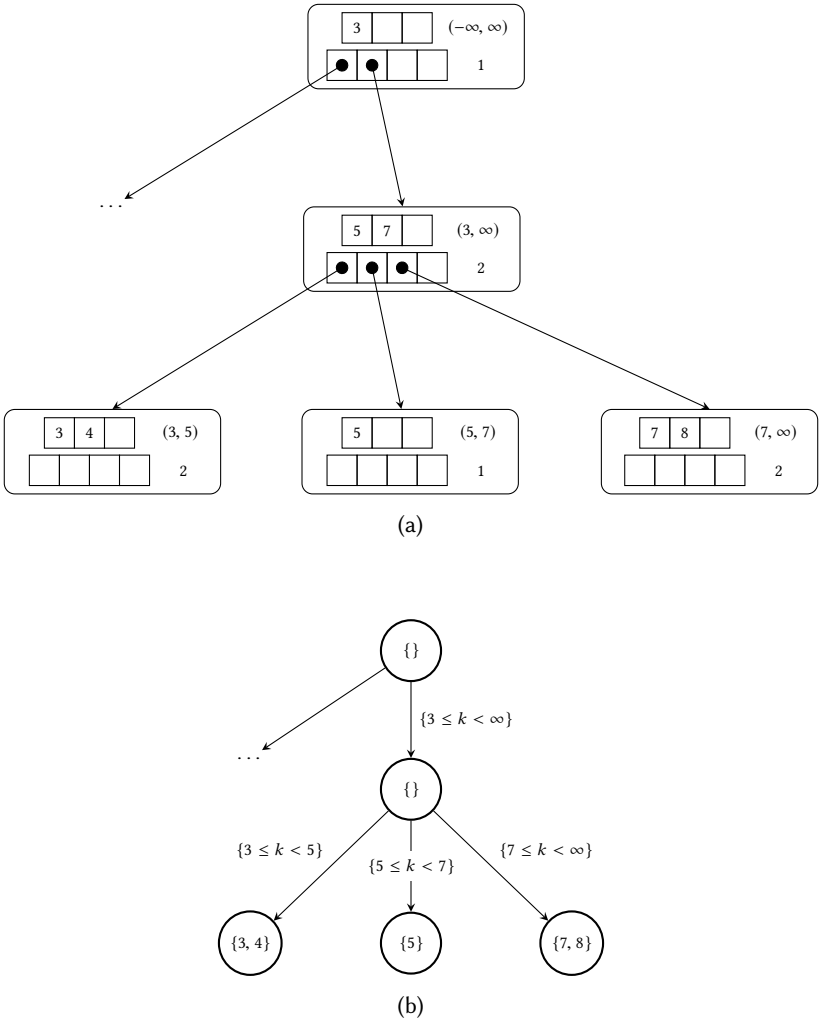

In this section we describe the B+ tree data structure and show how it is an instantiation of the give-up template, hence ensuring its memory safety and linearizability.

A B+ tree is a generalization of a binary search tree that implements a dictionary data structure.
In a binary search tree (BST), each node contains a value $k$ and up to two pointers $y_l$ and $y_r$.
An operation takes the left branch if its query key is less than $k$ and the right branch otherwise.
A B+ tree generalizes this by having $l$ values $k_0, \dots, k_{l-1}$
and $l+1$ pointers $y_0, \dots, y_{l}$ at each node, such that $l+1$
is between $B$ and $2B$ for some constant $B$.
At internal nodes, an operation takes the branch $y_i$ if its query key is between $k_{i-1}$ and $k_{i}$.
From the perspective of the edgeset framework, the edgeset of an edge to $y_i$ is $\setcomp{k}{k_{i-1} \leq k < k_{i}}$.
We store the inset of each node in a range field, shown on the top right of each node as two values $(r_0, r_1)$ that denote the set $\setcomp{k}{r_{0} \leq k < r_{1}}$.
The values stored at a leaf node are keys and make up the contents of the B+ tree.
Nodes also contain $l$, the number of values in a node, and a lock field to store the thread ID of the thread that holds a lock on the node.
\rF{fig-b-tree-a} shows an example of a B+ tree with $B = 2$, and the corresponding flow graph is shown in \rF{fig-b-tree-b}.

In our code, we assume that a B+ tree node is implemented by the following struct definition:
\begin{lstlisting}
struct Node {
  Int lock, len;
  Key range[2];
  Key keys[2B];
  Node ptrs[2B];
}
\end{lstlisting}
This node structure is described in our logic by the following instantiations of the $\gamma_g$ and $\gamma_b$ predicates:
\begin{align*}
  \gamma_g(x, \inflow, C, t, \flow) & \defeq \exists l, r_0, r_1, \vec{k}, \vec{y}.\; x \mapsto t, l, r_0, r_1, k_0, \dots, k_{2B-1}, y_0, \dots, y_{2B-1} \\
  & \quad \land 0 \leq l < 2B \land \inflow(x) = \setcomp{k}{r_0 \leq k < r_1} \land r_0 \leq k_0 \land k_{l-1} < r_1 \\
  & \quad \land C = \ite(y_0 \neq \nullVal, \emptyset, \setcomp{k_i}{0 \leq i < l}) \\
  & \quad \land \flow = \setcomp{(x, y_i) \goesto \setcomp{k}{\ite(i \leq 0, -\infty, k_{i-1}) \leq k < \ite(i \geq l, \infty, k_{i})}}{y_i \neq \nullVal} \\
  & \quad \land \forall 0 < i < 2B.\; (i < l \impl k_{i-1} < k_i) \land (i \geq l \impl k_i = \nullVal) \\
  & \quad \land \exists l'.\;  (l' = 0 \lor l' = l + 1) \land \forall i < 2B.\; (l' \leq i \iff y_i = \nullVal) \\
  \\
  \gamma_b(x, t) &\defeq x \mapsto t, \_, \dots, \_
\end{align*}
We also re-write the global data structure invariant as follows, to simplify notation in our proofs:
\[\Phi \defeq \exists \interface.\; \graphPred(\interface) \land \varphi(\interface) \qquad \varphi(\interface) \defeq \exists \inflow \in \inflowClassOfInt{\interface}.\; \inflow \intPlus \zerofun = \set{r \goesto \KS} \intPlus \zerofun \land \flowmapOfInt{\interface} = \emptyFn\]

We now show the instantations of the helper functions for the give-up template algorithm, along with the implementation of the delete operation (the insert and member operations are similar).
The code is annotated with intermediate assertions that all imply the global data structure invariant $\Phi$.
We assume that $t$ is the thread ID of the current thread.
We have omitted the stability proofs, but these are syntactic and are similar to the example stability proof shown in \rSc{sec-stability-proof}.

\begin{lstlisting}
procedure lock(Node c) { $\annot{\boxed{\nodePred(c, \interface_c) \sepincl \Phi}}$
  while (true) { $\annot{\boxed{\nodePred(c, \interface_c) \sepincl \Phi}}$
    atomic {  // Action: @(\ref{eqn-action-lock})@
      if (c.lock == 0) { $\annot{\paren{\bracket{c \mapsto 0, \_, \dotsc, \_}_{\_} \land \nodePred(c, \interface_c)}  * \graphPred(\interface_2) \land \interface = \interface_c \intComp \interface_2 \land \varphi(\interface)}$
        c.lock := t; $\annot{\bracket{\nodePred(c, \interface'_c)}_{\interface_c} * \graphPred(\interface_2) \land \interface = \interface_c \intComp \interface_2 \land \varphi(\interface) \land \interface_c \intEquiv \interface'_c \land \nodelabelOfInt{\interface_c} = (C, 0) \land \nodelabelOfInt{\interface'_c} = (C, t)}$
        var I$'_c$ :| $\bracket{\graphPred(\interface'_c)}_{\interface_c} \land \set{c \goesto \_} \in \inflowClassOfInt{\interface_c}$;
        sync(I$'_c$); $\annot{\nodePred(c, \interface'_c) * \graphPred(\interface_2) \land \interface' = \interface'_c \intComp \interface_2 \land \varphi(\interface') \land \nodelabelOfInt{\interface'_c} = (C, t)}$
        break;
      }
    }
  } $\annot{\boxed{\nodePred(c, \interface'_c) \sepincl \Phi} \land \nodelabelOfInt{\interface'_c} = (\_, \set{t})}$
}

procedure unlock(Node c) { $\annot{\boxed{\nodePred(c, \interface_c) \sepincl \Phi} \land \nodelabelOfInt{\interface_c} = (\_, \set{t})}$
  atomic {  // Action: @(\ref{eqn-action-sync})@
    $\annot{\paren{\bracket{c \mapsto t, \_, \dotsc, \_}_{\_} \land \nodePred(c, \interface_c)} * \graphPred(\interface_2) \land \interface = \interface_c \intComp \interface_2 \land \varphi(\interface)}$
    c.lock := 0; $\annot{\bracket{\nodePred(c, \interface'_c)}_{\interface_c} * \graphPred(\interface_2) \land \interface = \interface_c \intComp \interface_2 \land \varphi(\interface) \land \interface_c \intEquiv \interface'_c \land \nodelabelOfInt{\interface_c} = (C, t) \land \nodelabelOfInt{\interface'_c} = (C, 0)}$
    var I$'_c$ :| $\bracket{\graphPred(\interface'_c)}_{\interface_c} \land \set{c \goesto \_} \in \inflowClassOfInt{\interface_c}$;
    sync(I$'_c$); $\annot{\nodePred(c, \interface'_c) * \graphPred(\interface_2) \land \interface' = \interface'_c \intComp \interface_2 \land \varphi(\interface') \land \nodelabelOfInt{\interface'_c} = (C, 0)}$
  } $\annot{\boxed{\nodePred(c, \interface'_c) \sepincl \Phi} \land \nodelabelOfInt{\interface'_c} \not\nodelabelLeq (\KS, \set{t})}$
}

procedure inRange(Node c, Key k) { $\annot{\boxed{\nodePred(c, \interface_c) \sepincl \Phi} \land \nodelabelOfInt{\interface_c} = (\_, \set{t})}$
  $\annot{\boxed{\bracket{\gamma_g(c, \_, \_, t, \_)}_{\interface_c} \sepincl \Phi} \land \nodelabelOfInt{\interface_c} = (\_, \set{t})}$
  if (c.range[0] <= k && k < c.range[1]) {
    return true; $\annot{\boxed{\nodePred(c, \interface_c) \sepincl \Phi} \land \nodelabelOfInt{\interface_c} = (\_, \set{t}) \land k \in \inflowClassOfInt{\interface_c}(c)}$
  }
  return false; $\annot{\boxed{\nodePred(c, \interface_c) \sepincl \Phi} \land \nodelabelOfInt{\interface_c} = (\_, \set{t}) \land k \not\in \inflowClassOfInt{\interface_c}(c)}$
}

procedure findNext(Node c, Key k) { $\annot{\boxed{\nodePred(c, \interface_c) \sepincl \Phi} \land \nodelabelOfInt{\interface_c} = (\_, \set{t})}$
  $\annot{\boxed{\bracket{c \mapsto t, l, \_, \_, \vec{k}, \vec{y}}_{\interface_c} \sepincl \Phi} \land \nodelabelOfInt{\interface_c} = (\_, \set{t})}$
  var i := 0;
  while (i < c.len && k >= c.keys[i]) {
    i := i + 1;
  } $\annot{\boxed{\bracket{c \mapsto t, l, \_, \_, \vec{k}, \vec{y}}_{\interface_c} \sepincl \Phi} \land \nodelabelOfInt{\interface_c} = (\_, \set{t}) \land \ite(i \leq 0, -\infty, k_{i-1}) \leq k < \ite(i \geq l, \infty, k_i) }$
  if (i == c.len) {
    return null; $\annot{\boxed{\nodePred(c, \interface_c) \sepincl \Phi} \land \nodelabelOfInt{\interface_c} = (\_, \set{t}) \land n = \nullVal \land \forall x \in \flowmapOfInt{\interface_c}.\; k \not\in \flowmapOfInt{\interface_c}(c, x)}$
  } else {
    return c.ptrs[i]; $\annot{\boxed{\nodePred(c, \interface_c) \sepincl \Phi} \land \nodelabelOfInt{\interface_c} = (\_, \set{t}) \land n \neq \nullVal \land k \in \flowmapOfInt{\interface_c}(c, n)}$
  }
}

procedure delete(Node c, Key k) { $\annot{\boxed{\nodePred(c, \interface_c) \sepincl \Phi} \land \nodelabelOfInt{\interface_c} = (C, \set{t}) \land k \in \inflowClassOfInt{\interface_c}(c) \land \forall x \in \flowmapOfInt{\interface_c}.\; k \not\in \flowmapOfInt{\interface_c}(c, x)}$
  $\annot{\boxed{\nodePred(c, \interface_c) \sepincl \Phi} \land k \in \inflowClassOfInt{\interface_c}(c) \land \flowmapOfInt{\interface_c} = \emptyFn}$  // Because union of outgoing edgesets is $\KS$
  var i := 0;
  while (i < c.len && k >= c.keys[i]) {
    i := i + 1;
  } $\annot{\boxed{\bracket{c \mapsto \set{t, l, \_, \_, \vec{k}, \_}}_{\interface_c} \sepincl \Phi} \land k \in \inflowClassOfInt{\interface_c}(c) \land \flowmapOfInt{\interface_c} = \emptyFn \land 0 \leq i \leq l \land (i = l \lor k < k_i)}$
  if (i == c.len || c.keys[i] != k) {
    return false; $\annot{\boxed{\nodePred(c, \interface_c) \sepincl \Phi} \land \nodelabelOfInt{\interface_c} = (C, \set{t}) \land k \not\in C}$
  } else { $\annot{\boxed{\bracket{c \mapsto \set{t, l, \_, \_, \vec{k}, \_}}_{\interface_c} \sepincl \Phi} \land 0 \leq i < l \land k = k_i}$
    // First change node label to $\overline{t}$ to allow breaking good condition
    var I$'_c$ :| $\graphPred(\interface_c) \land \set{c \goesto \_} \in \inflowClassOfInt{\interface_c} \land \interface_c \intEquiv \interface'_c \land \nodelabelOfInt{\interface_c} = (C, \set{t}) \land \nodelabelOfInt{\interface'_c} = (C, \{\overline{t}\})$;
    sync(I$'_c$);  // Action: @(\ref{eqn-action-sync})@
    $\annot{\boxed{\bracket{c \mapsto \set{t, l, \_, \_, \vec{k}, \_}}_{\interface'_c} \sepincl \Phi} \land 0 \leq i < l \land k = k_i \land \nodelabelOfInt{\interface'_c} = (\_, \{\overline{t}\})}$
    while (i < c.len - 1) {  // All these actions are also @(\ref{eqn-action-sync})@s, with unchanged interfaces
      c.keys[i] := c.keys[i + 1];
      i := i + 1;
    }
    c.keys[i] := null;
    $\annot{\boxed{\paren{\bracket{\nodePred(c, \interface''_c)}_{\interface'_c} \land \nodePred(c, \interface'_c)} \sepincl \Phi} \land \interface'_c \intEquiv \interface''_c \land \nodelabelOfInt{\interface'_c} = (C, \{\overline{t}\}) \land \nodelabelOfInt{\interface''_c} = (C', \set{t}) \land C' = C \setminus \set{k}}$
    var I$''_c$ :|  $\bracket{\graphPred(\interface''_c)}_{\interface_c} \land \set{c \goesto \_} \in \inflowClassOfInt{\interface_c}$;
    sync(I$''_c$);  // Action: @(\ref{eqn-action-sync})@
    $\annot{\boxed{\nodePred(c, \interface''_c) \sepincl \Phi} \land \nodelabelOfInt{\interface''_c} = (C', \set{t}) \land k \not\in C'}$
    return true;
  }
}
\end{lstlisting}

%%% Local Variables:
%%% mode: latex
%%% TeX-master: "main"
%%% End:

\section{Example Stability Proof}
\label{sec-stability-proof}

In $\rgsep$, to show that an assertion $S$ is stable under an action $P \rightsquigarrow Q$, we need to show that $(P \septract S) * Q \impl S$, where $\septract$ is the \emph{septraction} operator.
Septraction is the dual operator to magic wand ($P \septract Q \iff \neg (P \magicwand \neg Q)$), and intuitively $(P \septract S)$ is the state obtained when we remove a state satisfying $P$ from a state satisfying $S$.

As an example, we now prove that the intermediate assertion at line~\ref{give-up-inRange} of \rF{fig-give-up-template} is stable under the action (\ref{eqn-action-sync}).
This shows that if a thread has a reference to a node $c$ that it has locked and it knows that $k$ is in the inset of $c$, then this knowledge is preserved after any combination of interference by other threads.
Moreover, the global invariant of the data structure described by the good node condition is maintained.
Note that this proof does not depend on the concrete good condition predicates, and so holds for any concrete implementation.
Proving stability of all intermediate assertions is similarly easy, and furthermore the proofs are purely syntactic, and use generic lemmas that are not specific to the inset flow.

\begin{lemma}
  The formula $\nodePred(c, \interface_c) \sepincl \Phi \land \nodelabelOfInt{\interface_c} = (\_, \set{t}) \land k \in \inflowClassOfInt{\interface_c}(c)$ is stable under action (\ref{eqn-action-sync}).
\end{lemma}
\begin{proof}
  
\begin{enumerate}
\item We start with
  \[ \paren{\paren{\graphPred(\interface_1) \land \nodelabelOfInt{\interface_1} \nodelabelLeq (\_, \{t', \overline{t'}\})} \septract \paren{\nodePred(c, \interface_c) \sepincl \Phi \land \nodelabelOfInt{\interface_c} = (\_, \set{t}) \land k \in \inflowClassOfInt{\interface_c}(c)}} \]
  \[ {} * \graphPred(\interface'_1)  \land \nodelabelOfInt{\interface'_1} \nodelabelLeq (\_, \set{0, t'}) \land \interface_1 \intLessEquiv \interface'_1\]
\item By expanding $\Phi$, the spatial part of this formula is
  \[ \paren{\graphPred(\interface_1) \septract \paren{\nodePred(c, \interface_c) \sepincl \graphPred(\interface)}} * \graphPred(\interface'_1) \]
\item $\nodePred(c, \interface_c) \sepincl \graphPred(\interface) \equiv (\nodePred(c, \interface_c) * \true) \land \graphPred(\interface)$, so we use (\ref{rule-gr-decomp}) to get
  \[ \paren{\graphPred(\interface_1) \septract \paren{\paren{\nodePred(c, \interface_c) * \graphPred(\interface_2)} \land \graphPred(\interface)}} * \graphPred(\interface'_1) \]
\item We use the SL identity $P \septract Q \equiv P \septract (Q \land (P * \true))$ to get
  \[ \paren{\graphPred(\interface_1) \septract \paren{\paren{\graphPred(\interface_1) * \true} \land \paren{\nodePred(c, \interface_c) * \graphPred(\interface_2)} \land \graphPred(\interface)}} * \graphPred(\interface'_1) \]
\item By (\ref{rule-gr-decomp})
  \[ \paren{\graphPred(\interface_1) \septract \paren{\paren{\graphPred(\interface_1) * \graphPred(\_)} \land \paren{\nodePred(c, \interface_c) * \graphPred(\interface_2)} \land \graphPred(\interface)}} * \graphPred(\interface'_1) \]
\item As $\nodelabelOfInt{\interface_c} = (\_, \set{t}) \land \nodelabelOfInt{\interface_1} \nodelabelLeq (\_, \{t', \overline{t'}\}) \land t \neq t' \models \nodelabelOfInt{\interface_c} \not\nodelabelLeq \nodelabelOfInt{\interface_1}$, use (\ref{rule-disj})
  \[ \paren{\graphPred(\interface_1) \septract \paren{\paren{\graphPred(\interface_1) * \nodePred(c, \interface_c) * \graphPred(\interface_3)} \land \graphPred(\interface)}} * \graphPred(\interface'_1) \]
\item \label{step-new-int} Using (\ref{rule-gr-comp}) and (\ref{rule-uniq}) we get $\interface = \interface_1 \intComp \interface_c \intComp \interface_3$.
\item \label{step-new-state} By $P \septract (P * Q) \impl Q$ for precise $P$,
  \[ \nodePred(c, \interface_c) * \graphPred(\interface_3) * \graphPred(\interface'_1) \]
\item As $\interface_1 \intLessEquiv \interface'_1$, we use (\ref{rule-repl}) to get $\interface = \interface_1 \intComp \interface_c \intComp \interface_3 \intLessEquiv \interface'_1 \intComp \interface_c \intComp \interface_3 = \interface'$.
\item The pure parts of $\Phi$ were: $\exists \inflow \in \inflowClassOfInt{\interface}.\; \inflow \intPlus \zerofun = \set{r \goesto \KS} \intPlus \zerofun \land \flowmapOfInt{\interface} = \emptyFn$.
  By (\ref{rule-repl-in}) and (\ref{rule-repl-f}) on $\interface \intLessEquiv \interface'$, we get $\inflow \in \inflowClassOfInt{\interface'} \land \flowmapOfInt{\interface'} = \emptyFn$.
\item Now use (\ref{rule-gr-comp}) on (\ref{step-new-state}), and by (\ref{step-new-int}) we have
  \[ \paren{\nodePred(c, \interface_c) * \graphPred(\interface_3) * \graphPred(\interface'_1)} \land \graphPred(\interface') \]
\item This implies
  \[ \nodePred(c, \interface_c) \sepincl \Phi \land \nodelabelOfInt{\interface_c} = (\_, \set{t}) \land k \in \inflowClassOfInt{\interface_c}(c) \]
\end{enumerate}
\end{proof}

%%% Local Variables:
%%% mode: latex
%%% TeX-master: "main"
%%% End:

% \input{proofs}
}

\end{document}